\documentclass[12pt]{article}
\usepackage[margin=1.0in]{geometry}
\usepackage[utf8]{inputenc}
\usepackage[intlimits]{amsmath}
\usepackage{amssymb,amsthm,dsfont}
\usepackage{bbold}
\usepackage{tikz}
\usepackage{physics}
\usepackage[english]{babel}

\newtheorem*{proposition}{Proposition}

\renewcommand{\d}{\mbox{d}}
\renewcommand{\bar}[1]{\overline{#1}}

\newcommand{\rhs}{{\it r.h.s.} }

\numberwithin{equation}{section}
\newcommand{\ie}{{\it i.e.,} }

\begin{document}

	\vspace{1.7 cm}
	
	\begin{flushright}
		
		{\small FIAN/TD/2-2025}
	\end{flushright}
	\vspace{1.7 cm}
	
	\begin{center}
		{\large\bf Linearized Coxeter Higher-Spin Theories}
		
		\vspace{1 cm}
		
		{\bf A.A.~Tarusov$^{1}$, K.A.~Ushakov$^{1}$ and  M.A.~Vasiliev$^{1,2}$}\\
		\vspace{0.5 cm}
		\textbf{}\textbf{}\\
		\vspace{0.5cm}
		\textit{${}^1$ I.E. Tamm Department of Theoretical Physics,
			Lebedev Physical Institute,}\\
		\textit{ Leninsky prospect 53, 119991, Moscow, Russia}\\
		
		\vspace{0.7 cm} \textit{
			${}^2$ Moscow Institute of Physics and Technology,\\
			Institutsky lane 9, 141700, Dolgoprudny, Moscow region, Russia
		}
		
	\end{center}
	
	\vspace{0.4 cm}
	
\begin{abstract}
\noindent
A class of higher-spin gauge theories on $AdS_4$ associated with various Coxeter groups  $\mathcal{C}$ is analyzed at the linear order. For a general $\mathcal{C}$, a solution corresponding to the $AdS_4$ space and the form of the free unfolded equations are established. A disentanglement criterion has been formulated for Coxeter HS modules. The shifted homotopy technique is uplifted to the general Coxeter HS models. In case of the Coxeter group $B_2$ classification of unitary HS modules and a consistent truncation to them are determined, the dynamical content is discussed briefly.

\end{abstract}
	
\newpage
	
\vspace{-1cm}
\tableofcontents
	
\newpage

\section{Introduction}
	
Higher-spin (HS) gauge theories describe interactions of massless fields of all spins. The first example of a nonlinear HS theory was given for the $4d$ case in \cite{Vasiliev:1990en}, while its modern formulation was presented in \cite{Vasiliev:1992av}. A unique feature of HS gauge theories is that consistent interactions of propagating massless fields exist in a curved background, providing a length scale in HS interactions that contain higher derivatives. $AdS$ is the most symmetric curved background compatible with HS interactions \cite{Fradkin:1986qy, Vasiliev:1988sa}. The lowest dimension where the HS massless fields propagate is $d=4$
with $AdS_4$ as the most symmetric vacuum.

One of the fundamental questions in HS theory concerns the construction of more general HS models that could be related to String Theory. Arguments that String Theory possesses higher symmetries in the high-energy limit were given long ago in \cite{Gross:1987ar, Gross:1988ue}.  Although the conjecture that HS theory is related to String theory is supported by the analysis of the high-energy limit of string amplitudes \cite{Gross:1988ue} and passed some non-trivial tests \cite{Bianchi:2003wx}-\cite{Bianchi:2004npm}, no satisfactory understanding of this relation beyond the free field sector of the tensionless limit of String Theory \cite{Lindstrom:2003mg, Bonelli:2003kh, Sagnotti:2003qa} is available.

A potential candidate for a suitable extended HS model was proposed in \cite{Vasiliev:2018zer}, where a new class of higher-spin gauge theories associated with various Coxeter groups was constructed (see \cite{Bourbaki} for a detailed explanation of Coxeter groups). These extended models are based on deformed oscillator algebras, known as Cherednik algebras \cite{Cherednik:1992sy}. HS-like models of this class could have been formulated long ago, since the relevance of the Cherednik algebra to HS theory was mentioned in \cite{Brink:1993sz}. However, a naive extension of this class was not formulated because of the problem with the resulting spectrum of states. There is no room left for a massless spin-two state, \ie graviton, not allowing the description of the HS gravity. Fortunately, an extension of the standard Cherednik algebras by a set of idempotents, known as framed Cherednik algebras \cite{Vasiliev:2018zer}, allowed one to bypass the problem of missing massless HS fields in the spectrum.

It was conjectured in \cite{Vasiliev:2018zer} that a multiparticle extension of the
HS theory, \ie transition to the theory built upon a universal enveloping algebra of the HS algebra (see \cite{Vasiliev:2012tv} for the multiparticle extension), based on the Coxeter group $B_2$ has a rich enough symmetry and spectrum to match with String Theory.\footnote{If we denote a star-product algebra as $A$, then multiparticle algebra $M(A)$ is isomorphic to $U(\operatorname{Lie}(A))$, where $\operatorname{Lie(\bullet)}$ constructs a Lie algebra out of an associative one. As a vector space $M(A)$ is the direct sum of all graded-symmetric tensor degrees of $A$: $M(A) \simeq \overset{\infty}{\underset{i=1}{\oplus}}\operatorname{Sym}A^{\otimes i}$. Thus, $M(A)$ acts on the space of all multiparticle states.} (Note that the construction of multiparticle HS theory is somewhat analogous to the idea of singleton string whose spectrum is represented by multi-singletons  \cite{Engquist:2005yt, Engquist:2007pr}.)
This conjecture was based on several grounds.
One was that this model has two independent coupling constants associated with the two conjugacy classes of $B_2$. These were conjectured to be associated with the HS coupling constant and String coupling of the model. A related fact is that a multiparticle $B_2$ HS model has a room for the fields to be associated with multitrace operators in the holographic picture. Another motivation was due to
the observation of a doubled HS algebra with non-trivial mixing in the context of String theory on a special background \cite{Gaberdiel:2015mra, Gaberdiel:2015wpo}.
From that perspective, $B_2$ Coxeter model is the simplest non-trivial extension, possessing two copies of the HS algebra associated with the pair of orthogonal vectors belonging to the root system of $B_2$, which are mixed non-trivially by the added permutations. This is to be  contrasted to the $A_2$ group, in which orthogonal root vectors do not exist, that does not allow independent copies.
A related fact is that $A_2$ has a single conjugacy class. On the other hand, the models based on the Coxeter groups of higher ranks were argued in \cite{Vasiliev:2018zer} to be associated with much richer  tensor extensions of String Theory.

\begin{tikzpicture}[scale=1.5]
 \begin{scope}[xshift=-3cm]
    \foreach \angle in {0, 60, 120, 180, 240, 300} {
      \draw[->, thick] (0,0) -- (\angle:2cm);
    }
    \node[below] at (0,-2.2) {Root system \( A_2 \)};
  \end{scope}

\begin{scope}[xshift=3cm]
    \foreach \angle in {0, 90, 180, 270} {
      \draw[->, red, thick] (0,0) -- (\angle:1.414cm); 
    }
    \foreach \angle in {45, 135, 225, 315} {
      \draw[->, blue, thick] (0,0) -- (\angle:2cm);
    }
    \node[below] at (0,-2.2) {Root system \( B_2 \)};
  \end{scope}

\end{tikzpicture}

To prove the conjectured link between the $B_2$ Coxeter extended multiparticle model with String Theory one has to spontaneously break the extended HS symmetry to the space-time symmetry and compare the resulting massive spectrum to the string one. Since this procedure requires knowledge of the theory at the linear level, analysis of the linearized multiparticle Coxeter HS (CHS) models should be performed.

In this paper, we consider linearization of a general CHS theory, determine the $AdS_4$ background solution and extract the form of the First On-Shell Theorem (\ie the linearized unfolded field equations) for a general Coxeter group. A new type of HS modules that are not equivalent to the tensor products of twisted-adjoint and adjoint modules of standard $4d$ HS theory has been found. We propose a criterion for the disentanglement of a module in the case of a general group $\mathcal{C}$, \ie necessary and sufficient conditions for a CHS module to be a tensor product of adjoint and twisted-adjoint modules of a standard HS theory. Moreover, we classify all unitary modules in the $B_2$ model, provide a consistent truncation of CHS modules in a zero-form sector to unitary submodules and briefly discuss the dynamical content in the $B_2$ case. It is argued that the dynamical fields consist of copies of fields $C$ corresponding to standard generalized Weyl tensors and $\omega$ corresponding to Fronsdal fields and their combinations. In addition, the shifted homotopy technique \cite{Didenko:2018fgx} is extended to CHS models.

The paper is organized as follows. We start with recalling the construction of CHS models, proposed in \cite{Vasiliev:2018zer}, in Section \ref{Coxeter HS models}. In Section \ref{AdS solution} we obtain the embedding of $AdS_4$ in $4d$ general CHS model and show its uniqueness. Then in Section \ref{Derivative and modules} we consider CHS modules and propose the disentanglement criterion in the case of a general Coxeter group. In the $B_2$ model we present a realization of the CHS linear equations in terms of the field-theoretical Fock modules and classify CHS modules according to the unitarity/non-unitarity through the identification with $su(2,2)$ modules induced by a Bogolyubov transform. Generalization of the First On-Shell Theorem for general CHS theory and modified shifted homotopy technique are derived in Section \ref{FOST}. In Section \ref{dynamics} we discuss the dynamical content of the $B_2$ theory. Our conclusions are in Section \ref{Conclusion}.

\section{Coxeter higher-spin models}\label{Coxeter HS models}

\subsection{Coxeter groups and framed Cherednik algebra}

Following \cite{Vasiliev:2018zer}, we start with the definition of a Coxeter group. A rank-$p$ Coxeter group $\mathcal{C}$ is generated by reflections with respect to a system of root vectors $v_a$ in a $p$-dimensional Euclidean vector space $V$ with the scalar product $(x, y) \in \mathds{R}$, $x, y \in V$. An elementary reflection associated with the root vector $v_a$ acts on $x \in V$ as follows
\begin{equation}
    R_{v_a}x^i = x^i - 2 \frac{(v_a,x)}{(v_a,v_a)}v_a^i\,, \quad R_{v_a}^2 = Id\,.
\end{equation}

In the sequel, we will be mainly concerned with the groups $A_p$ and $B_p$. The root system of $A_p$ consists of the vectors $v^{ij} = e^i - e^j$, where $e^i$ form an orthonormal frame in $\mathds{R}^{p+1}$. $V$ is the $p$-dimensional subspace of relative coordinates in $\mathds{R}^{p+1}$ spanned by $v^{ij}$. The root system of $B_p$ consists of two conjugacy classes under the action of $B_p$
\begin{equation}
    \mathcal{R}_1 = \{\pm e^i, 1 \leq i \leq p\}\,, \quad \mathcal{R}_2 = \{\pm e^i \pm e^j, 1 \leq i < j \leq p\}\,.
\end{equation}
In addition to permutations, $B_p$ contains reflections of any basis axis in $V = \mathds{R}^p$ generated by $v^i_\pm = \pm e^i$ \cite{Bourbaki}.

We introduce a set of idempotents $I_n$, a set of oscillators $q_\alpha^n$ and dressed Klein operators $\hat{K}_v$ for each root vector $v$ (here $\alpha \in \{1,2\}, n\in\{1,...,p\}$) that obey
\begin{equation}\label{e: qIK1}
     I_n I_m = I_m I_n\,, \quad I_n I_n = I_n\,, \quad I_n q_\alpha^n = q_\alpha^n I_n = q_\alpha^n\,, \quad I_m q_\alpha^n = q_\alpha^n I_m\,,
\end{equation}
with no summation over repeated Latin indices, and
\begin{equation}\label{e: qIK2}
    \hat{K}_v q_\alpha^n = R_v{}^n{}_m q_\alpha^m \hat{K}_v\,, \quad \hat{K}_v \hat{K}_u = \hat{K}_u \hat{K}_{R_u(v)} = \hat{K}_{R_v(u)}\hat{K}_v\,, \quad \hat{K}_v \hat{K}_v = \prod I_{i_1(v)}...I_{i_k(v)}\,, \quad \hat{K}_v = \hat{K}_{-v}\,,
\end{equation}
\begin{equation}\label{e:cherednik comm}
    [q_\alpha^n,q_\beta^m] = -i \varepsilon_{\alpha\beta}\bigg(2\delta^{nm}I_n + \sum_{v \in \mathcal{R}} \nu(v) \frac{v^n v^m}{(v,v)} \hat{K}_v\bigg)\,,
\end{equation}
where $\mathcal{R}$ is a set of conjugacy classes of root vectors under the action of  $\mathcal{C}$, $\nu(v)$ is a function of the conjugacy classes, and the labels $i_1(v), . . . , i_k(v)$ enumerate those idempotents $I_n$ that carry labels affected by the reflection $R_v$. For instance, in the case of $B_p$ there are two types of dressed Klein operators: $\hat K_{ij}$ corresponding to the root vector $v^{ij}$ and $\hat K_i$ corresponding to the vector $e^i$.
As a consequence, the dressed Klein operators can be naively related to the regular ones as
\begin{equation}
    \hat{K}_v = K_v \prod I_{i_1(v)}...I_{i_k(v)}\,.
\end{equation}
Dressed Klein operators $\hat{K}_v$ are demanded to obey
\begin{equation}\label{e:KI comm}
    I_n \hat{K}_v = \hat{K}_v I_n\,, \forall n  \in \{1,...,p\}\,,
\end{equation}
\begin{equation}\label{e:KI absorb}
    I_{n} \hat{K}_v = \hat{K}_v I_n= \hat{K}_v \,, \forall n  \in \{i_1(v), . . . , i_k(v)\}\,.
\end{equation}
It should be stressed that the unhatted Klein operators do not appear in the construction of the framed Cherednik algebra. One can check that the double commutator of $q_\alpha^n$ satisfies Jacobi identity which is the most fundamental property of the Cherednik algebra. Indeed, the non-zero part of the triple commutator of $q_n^\alpha, q_m^\beta, q_k^\gamma$ is proportional to $v_n v_m v_k$ and hence contains the total antisymmetrization over three two-component indices $\alpha, \beta, \gamma$ giving zero.

For any Coxeter root system the generators
\begin{equation}
    t_{\alpha\beta} = \frac{i}{4}\sum_{n=1}^{p}\{q^n_\alpha, q^n_\beta\}I_n\,
\end{equation}
obey the $sp(2)$ commutation relations
\begin{equation}
\left[t_{\alpha \beta}, t_{\gamma \delta}\right]=\epsilon_{\beta \gamma} t_{\alpha \delta}+\epsilon_{\beta \delta} t_{\alpha \gamma}+\epsilon_{\alpha \gamma} t_{\beta \delta}+\epsilon_{\alpha \delta} t_{\beta \gamma}\,,
\end{equation}
properly rotating all Greek indices,
\begin{equation}\label{e: sl2 invariance}
    [t_{\alpha\beta}, q^n_\gamma] = \epsilon_{\beta\gamma} q^n_\alpha + \epsilon_{\alpha\gamma} q^n_\beta\,.
\end{equation}

The main feature of the framed Cherednik algebra compared to the standard one is the presence of idempotents $I_n$ which "split" the identity operator and induce filtration of the algebra. This extension makes it possible to resolve the long-standing problem of rising vacuum energy with an increase in the number of oscillator copies (see \cite{Vasiliev:2018zer} for details). Note that usual Cherednik algebra results from the framed one by quotioning out the ideal identifying all $I_n$ with the unit element of the algebra.

\subsection{Coxeter higher-spin equations}\label{Coxeter higher-spin equations}

Consider $x$-dependent fields $W$, $S$ and $B$ which also depend on $p$ sets of variables enumerated by the label $n \in \{1,...,p\}$, that include $Y^n_A , Z^n_A$ $(A \in \{1,...,4\})$, idempotents $I_n$, anticommuting differentials $dZ^A_n$ and dressed Klein operators $\hat{K}_v$ associated with all root vectors of a chosen Coxeter group $\mathcal{C}$ (at the convention $\hat K_{-v} = \hat K_v)$. The field $W(Y,Z,I;\hat{K}|x)$ is a $dx$ one-form,  $S(Y,Z,I;\hat{K}|x)$ is a $dZ$ one-form and  $B(Y,Z,I;\hat{K}|x)$ is a zero-form. The field equations associated with the framed Cherednik algebra (\ref{e:cherednik comm}) are formulated in terms of the star product analogous to the standard HS one of \cite{Vasiliev:1992av}
\begin{equation}\label{e:star product}
    (f*g)(Y,Z,I) = \frac{1}{(2\pi)^{4 p}}\int d^{4p}S d^{4p}T \exp\bigg(iS^A_n T^B_m C_{AB}\delta^{nm}\bigg) f(Y_i+I_i S_i,Z_i+I_i S_i,I)g(Y+T,Z-T,I)\,,
\end{equation}
where
\begin{equation}
    C_{AB} = \begin{pmatrix}
\varepsilon_{\alpha\beta} & 0 \\
0 & \bar\varepsilon_{\dot \alpha \dot \beta}
\end{pmatrix}\,.
\end{equation}
The spinor indices are raised and lowered by the Lorentz invariant antisymmetric tensors $\varepsilon^{\alpha\beta}$ and $\bar\varepsilon^{\dot\alpha\dot\beta}$ according to the rules
\begin{equation}
    A^\alpha = \varepsilon^{\alpha\beta}A_\beta\,, \quad A_\beta = \varepsilon_{\alpha\beta}A^\alpha\,, \quad A^{\dot\alpha} = \bar\varepsilon^{\dot\alpha\dot\beta}A_{\dot\beta}\,, \quad A_{\dot\beta} = \bar\varepsilon_{\dot\alpha\dot\beta}A^{\dot\alpha}\,.
\end{equation}

We demand that central elements $I_n$ obey (no summation over indices implied)
\begin{gather}
    Y^m_A * I_n = I_n * Y^m_A\,, \quad Y^n_A * I_n = I_n * Y^n_A = Y^n_A\,, \quad Z^m_A * I_n = I_n * Z^m_A\,, \\
    Z^n_A * I_n = I_n * Z^n_A = Z^n_A\,, \quad I_n*I_n = I_n\,, \quad I_n * I_m = I_m * I_n\,.
\end{gather}

 This is achieved by replacing standard oscillators $Y_A$, $Z_A$ with the tensor products $Y_A^n = Y_A \otimes e^n$, $Z_A^n = Z_A \otimes e^n$, where the basis element of the root space $e^n$ absorbs the corresponding idempotents, \ie $e^n I_n = I_n e^n = e^n$. It is important to note that these properties imply that any oscillator variable is accompanied (sometimes implicitly) by an indempotent sharing the same Coxeter index. Moreover, the explicit presence of idempotents in a star product means that there are no constant terms not multiplied by some idempotent, which is crucial for the resolution of the problem of missing massless states in the spectrum  \cite{Vasiliev:2018zer}. Therefore, the full nonlinear theory decomposes into  different sectors that mix in a triangle-like way. For instance, for $B_2$ theory, the terms with $I_2$ and $I_1 I_2$ do not contribute to the $I_1$ terms (and vice versa for $I_2$), while a product of $I_1$ and $I_2$ does contribute to the $I_1 I_2$ sector. In a $B_p$ CHS theory at the lowest level this brings a number of copies of the standard nonlinear HS theories associated with every idempotent $I_n$. Their mixing occurs at the higher multiparticle levels $\prod_{n\in X} I_n$, where $X$ is a subset of $\{1,\dots,p\}$.

From the star product and properties of $I_n$ it follows
\begin{equation}
    [Y_A^n,Y_B^m]_* = - [Z_A^n,Z_B^m]_* = 2i C_{AB} \delta^{nm} I_n\,, \quad [Y_A^n, Z_B^m]_* = 0\,.
\end{equation}
The appearance of the idempotents on the \rhs of the commutators distinguishes the framed Cherednik algebra from the standard one and leads to the resolution of the aforementioned rising energy problem.

From the (\ref{e:star product}) it is easy to derive
\begin{gather}
    Y_A^n * = Y_A^n +i \hat\partial_{Y}{}_A^n - i \hat\partial_{Z}{}_A^n\,, \quad * Y_A^n = Y_A^n -i \hat\partial_{Y}{}_A^n - i \hat\partial_{Z}{}_A^n\,, \\
    Z_A^n * = Z_A^n +i \hat\partial_{Y}{}_A^n - i \hat\partial_{Z}{}_A^n\,, \quad * Z_A^n = Z_A^n +i \hat\partial_{Y}{}_A^n + i \hat\partial_{Z}{}_A^n\,,
\end{gather}
where we introduce useful notation
\begin{equation}
    \hat\partial_{Y}{}_A^n := I_n \partial_{Y}{}_A^n\,, \quad \hat\partial_{Z}{}_A^n := I_n \partial_{Z}{}_A^n\,.
\end{equation}

Analogously to the standard HS construction, this star product admits inner Klein operators $\varkappa_v$, $\bar{\varkappa}_v$ associated with
the root vectors $v$
\begin{equation}
    \varkappa_v = \exp\bigg(i \frac{v^n v^m}{(v,v)} z_{\alpha n}y^\alpha_m \bigg)\,, \quad \bar\varkappa_v = \exp\bigg(i \frac{v^n v^m}{(v,v)} \bar z_{\dot\alpha n}\bar y^{\dot\alpha}_m \bigg)\,.
\end{equation}

One can see that the inner Klein operators $\varkappa_v$ generate the star product realization of the Coxeter group via
\begin{equation}
    \varkappa_v * q_\alpha^n = R_v{}^n{}_m q_\alpha^m * \varkappa_v\,, \quad q_\alpha^n = y_\alpha^n, z_\alpha^n\,,
\end{equation}
(and analogously for $\bar{q}_{\dot\alpha}$)
since $v^n = e^n (v,e^n)$,  where $e^n$ is the basis element of the root space.

Nonlinear equations for the generalized HS theory associated with the Coxeter group $\mathcal{C}$
are \cite{Vasiliev:2018zer}
\begin{gather}
    \d_x W + W*W = 0\,, \label{e:nonlinear system 1}\\
    \d_x B + W*B - B*W = 0\,, \label{e:nonlinear system 2}\\
    \d_x S + W*S + W*S = 0\,, \label{e:nonlinear system 3}\\
    S*B = B*S\,, \label{e:nonlinear system 4} \\
    S*S = i \bigg(dZ^{An}dZ_{An} + \sum_i \sum_{v \in \mathcal{R}_i} \bigg[F_{i*}(B) \frac{v^n v^m}{(v,v)}dz^\alpha_n dz_{\alpha m} * \varkappa_v \hat{k}_v  + \bar F_{i*}(B) \frac{v^n v^m}{(v,v)}d\bar z^{\dot\alpha}_n d\bar z_{\dot\alpha m} * \bar\varkappa_v \hat{\bar k}_v\bigg]\bigg)\,, \label{e:nonlinear system 5}
\end{gather}
where $\varkappa_v \hat{k}_v$ acts on $dz^\alpha_n$ as
\begin{equation}
    \varkappa_v\hat{k}_v * dz^\alpha_n = R_v{}_n{}^m dz^\alpha_m *  \varkappa_v\hat{k}_v\,,
\end{equation}
$F_{i*}(B)$ is any star product function of the zero-form $B$ on the conjugacy classes $\mathcal{R}_i$ of $\mathcal{C}$. In the following considerations, we set $F_{i*}(B) = \eta_i B$ to avoid problems with locality of expressions yielded by the star product.
Equations (\ref{e:nonlinear system 1})-(\ref{e:nonlinear system 5}) are formally consistent since the relations (\ref{e:cherednik comm}) respect the Jacobi identities, which in terms of the field equations are fulfilled due to the property that the \rhs of (\ref{e:nonlinear system 5}) is central. Indeed, one can check that
\begin{equation}
    \hat{\gamma}_i = \sum_{v \in \mathcal{R}_i} \frac{v^n v^m}{(v,v)}dz^\alpha_n dz_{\alpha m} * \varkappa_v \hat{k}_v
\end{equation}
and its conjugated counterpart $\hat{\bar{\gamma}}_i$ are central with respect to the star product (\ref{e:star product}). Therefore, the centrality of $\hat{\gamma}_i$, $\hat{\bar{\gamma}}_i$ and eq.(\ref{e:nonlinear system 4}) guarantee that $[S,S*S]_* = 0$. The equation (\ref{e:nonlinear system 5}) can be represented as (\ref{e:cherednik comm}) after the substitution $S = dz^\alpha_n q_\alpha^n$, $F_{i*}(B) = \nu(v)$ and some redefinition of the Klein operators. Therefore, the consistency condition $[S,S*S]_* = 0$ transforms into the Jacobi identity of the framed Cherednik algebra.

\section{$AdS_4$ solution}\label{AdS solution}

In this section we find the vacuum solution of the nonlinear system (\ref{e:nonlinear system 1})-(\ref{e:nonlinear system 5}), that describes $AdS_4$. It is easy to see that
\begin{equation}
    B_0 = 0\,, \quad S_0 = dZ^{An}Z_{An}\,, \quad W = W_0(Y,I|x)\,
\end{equation}
solve nonlinear equations provided that $W_0(Y,I|x)$ obeys the equation
\begin{equation}\label{e:zero_curv}
    \d_x W_0(Y,I|x) + W_0(Y,I|x)*W_0(Y,I|x) = 0\,.
\end{equation}

Consider a bilinear ansatz for $W_0(Y,I|x)$ that includes $dx$ one-forms $\omega^{nm}_{\alpha\beta}(I|x), \bar \omega^{nm}_{\dot \alpha \dot \beta}(I|x)$ and $e^{nm}_{\alpha \dot \alpha}(I|x)$
\begin{equation}\label{e:ansatz}
    W_0(Y,I|x) = -\frac{i}{4}\bigg(\omega^{nm}_{\alpha\beta}(I|x)y^\alpha_n y^\beta_m + \bar \omega^{nm}_{\dot \alpha \dot \beta}(I|x) \bar y^{\dot \alpha}_n \bar y^{\dot \beta}_m + 2 e^{nm}_{\alpha \dot \alpha}(I|x) y^\alpha_n \bar y^{\dot \alpha}_m\bigg)\,.
\end{equation}

Insertion of (\ref{e:ansatz}) into (\ref{e:zero_curv}) yields a set of equations on the one-forms $\omega^{nm}_{\alpha\beta}(I|x), \bar \omega^{nm}_{\dot \alpha \dot \beta}(I|x)$ and $e^{nm}_{\alpha \dot \alpha}(I|x)$:
\begin{gather}
    \bigg(\d_x \omega^{nm}_{\alpha\beta} + \sum_q \varepsilon^{\gamma \lambda}\omega^{nq}_{\alpha\gamma}\wedge \omega^{mq}_{\beta\lambda} I_q + \sum_q\bar \varepsilon^{\dot \alpha \dot \beta} e^{nq}_{\alpha \dot \alpha}\wedge e^{mq}_{\beta \dot \beta} I_q\bigg)y^\alpha_n y^\beta_m = 0\,, \label{e:AdS 1}\\
    \bigg(\d_x \bar\omega^{nm}_{\dot\alpha\dot\beta} + \sum_q\bar\varepsilon^{\dot\gamma \dot\lambda}\bar\omega^{nq}_{\dot\alpha\dot\gamma}\wedge \bar\omega^{mq}_{\dot\beta\dot\lambda} I_q +  \sum_q\varepsilon^{\alpha \beta} e^{nq}_{\alpha \dot \alpha}\wedge e^{mq}_{\beta \dot \beta} I_q \bigg)\bar y^{\dot \alpha}_n \bar y^{\dot \beta}_m = 0\,, \\
    \bigg(\d_x e^{nm}_{\alpha \dot \alpha} +  \sum_q\varepsilon^{\gamma \lambda}\omega^{nq}_{\alpha\gamma}\wedge e^{qm}_{\lambda \dot \alpha} I_q + \sum_q\bar\varepsilon^{\dot\gamma \dot\lambda}\bar\omega^{nq}_{\dot\alpha\dot\gamma}\wedge e^{qm}_{\alpha \dot \lambda} I_q\bigg) y^\alpha_n \bar y^{\dot \alpha}_m = 0\,. \label{e:AdS 3}
\end{gather}

Further restricting the components of $W_0(Y,I|x)$ as
\begin{equation}
    \omega^{nm}_{\alpha\beta}(I|x) = \omega_{\alpha\beta}(x) \delta^{nm}\,, \quad \bar \omega^{nm}_{\dot \alpha \dot \beta}(I|x) = \bar \omega_{\dot \alpha \dot \beta}(x)\delta^{nm}\,, \quad e^{nm}_{\alpha \dot \alpha}(I|x) = e_{\alpha \dot \alpha}(x)\delta^{nm}\,,
\end{equation}
where $\delta^{nm}$ is invariant under the action of any Coxeter group, as they are subgroups of $O(p)$, equations (\ref{e:AdS 1})-(\ref{e:AdS 3}) yield
\begin{gather}
    \d_x \omega_{\alpha\beta} + \varepsilon^{\gamma \lambda}\omega_{\alpha\gamma}\wedge \omega_{\beta\lambda}  + \bar \varepsilon^{\dot \alpha \dot \beta} e_{\alpha \dot \alpha}\wedge e_{\beta \dot \beta} = 0\,, \\
    \d_x \bar\omega_{\dot\alpha\dot\beta} + \bar\varepsilon^{\dot\gamma \dot\lambda} \bar\omega_{\dot\alpha\dot\gamma}\wedge\bar\omega_{\dot\beta\dot\lambda} +  \varepsilon^{\alpha \beta} e_{\alpha \dot \alpha}\wedge e_{\beta \dot \beta} = 0\,, \\
    \d_x e_{\alpha \dot \alpha} +  \varepsilon^{\gamma \lambda}\omega_{\alpha\gamma}\wedge e_{\lambda \dot \alpha} + \bar\varepsilon^{\dot\gamma \dot\lambda}\bar\omega_{\dot\alpha\dot\gamma}\wedge e_{\alpha \dot \lambda}  = 0\,,
\end{gather}
which encode $AdS_4$ spin-connections $\omega_{\alpha\beta}, \bar\omega_{\dot\alpha\dot\beta}$ and vierbein $e_{\alpha \dot \alpha}$. Therefore, in a general CHS theory, $AdS_4$  is represented by a $dx$ one-form
\begin{equation}\label{e:AdS connection}
    \Omega_{AdS}(Y|x) = -\frac{i}{4}\delta^{nm}\bigg(\omega_{\alpha\beta}(x)y^\alpha_n y^\beta_m + \bar \omega_{\dot \alpha \dot \beta}(x) \bar y^{\dot \alpha}_n \bar y^{\dot \beta}_m + 2 e_{\alpha \dot \alpha}(x) y^\alpha_n \bar y^{\dot \alpha}_m\bigg)\,.
\end{equation}

It is worth noting that the  $AdS_4$ connection (\ref{e:AdS connection}) has no explicit dependence on idempotents $I_n$, which means that the covariant derivative preserves the filtration of the fields with respect to idempotents. This happened because we introduced a set of idempotents in a way that does not distinguish between holomorphic $y^\alpha_n, z^\alpha_n, \hat{k}_v$ and anti-holomorphic $\bar{y}^{\dot{\alpha}}_n, \bar{z}^{\dot{\alpha}}_n, \hat{\bar{k}}_v$ variables and Klein operators. One may consider a model $\mathcal{C}\times \mathcal{C}$ in $4d$ space with doubled set of idempotents $I_n, \bar{I}_n$ and even find a solution of (\ref{e:zero_curv}) corresponding to the $AdS_4$, that has an explicit dependence on idempotents $I_n, \bar{I}_n$. However, the analysis of the lower-rank states \cite{Vasiliev:2018zer} and the $AdS_4$ covariant derivative shows that such model cannot be interpreted as a generalization of the standard $4d$ HS theory, but rather being a product of the two $3d$ ones. Due to (\ref{e:KI comm}), both $I_n$ and $\bar{I}_n$ commute with dressed Klein operators $\hat{k}_v$ and $\hat{\bar k}_v$ and $I_n - \bar{I}_n$ generates an ideal $\mathcal{J}$ of the $\mathcal{C}\times \mathcal{C}$ system. In the model $(\mathcal{C}\times \mathcal{C})/\mathcal{J}$ the lowest states are associated with $4d$ massless fields represented by functions of a single copy of oscillators $y^\alpha_n, z^\alpha_n, \bar{y}^{\dot \alpha}_n, \bar{z}^{\dot \alpha}_n$ and $I_n$, \ie the fields $\omega$ and $C$ -- lowest-rank $Z$-independent parts of the $W$ and $B$ fields
\begin{gather}
\omega=\sum_{i}^{p} \omega\left(y_i, \hat{k}_i ; \bar{y}_i, \hat{\bar{k}}_i | x\right) * I_i, \quad \omega\left(y_i, \hat{k}_i ; \bar{y}_i, \hat{\bar{k}}_i | x\right)=\omega\left(y_i,-\hat{k}_i ; \bar{y}_i,-\hat{\bar{k}}_i | x\right)\,, \\
C=\sum_{i}^{p} C\left(y_i, \hat{k}_i ; \bar{y}_i, \hat{\bar{k}}_i | x\right) * I_i, \quad C\left(y_i, \hat{k}_i ; \bar{y}_i, \hat{\bar{k}}_i | x\right)=-C\left(y_i,-\hat{k}_i ; \bar{y}_i,-\hat{\bar{k}}_i | x\right)\,
\end{gather}
describe the massless fields of standard HS theory. Therefore, in the sequel we use a set of idempotents $I_n$ that does not distinguish between holomorphic and anti-holomorphic sectors.

An interesting observation is that zero-curvature equation (\ref{e:zero_curv}) admits a set of solutions parameterized by a $SO(p,\mathds{R})$ rotation with an explicit dependence on idempotents which still encode $AdS_4$ background geometry. The representative of such family has a form
\begin{equation}
    \Omega_{AdS}(Y,I|x|A) = -\frac{i}{4}\bigg(\omega_{\alpha\beta}(x)\delta^{nm} y^\alpha_n y^\beta_m + \bar \omega_{\dot \alpha \dot \beta}(x)\delta^{nm} \bar y^{\dot \alpha}_n \bar y^{\dot \beta}_m + 2 e_{\alpha \dot \alpha}(x) A^n{}_m y^\alpha_n \bar y^{\dot \alpha}{}^m\bigg)\prod^{p}_{j=1}I_{j}\,,
\end{equation}
where
\begin{equation}
     A^T A = \mathbb{1}\,, \quad A \in SO(p,\mathds{R})\,.
\end{equation}
This family of solution exists due to the presence of $SO(p,\mathds{R})$-invariant contraction between auxiliary parameters of integration $S^A_n$ and $T^A_n$ in the star product (\ref{e:star product}). It allows us to redefine variables $y^\alpha_n$ and $\bar{y}^{\dot\alpha}_n$ to absorb $SO(p,\mathds{R})$-rotation and return to the curvature (\ref{e:AdS connection}) multiplied by the full set of idempotents. However, such redefinition of variables will change the action of Klein operators $\hat{K}_v$ and affect the structure of underlying modules. Therefore, we naively have a set of nonequivalent vacua to study. Fortunately, the requirement of anti-hermicity of the connection under the conjugation $y^\dagger = \bar{y}$ and the preservation of the field filtration by the covariant derivative rule out any vacuum connection with a non-trivial $A$. Thus, we are left with the connection (\ref{e:AdS connection}).

\section{Covariant derivatives and modules}\label{Derivative and modules}

\subsection{Covariant derivative}
In this section we analyze the covariant derivative
\begin{equation}
    D_\Omega (\bullet)= \d_x (\bullet) + [\Omega_{AdS}, \bullet]_*\,
\end{equation}
built from $\Omega_{AdS}(Y|x)$ acting on various related modules.

After some calculations involving star product (\ref{e:star product}) and commutation properties of the dressed Klein operators $\hat{k}_v, \hat{\bar k}_v$ we get
\begin{multline}\label{e: covariant derivative}
    D_\Omega f(Y,I;\hat{k},\hat{\bar k}|x) = \bigg[D_L + \frac{1}{2}\delta^{nm}e^{\alpha \dot\alpha}\bigg(\mathbb{1}^k_n \bar{\mathbb{1}}^l_m + R(k){}^k_n \bar{R}(\bar{k}){}^l_m \bigg)(y_{\alpha k} \hat{\bar\partial}_{\dot\alpha l} + \bar y_{\dot\alpha l} \hat{\partial}_{\alpha k}) - \\
    -\frac{i}{2}\delta^{nm}e^{\alpha \dot\alpha}\bigg(\mathbb{1}^k_n \bar{\mathbb{1}}^l_m - R(k){}^k_n \bar{R}(\bar{k}){}^l_m \bigg)(y_{\alpha k}\bar y_{\dot\alpha l} -  \hat\partial_{\alpha k}\hat{\bar\partial}_{\dot\alpha l})\bigg]f(Y,I;\hat{k},\hat{\bar k}|x)\,,
\end{multline}
\begin{equation}
    D_L f(Y,I;\hat{k},\hat{\bar k}|x) :=  \d_x f(Y,I;\hat{k},\hat{\bar k}|x) + \delta^{nm}\bigg(\omega^{\alpha \beta} y_{\alpha n}\partial_{\beta m} + \bar \omega^{\dot\alpha \dot\beta} \bar y_{\dot\alpha n}\bar\partial_{\dot\beta m}\bigg)f(Y,I;\hat{k},\hat{\bar k}|x)\,,
\end{equation}
where $\mathbb{1}^k_n$ and $\bar{\mathbb{1}}^l_m$ are identity matrices, $\hat{k}$ and $\hat{\bar k}$ are products of some elementary dressed Klein operators $\hat{k}_v$ and $\hat{\bar k}_v$
or equal to the unity element, matrices $R(k){}^k_n$ and $\bar{R}(\bar{k}){}^l_m$ are reflections in the root space which correspond to the products of elementary dressed Klein operators $\hat{k}$ and $\hat{\bar k}$ (identity matrix in case of $\hat{k} = 1$). In the sequel, we often write $R{}^k_n$ and $\bar{R}{}^l_m$ omitting the dependence on the Klein operators if its source is obvious.

Let us stress that Lorenz covariant derivative $D_L$ acquires its canonical form due to the $sl(2)$ invariance of the framed Cherednik algebra (\ref{e: sl2 invariance}). Notice that one can use unhatted derivatives $\partial_{\alpha n}$ in $D_L$ due to (\ref{e: qIK1}).

The form of equation (\ref{e: covariant derivative}) can be expected considering $\delta^{nm}$ in (\ref{e:AdS connection}) is an invariant $O(p)$ metric ($p$ being the rank of the Coxeter group). Since the matrices $R(k)$, generated by Coxeter elements $\hat{k}_v$, also reside in the $O(p)$ group, any combination $\delta^{nm} R(k){}^k_n \bar{R}(\bar{k}){}^l_m $ can be rewritten as a linear combination of $\delta^{nl} R(k){}^k_n$, \ie the anti-holomorphic variables do not add new independent equations. Thus, the total number of independent covariant derivatives and thus equations at the linear level is bound to be equal to the order of the Coxeter group in question. This statement holds only in the linear case, as different fields which obey the same linearized equations in general obey different field equations beyond the free field approximation.

The above reasoning can be illustrated by the standard $4d$ HS theory with the Coxeter group $A_1 \cong \mathds{Z}_2$. The group $A_1$ has a one-dimensional root space and a single element $\hat{k}$ that changes a sign of the unique root vector. Thence, covariant derivative (\ref{e: covariant derivative}) reduces to the two cases
\begin{gather}
    D_\Omega f(Y|x) = D_L f(Y|x) + e^{\alpha \dot\alpha}(y_{\alpha} \bar\partial_{\dot\alpha } + \bar y_{\dot\alpha } \partial_{\alpha })f(Y|x)\,, \label{e: adj standard}\\
    D_\Omega (f(Y|x)\hat{k}) = D_L f(Y|x)\hat{k} - i e^{\alpha \dot\alpha}(y_{\alpha}\bar y_{\dot\alpha} - \partial_{\alpha}\bar\partial_{\dot\alpha})f(Y|x)\hat{k} \label{e: tw standard}\,.
\end{gather}
The first case is an adjoint module in which physical fields $\omega(Y;\hat{K}|x)$ are valued and the second one describes a twisted-adjoint module of physical fields $C(Y;\hat{K}|x)$. It is well-known in a standard HS theory that the adjoint module is non-unitary since it is an infinite sum of finite (and thus non-unitary) modules of a non-compact algebra while the twisted-adjoint module is an infinite sum of infinite modules, complex equivalent to the unitary ones used to describe single particle states \cite{Vasiliev:2001zy}.

In a general CHS model a mixing of adjoint and twisted-adjoint modules occurs. Moreover, some modules do not have a form of the tensor products of standard HS adjoint and twisted-adjoint modules. We will refer to those modules that are not isomorphic to the tensor product of standard adjoint and twisted-adjoint modules as \textit{entangled} and to those that are as \textit{disentangled}. The structure of the resulting CHS module depends on the properties of matrices
\begin{equation} \label{def:orthogonalProjectors}
P_{\pm}^{kl}=\frac{1}{2}\delta^{nm}\bigg(\mathbb{1}^k_n \bar{\mathbb{1}}^l_m \pm R(k){}^k_n \bar{R}(\bar{k}){}^l_m \bigg).
\end{equation}
Matrices $P_{\pm}^{kl}$ resemble a pair of orthogonal projectors. However, to be a set of projectors the condition $(R\bar{R}{}^T)^2 = \mathbb{1}$ must be met. In that case a pair $P_{\pm}^{kl}$ are orthogonal projectors and the corresponding module disentangle into the product of standard HS modules.

\begin{proposition}[Disentanglement criterion]
\label{Disentanglement criteria}
    $(R\bar{R}{}^T)^2 = \mathbb{1}$ is necessary and sufficient condition for the module to be disentangled.
\end{proposition}
\begin{proof}
Indeed, if the module is a product of standard HS modules then
\begin{equation*}
    R\bar{R}{}^T=\text{diag}(+1,...,+1,-1,...,-1)
\end{equation*}
and $(R\bar{R}{}^T)^2 = \mathbb{1}$.\\
If $(R\bar{R}{}^T)^2 = \mathbb{1}$ then the minimal polynomial of $R\bar{R}{}^T$ is either $q_{min}(t)= t \pm 1$ or $q_{min}(t) = t^2 - 1$. In the first case $R\bar{R}{}^T = \pm \mathbb{1}$. In the second case $R\bar{R}{}^T$ is diagonalizable with eigenvalues $\lambda = \pm 1$. Since matrices $R\bar{R}{}^T$ and eigenvalues $\lambda$ are real, the diagonalization of $R\bar{R}{}^T$ occurs over the $\mathds{R}$-field.
\end{proof}

Disentangled modules exist in any CHS theory since $(R_v)^2 = \mathbb{1}$ for any root vector $v$. In general $(R\bar{R}{}^T)^2 \neq \mathbb{1}$ for any non-trivial (\textit{i.e.}, beyond $A_1 \cong \mathds{Z}_2$) Coxeter group, therefore the resulting CHS module is a product of standard HS modules and infinite-dimensional entangled modules of the new type. Such modules appear in all CHS models with a non-trivial $\mathcal{C}$. For example, a group $B_{p}$ with $p\geq3$ contains cycles of length $n$ with $3\leq n\leq p$ and the square of the $n$-cycle is not an identity transformation which means that the corresponding module is entangled. Thus, a question of unitarizability of the CHS modules arises. In the sequel of this section we perform a full classification of unitary and non-unitary $B_2$ modules.

\subsection{Covariant constancy equations in the $B_2$ theory}

The root system of $B_2$ consists of two conjugacy classes
\begin{equation}\label{e: B2 conjugacy classes}
    \mathcal{R}_1 = \{\pm e^1, \pm e^2\}\,, \quad \mathcal{R}_2 = \{\pm e^1 \pm e^2\}\,.
\end{equation}
A generating set of $B_2$ is $\{R_{e^i}, R_{e^1-e^2}\}$. Holomorphic Klein operators associated with the generating reflections are $k_i$ and $k_{12}$. The holomorphic group $B_2$ is generated by
\begin{equation}
   \Bigl\{ k_i, k_{12}| k_i^2 = 1, k_{12}^2 = 1, k_1 k_{12} = k_{12} k_2, k_2 k_{12} = k_{12} k_1, i \in \{1,2\} \Bigl\} \,.
\end{equation}
It is useful to denote the product of all generators as
\begin{equation}
k^+_{12} := k_1 k_2 k_{12}\,
\end{equation}
and view $k^+_{12}$ as an additional redundant generator that corresponds to the root vector $e^1+e^2 \in \mathcal{R}_2$ (reflection with respect to $e^1+e^2$ is equivalent to the composition of reflections with respect to $e^1-e^2$ and basis vectors $e^i$). By doing this, we equate the number of Klein operators corresponding to the conjugacy classes $\mathcal{R}_i$ ($\mathcal{R}_1$ corresponds to two Klein operators (reflections) $\{k_1, k_2\}$ and $\mathcal{R}_2$ corresponds to $\{k_{12}, k_{12}^+\}$).

The reflection matrices $R(k)$ in (\ref{e: covariant derivative}) are
\begin{gather}
R(1) = \begin{pmatrix}
1 & 0 \\
0 & 1
\end{pmatrix}\,,
R(k_1) = \begin{pmatrix}
-1 & 0 \\
0 & 1
\end{pmatrix}\,,
R(k_2) = \begin{pmatrix}
1 & 0 \\
0 & -1
\end{pmatrix}\,,
R(k_{12}) = \begin{pmatrix}
0 & 1 \\
1 & 0
\end{pmatrix}\,, \label{e: B2 R first}\\
R(k_1 k_2) = \begin{pmatrix}
-1 & 0 \\
0 & -1
\end{pmatrix}\,,
R(k_1 k_{12}) = \begin{pmatrix}
0 & -1 \\
1 & 0
\end{pmatrix}\,,
R(k_2 k_{12}) = \begin{pmatrix}
0 & 1 \\
-1 & 0
\end{pmatrix}\,,\\
R(k^+_{12}) = \begin{pmatrix}
0 & -1 \\
-1 & 0
\end{pmatrix}\,. \label{e: B2 R last}
\end{gather}
Analogous matching of the reflection matrices $\bar R(\bar k)$ takes place for the anti-holomorphic Klein operators $\bar k_i, \bar k_{12}$.

In the $B_2$ HS model, the zero-form field $C(Y_1,Y_2,I;\hat{k},\hat{\bar k}|x)$ has $64$  component fields in the $I_1 I_2$ sector
\begin{equation}
    C(Y_1,Y_2;\hat{k},\hat{\bar k}|x)*I_1 I_2 = \sum_{a,b,c,\bar a,\bar b,\bar c = 0}^1 C_{abc\bar a \bar b \bar c}(Y_1,Y_2|x)*I_1 I_2*\hat k_1^a \hat k_2^b \hat k_{12}^c \hat{\bar k}{}_1^{\bar a}\hat{\bar k}{}_2^{\bar b}\hat{\bar k}{}_{12}^{\bar c}\,
\end{equation}
and $4$ component fields in each $I_i$ sector
\begin{equation}
    C(Y_i;\hat{k},\hat{\bar k}|x)*I_i = \sum_{a,\bar a = 0}^1 C_{a \bar a}(Y_i|x)*I_i*\hat k_i^a \hat{\bar k}{}_i^{\bar a}\,
\end{equation}
that naively leads to $64$ linearized  covariant constancy equations (\ref{e: covariant derivative}). However, as discussed in the previous section, the actual number of types of independent equations is equal to the order of the Coxeter group.

 In particular, in the case of $B_2$,  all possible matrix products $R(k)\bar R(\bar k){}^T$ group into the $8$ categories
\begin{gather}
    R(k)\bar{R}(\bar{k}){}^T = \Biggl\{ \begin{pmatrix}
1 & 0 \\
0 & 1
\end{pmatrix}\,,
\begin{pmatrix}
-1 & 0 \\
0 & 1
\end{pmatrix}\,,
\begin{pmatrix}
1 & 0 \\
0 & -1
\end{pmatrix}\,,
\begin{pmatrix}
-1 & 0 \\
0 & -1
\end{pmatrix}\,,
\begin{pmatrix}
0 & 1 \\
1 & 0
\end{pmatrix}\,, \label{e: RbarR1} \\
\begin{pmatrix}
0 & -1 \\
-1 & 0
\end{pmatrix}\,,
\begin{pmatrix}
0 & -1 \\
1 & 0
\end{pmatrix}\,,
\begin{pmatrix}
0 & 1 \\
-1 & 0
\end{pmatrix}\, \label{e: RbarR2}\Biggr\}.
\end{gather}
Therefore, there are $8$ types of covariant constancy equations (modules)
\begin{gather}
    \bigg(D_L + e^{\alpha \dot\alpha}\sum_{i = 1}^2 (y_{\alpha i}\bar\partial_{\dot \alpha i} + \bar y_{\dot\alpha i}\partial_{\alpha i})\bigg)C(Y_1,Y_2,I;\hat{k},\hat{\bar k}|x) = 0\,, \label{e:cov1}\\
    \bigg(D_L - i e^{\alpha \dot\alpha}(y_{\alpha 1}\bar y_{\dot \alpha 1} - \partial_{\alpha 1} \bar \partial_{\dot\alpha 1}) + e^{\alpha \dot\alpha}(y_{\alpha 2}\bar\partial_{\dot \alpha 2} + \bar y_{\dot\alpha 2}\partial_{\alpha 2})\bigg)C(Y_1,Y_2,I;\hat{k},\hat{\bar k}|x) = 0\,, \label{e:cov2}\\
    \bigg(D_L + e^{\alpha \dot\alpha}(y_{\alpha 1}\bar\partial_{\dot \alpha 1} + \bar y_{\dot\alpha 1}\partial_{\alpha 1}) - i e^{\alpha \dot\alpha}(y_{\alpha 2}\bar y_{\dot \alpha 2} - \partial_{\alpha 2} \bar \partial_{\dot\alpha 2}) \bigg)C(Y_1,Y_2,I;\hat{k},\hat{\bar k}|x) = 0\,, \label{e:cov3}\\
    \bigg(D_L - i e^{\alpha \dot\alpha}\sum_{i = 1}^2 (y_{\alpha i}\bar y_{\dot \alpha i} - \partial_{\alpha i} \bar \partial_{\dot\alpha i})\bigg)C(Y_1,Y_2,I;\hat{k},\hat{\bar k}|x) = 0\,, \label{e:cov4}
\end{gather}
\begin{multline}
    \bigg(D_L + \frac{1}{2}e^{\alpha \dot \alpha}\bigg[(y_{\alpha 1} + y_{\alpha 2})(\bar\partial_{\dot \alpha 1} + \bar\partial_{\dot \alpha 2}) + (\bar y_{\dot \alpha 1} + \bar y_{\dot \alpha 2})(\partial_{\alpha 1} + \partial_{\alpha 2})\bigg] - \\- \frac{i}{2}e^{\alpha \dot \alpha}\bigg[(y_{\alpha 1} - y_{\alpha 2})(\bar y_{\dot\alpha 1} - \bar y_{\dot\alpha 2}) - (\partial_{\alpha 1} - \partial_{\alpha 2})(\bar\partial_{\dot \alpha 1} - \bar\partial_{\dot \alpha 2})\bigg]\bigg)C(Y_1,Y_2,I;\hat{k},\hat{\bar k}|x) = 0\,, \label{e:cov5}
\end{multline}
\begin{multline}
    \bigg(D_L + \frac{1}{2}e^{\alpha \dot \alpha}\bigg[(y_{\alpha 1} - y_{\alpha 2})(\bar\partial_{\dot \alpha 1} - \bar\partial_{\dot \alpha 2}) + (\bar y_{\dot \alpha 1} - \bar y_{\dot \alpha 2})(\partial_{\alpha 1} - \partial_{\alpha 2})\bigg] - \\- \frac{i}{2}e^{\alpha \dot \alpha}\bigg[(y_{\alpha 1} + y_{\alpha 2})(\bar y_{\dot\alpha 1} + \bar y_{\dot\alpha 2}) - (\partial_{\alpha 1} + \partial_{\alpha 2})(\bar\partial_{\dot \alpha 1} + \bar\partial_{\dot \alpha 2})\bigg]\bigg)C(Y_1,Y_2,I;\hat{k},\hat{\bar k}|x) = 0\,, \label{e:cov6}
\end{multline}
\begin{multline}
    \bigg(D_L + \frac{1}{2}e^{\alpha \dot \alpha}\bigg[y_{\alpha 1}(\bar\partial_{\dot\alpha 1} - \bar\partial_{\dot\alpha 2}) + y_{\alpha 2}(\bar\partial_{\dot\alpha 1} + \bar\partial_{\dot\alpha 2}) + \bar y_{\dot\alpha 1}(\partial_{\alpha 1} + \partial_{\alpha 2}) - \bar y_{\dot \alpha 2}(\partial_{\alpha 1} - \partial_{\alpha 2})\bigg] - \\- \frac{i}{2}e^{\alpha \dot \alpha}\bigg[y_{\alpha 1}(\bar y_{\dot\alpha 1} + \bar y_{\dot\alpha 2}) - y_{\alpha 2}(\bar y_{\dot\alpha 1} - \bar y_{\dot\alpha 2}) - \partial_{\alpha 1}(\bar\partial_{\dot\alpha 1} + \bar\partial_{\dot\alpha 2}) + \partial_{\alpha 2}(\bar\partial_{\dot\alpha 1} - \bar\partial_{\dot\alpha 2})\bigg]\bigg)C(Y_1,Y_2,I;\hat{k},\hat{\bar k}|x) = 0\,,\label{e:cov7}
\end{multline}
\begin{multline}
    \bigg(D_L + \frac{1}{2}e^{\alpha \dot \alpha}\bigg[y_{\alpha 1}(\bar\partial_{\dot\alpha 1} + \bar\partial_{\dot\alpha 2}) - y_{\alpha 2}(\bar\partial_{\dot\alpha 1} - \bar\partial_{\dot\alpha 2}) + \bar y_{\dot\alpha 1}(\partial_{\alpha 1} - \partial_{\alpha 2}) + \bar y_{\dot \alpha 2}(\partial_{\alpha 1} + \partial_{\alpha 2})\bigg] - \\- \frac{i}{2}e^{\alpha \dot \alpha}\bigg[y_{\alpha 1}(\bar y_{\dot\alpha 1} - \bar y_{\dot\alpha 2}) + y_{\alpha 2}(\bar y_{\dot\alpha 1} + \bar y_{\dot\alpha 2}) - \partial_{\alpha 1}(\bar\partial_{\dot\alpha 1} - \bar\partial_{\dot\alpha 2}) - \partial_{\alpha 2}(\bar\partial_{\dot\alpha 1} + \bar\partial_{\dot\alpha 2})\bigg]\bigg)C(Y_1,Y_2,I;\hat{k},\hat{\bar k}|x) = 0\,,\label{e:cov8}
\end{multline}
corresponding to the matrices (\ref{e: RbarR1}), (\ref{e: RbarR2}) reading from left to right, from top to bottom. It is worth noting that unhatted derivatives $\partial_{\alpha i}$ appear in the covariant constancy equations due to the properties (\ref{e: qIK1}) and (\ref{e:KI absorb}).

In terms of types of modules, all but the last two modules are disentangled. However, there is a way to represent them as deformed disentangled via a nonlocal field redefinition. More precisely, exponential ansatzes
\begin{align}
    C(Y_1,Y_2,I;\hat{k},\hat{\bar k}|x) &= \exp\bigg(-i y_{1\alpha}y_2^{\alpha} + i \bar{y}_{1 \dot\alpha}\bar{y}_2^{\dot\alpha}\bigg) \tilde{C}(Y_1,Y_2,I;\hat{k},\hat{\bar k}|x)\,, \label{def:expAnsatz} \\
    C(Y_1,Y_2,I;\hat{k},\hat{\bar k}|x) &= \exp\bigg(i y_{1\alpha}y_2^{\alpha} - i \bar{y}_{1 \dot\alpha}\bar{y}_2^{\dot\alpha}\bigg) \tilde{C}(Y_1,Y_2,I;\hat{k},\hat{\bar k}|x)\,
\end{align}
transform entangled equations into
\begin{equation}
    \bigg(D_L - \frac{i}{2}e^{\alpha \dot \alpha}\bigg[2 y_{\alpha 1}(\bar y_{\dot\alpha 1} + \bar y_{\dot\alpha 2}) - 2y_{\alpha 2}(\bar y_{\dot\alpha 1} - \bar y_{\dot\alpha 2}) - \partial_{\alpha 1}(\bar\partial_{\dot\alpha 1} + \bar\partial_{\dot\alpha 2}) + \partial_{\alpha 2}(\bar\partial_{\dot\alpha 1} - \bar\partial_{\dot\alpha 2})\bigg]\bigg)\tilde{C}(Y_1,Y_2,I;\hat{k},\hat{\bar k}|x) = 0\,,
\end{equation}
\begin{equation}
    \bigg(D_L - \frac{i}{2}e^{\alpha \dot \alpha}\bigg[2 y_{\alpha 1}(\bar y_{\dot\alpha 1} - \bar y_{\dot\alpha 2}) + 2 y_{\alpha 2}(\bar y_{\dot\alpha 1} + \bar y_{\dot\alpha 2}) - \partial_{\alpha 1}(\bar\partial_{\dot\alpha 1} - \bar\partial_{\dot\alpha 2}) - \partial_{\alpha 2}(\bar\partial_{\dot\alpha 1} + \bar\partial_{\dot\alpha 2})\bigg]\bigg)\tilde{C}(Y_1,Y_2,I;\hat{k},\hat{\bar k}|x) = 0\,.
\end{equation}
By a linear change limited to holomorphic variables, for example, $(y^{'}_{\alpha 1} = y_{\alpha 1} - y_{\alpha 2}\,; y^{'}_{\alpha 2} = y_{\alpha 1} + y_{\alpha 2})$ in the first equation, one can transform the remaining equation on $\tilde{C}(Y_1,Y_2, I;\hat{k},\hat{\bar k}|x)$ into the equation (\ref{e:cov4}), which describes a tensor product of two standard twisted-adjoint modules. While such a transformation is obviously inconsistent with the conjugation rules of $Y^A_n$ (\ie oscillators $y^{'}_{\alpha i}$ and $\bar{y}_{\dot\alpha i}$ are not conjugated), our two entangled modules resemble a product of two twisted-adjoint ones entangled by the exponential factor. This resemblance does not imply an isomorphism between the modules and, as will be shown later, the modules have different properties, in particular, in regards to unitarity. The suggested ansatzes have to be treated with caution as the star product behavior of the exponential factors is ill-defined. Fortunately, this problem is a feature of the star product in the current approach with oscillators $Y^A_n$ that can be avoided in the linear order after a transition to the doubled set of oscillators as will be shown in Section \ref{Fock Space Realization}.

\subsection{Boundary Conditions} \label{BoundaryConditions}
Equations (\ref{e:cov1})-(\ref{e:cov4}) describe tensor products of adjoint and twisted-adjoint modules of the standard $4d$ HS theory that will be denoted as $\{M_{adj\otimes adj}, M_{tw\otimes adj}, M_{adj\otimes tw}, M_{tw\otimes tw}\}$. In the equations (\ref{e:cov5}) and (\ref{e:cov6}) one can perform a change of variables
\begin{equation}\label{e:change of variables y}
    y_{\alpha +} = \frac{1}{\sqrt{2}}(y_{\alpha 1} + y_{\alpha 2})\,,  \quad y_{\alpha -} = \frac{1}{\sqrt{2}}(y_{\alpha 1} - y_{\alpha 2})\,, \quad
    \bar{y}_{\dot\alpha +} = \frac{1}{\sqrt{2}}(\bar y_{\dot\alpha 1} + \bar y_{\dot\alpha 2})\,, \quad \bar{y}_{\dot\alpha -} = \frac{1}{\sqrt{2}}(\bar y_{\dot\alpha 1} - \bar y_{\dot\alpha 2})\,,
\end{equation}
to transform them into  equations (\ref{e:cov2}) and (\ref{e:cov3}). The existence of this change of variables is attributed to the fact that for these specific cases of reflection matrices $R\bar{R}{}^T$ operators  $P_{\pm}^{kl}$ (\ref{def:orthogonalProjectors}) are orthogonal projectors. Hence, modules (\ref{e:cov5}) and (\ref{e:cov6}) also describe tensor products of adjoint and twisted-adjoint modules of the standard $4d$ HS theory in appropriate variables and therefore equations (\ref{e:cov1})-(\ref{e:cov6}) correspond to the unitary modules, provided that the adjoint part is eliminated (set to be a constant) by imposing appropriate boundary conditions. An interesting observation is that the change of variables (\ref{e:change of variables y}) swaps conjugacy classes $\mathcal{R}_1$ and $\mathcal{R}_2$ and therefore swaps Klein operator $\{\hat{k}_1,\hat{k}_2\}$ and $\{\hat{k}_{12},\hat{k}^+_{12}\}$ in the $I_1 I_2$ sector. This is a unique feature of the $B_2$ group because for a general group $B_p$ conjugacy classes $\mathcal{R}_1$ and $\mathcal{R}_2$ have different sizes.

To impose the required boundary conditions, let us consider adjoint and twisted-adjoint equations (\ref{e: adj standard}) and (\ref{e: tw standard}) in the standard HS theory. In the stereographic coordinates for the hyperboloid realization of $AdS_4$
\begin{equation}
\begin{aligned}
& e_{0 \underline{n}}{ }^{\alpha \dot{\beta}}=-z^{-1} \sigma_{\underline{n}}{ }^{\alpha \dot{\beta}}\,,\\
& \omega_{0 \underline{n}}^{\alpha \beta}=- \lambda^2(2z)^{-1}(\sigma_{\underline{n}}^{\alpha \dot{\beta}} x^\beta{ }_{\dot{\beta}} + \sigma_{\underline{n}}^{\beta \dot{\beta}} x^\alpha{ }_{\dot{\beta}})\,, \\
& \bar{\omega}_{0 \underline{n}}{ }^{\dot{\alpha} \dot{\beta}}=- \lambda^2(2z)^{-1} (\sigma_{\underline{n}}{ }^{\alpha \dot{\alpha}} x_\alpha{ }^{\dot{\beta}} + \sigma_{\underline{n}}{ }^{\alpha \dot{\beta}} x_\alpha{ }^{\dot{\alpha}})\,, \\
&
\end{aligned}
\end{equation}
where $\lambda$ is an inverse radius of $AdS_4$,
\begin{equation}
x_{\alpha \dot{\beta}}=\sigma_{\alpha \dot{\beta}}^a x_a, \quad x^2=x_a x^a=\frac{1}{2} x_{\alpha \dot{\beta}} x^{\alpha \dot{\beta}}, \quad z=1+\lambda^2 x^2,
\end{equation}
and sigma-matrices $\sigma_{\alpha \dot{\beta}}^a$ are Hermitian, with the normalization $\sigma_{a\alpha \dot{\beta}} \sigma^{\alpha \dot{\beta}}_b = 2\eta_{ab}$ where $\eta_{ab} = diag(1,-1,-1,-1)$.

As was shown in \cite{Bolotin:1999fa}, $AdS_4$ connection $\Omega_{AdS}(y, \bar{y} | x)$ can be represented as
\begin{equation}
    \Omega_{AdS}(y, \bar{y} | x) = g^{-1}(y, \bar{y} | x)*\d_x g(y, \bar{y} | x)\,,
\end{equation}
where
\begin{equation}
    g(y, \bar{y} | x)=2 \frac{\sqrt{z}}{1+\sqrt{z}} \exp \left[\frac{i \lambda}{1+\sqrt{z}} x^{\alpha \dot{\alpha}} y_\alpha \bar{y}_{\dot{\alpha}}\right]
\end{equation}
with the inverse
\begin{equation}
    g^{-1}(y, \bar{y} | x)= \tilde{g}(y, \bar{y} | x) = 2 \frac{\sqrt{z}}{1+\sqrt{z}} \exp \left[-\frac{i\lambda}{1+\sqrt{z}} x^{\alpha \dot{\alpha}} y_\alpha \bar{y}_{\dot{\alpha}}\right]\,.
\end{equation}

Then the general solutions $C_{adj}(Y|x)$ of (\ref{e: adj standard}) and $C_{tw}(Y|x)$ of (\ref{e: tw standard}) are \cite{Didenko:2003aa}
\begin{equation}\label{e: general solution standard}
    C_{tw}(Y|x) = g^{-1}*C_{0tw}(Y)*\tilde{g}\,, \quad C_{adj}(Y|x) = g^{-1}*C_{0adj}(Y)*g\,,
\end{equation}
where $C_{0tw}(Y)$ and $C_{0adj}(Y)$ serve as initial data. After some calculations one can see that
\begin{equation}
    C_{tw}(Y|x) = z \exp{-i \lambda x^{\alpha\dot\alpha}y_\alpha \bar{y}_{\dot{\alpha}} + i(\sqrt{z}-1)y^\alpha p_{0\alpha} + i(\sqrt{z}-1)\bar{y}^{\dot{\alpha}} \bar{p}_{0\dot{\alpha}} -i \lambda x^{\alpha\dot\alpha}p_{0\alpha}\bar{p}_{0\dot{\alpha}}}C_{0tw}(Y)\,,
\end{equation}
\begin{equation}
    C_{adj}(Y|x) = C_{0adj}\bigg(\frac{1}{\sqrt{z}}(y^\alpha + \lambda x^{\alpha\dot\alpha}\bar{y}_{\dot{\alpha}}), \frac{1}{\sqrt{z}}(\bar{y}^{\dot\alpha} + \lambda x^{\alpha\dot\alpha} y_\alpha) \bigg)\,,
\end{equation}
where
\begin{equation}
p_{0\mu} C_0(Y|x) = C_0(Y|x)p_{0\mu} := -i \frac{\partial}{\partial y^\mu} C_0(Y|x)\,
\end{equation}
and the subscript 0 indicates that $p_{0\alpha}$ acts on the initial fields $C_{0tw}$.

The free parameters $C_{0tw}(Y)$ and $C_{0adj}(Y)$ describe all higher derivatives of the fields $C_{tw}(Y|x_0)$ and $C_{adj}(Y|x_0)$ at the point $x_0$ with $g(Y |x_0) = I$. Formula (\ref{e: general solution standard}) plays a role of the covariantized Taylor expansion reconstructing
generic solution in terms of its derivatives at $x = x_0$, which is a standard property of unfolded dynamics.

The solutions $C_{tw}(Y|x)$ and $C_{adj}(Y|x)$ have a different behavior in the limit $z\rightarrow 0$, \ie approaching the boundary of $AdS_4$. Since the initial data $C_{0tw}(Y)$ and $C_{0adj}(Y)$ are analytic functions of $Y$, the field $C_{tw}(Y|x)$ tends to zero and the field $C_{adj}(Y|x)$ blows up except the case $C_{0adj}(Y) = C_{0adj}$.

These results admit a straightforward generalization to the $B_2$ model. Obviously, in the stereographic coordinates $AdS_4$ connection (\ref{e:AdS connection}) can be represented as
\begin{equation}
    \Omega_{AdS}(Y_1, Y_2 | x) = G^{-1}(Y_1, Y_2 | x)*\d_x G(Y_1, Y_2 | x)\,,
\end{equation}
where
\begin{equation}
    G(Y_1, Y_2 | x)= g(y_1, \bar{y}_1 | x)g(y_2, \bar{y}_2 | x)\,.
\end{equation}
Then the solutions of equations (\ref{e:cov1})-(\ref{e:cov8}) have a form
\begin{equation}\label{e: solution of CovConst}
    C(Y_1, Y_2,I; \hat{K} | x) = G^{-1}*C_0(Y_1,Y_2,I)*\pi(G)*\hat{K}\,,
\end{equation}
where $\pi(\bullet)$ is an automorphism of the $B_2$ HS algebra induced by the dressed Klein operators $\hat{K}$. It is easy to see that the solutions of (\ref{e:cov1})-(\ref{e:cov6}) behave as products of the standard $C_{tw}(Y|x)$ and $C_{adj}(Y|x)$. For example, the solution to the eq.(\ref{e:cov2}) is
\begin{multline}
    C(Y_1, Y_2,I; \hat{K} | x) = z \exp[-i \lambda x^{\alpha\dot\alpha}y_{\alpha1} \bar{y}_{\dot{\alpha}1} + i(\sqrt{z}-1)y^\alpha_1 p_{0\alpha 1} + i(\sqrt{z}-1)\bar{y}^{\dot{\alpha}}_1 \bar{p}_{0\dot{\alpha}1} -i \lambda x^{\alpha\dot\alpha}p_{0\alpha 1}\bar{p}_{0\dot{\alpha} 1}]\\
    C_{0 tw\otimes adj}\bigg(y_1,\bar{y}_1,\frac{1}{\sqrt{z}}(y^\alpha_2 + \lambda x^{\alpha\dot\alpha}\bar{y}_{\dot{\alpha}2}), \frac{1}{\sqrt{z}}(\bar{y}^{\dot\alpha}_2 + \lambda x^{\alpha\dot\alpha} y_{\alpha 2}),I; \hat{K}\bigg)\,.
\end{multline}

Therefore, the boundary condition
\begin{equation}\label{e: boundary condition}
    \lim_{z\rightarrow 0} \frac{1}{\sqrt{z}}C(Y_1, Y_2,I; \hat{K} | x) = 0
\end{equation}
restricts the adjoint parts of (\ref{e:cov2})-(\ref{e:cov3}) and (\ref{e:cov5})-(\ref{e:cov6}) to constants, yielding unitarizable $B_2$ modules.

Note that equations (\ref{e:cov7}) and (\ref{e:cov8}) cannot be represented as products of standard HS modules by a change of variables since matrices $P_{\pm} =\frac{1}{2}\begin{pmatrix}
1 & \pm 1 \\
\mp 1 & 1
\end{pmatrix}$ cannot be simultaneously diagonalized over real numbers (it is easy to see that $(R\bar{R}{}^T)^2 \neq \mathbb{1}$). Thence, the question whether these infinite dimensional modules are unitarizable requires a thorough analysis. We tackle this question by the generalization of Bogolyubov transform method used in the standard HS theory \cite{Vasiliev:2001zy}, application of a plane wave solution \cite{Bolotin:1999fa} and a special ansatz.

\subsection{Fock Space Realization}\label{Fock Space Realization}

Following \cite{Vasiliev:2001zy}, we replace $Y_A^n$ oscillators by a doubled set of oscillators and reformulate linear equations in terms of Fock module valued fields. Consider the associative star-product algebra with $16$ generating elements $a_{1,2}{}_{A}$ and $b_{1,2}{}^{B}$ $(A,B \in \{1,...,4\})$. The particular star product realization of the algebra of oscillators we use represents the totally symmetric (\textit{i.e.}, Weyl) ordering.
\begin{multline}\label{e:star product ab}
    (f*g)(a,b) = \frac{1}{\pi^8}\int d^4 u_{1,2}d^4 v_{1,2} d^4 s_{1,2} d^4 t_{1,2} f(a+u,b+t)g(a+s,b+v) \times \\ \times exp\bigg(2 s_{1A} t_{1}^{A} - 2 u_{1A}v_1^{A} + 2 s_{2A} t_{2}^{A} - 2 u_{2A}v_2^{A}\bigg)\,.
\end{multline}

The Moyal star product (\ref{e:star product ab}) gives rise to the commutation relations
\begin{equation}\label{e: commutators ab (AB)}
    [a_{i}{}_{A},b_j{}^{B}]_* = \delta_{ij}\delta_{A}{}^{B}\,, \quad [a_{i}{}_{A},a_{i}{}_{B}]_* = 0\,, \quad [b_i{}^{A},b_j{}^{B}]_* = 0
\end{equation}
with $[f,g]_* = f*g - g*f$. From (\ref{e:star product ab}) it is easy to derive
\begin{gather}
    a_{iA} * = a_{iA} + \frac{1}{2}\frac{\partial}{\partial b_i^A}\,, \quad b_i^A * = b_i^A - \frac{1}{2}\frac{\partial}{\partial a_{iA}}\,, \label{e: ab*}\\
    * a_{iA} = a_{iA} - \frac{1}{2}\frac{\partial}{\partial b_i^A}\,, \quad * b_i^A = b_i^A + \frac{1}{2}\frac{\partial}{\partial a_{iA}}\,. \label{e: *ab}
\end{gather}

The Lie algebra $gl(4,\mathds{C}) \oplus gl(4,\mathds{C})$ is spanned by the bilinears
\begin{equation}
    T_i{}_{A}{}^{B} = a_{i}{}_{A}b_i{}^{B}  \equiv \frac{1}{2}(a_{i}{}_{A}*b_i{}^{B} + b_i{}^{B}*a_{i}{}_{A})\,.
\end{equation}
The central elements are
\begin{equation}
    H_i = a_{i}{}_{A}b_i{}^{A}\equiv \frac{1}{2}(a_{i}{}_{A}*b_i{}^{A} + b_i{}^{A}*a_{i}{}_{A})\,.
\end{equation}

Factorization by the central elements $H_i$ yields the Lie algebra $sl(4,\mathds{C}) \oplus sl(4,\mathds{C})$ spanned by
\begin{equation}
    t_i{}_{A}{}^{B} = (a_{i}{}_{A}b_i{}^{B} - \frac{1}{4}\delta_{A}{}^{B}H_i )\,.
\end{equation}

The $su(2,2) \oplus su(2,2)$ real form of $sl(4,\mathds{C}) \oplus sl(4,\mathds{C})$ results from the reality conditions
\begin{equation}\label{e: reality condition (AB)}
    \bar{a}_{i}{}_{A} = b_{i}{}^{B}C_{BA}\,, \quad \bar{b}_{i}{}^{A} = C^{AB}a_{i}{}_{B}\,,
\end{equation}
where bar denotes complex conjugation and $C_{AB} = - C_{BA}$ and $C^{AB} = - C^{BA}$ are real antisymmetric matrices obeying
\begin{equation}
    C_{AC}C^{BC} = \delta_{A}{}^{B}\,.
\end{equation}
Note that
\begin{equation}
    su(2,2) \oplus su(2,2) \subset sp(8) \oplus sp(8)\,.
\end{equation}
with  $sp(8) \oplus sp(8) $ spanned by various bilinears of $a_i$ and $b_i$ at $i=1$ or $2$.

In the sequel we set
\begin{equation}
    C^{AB} = \begin{pmatrix}
\varepsilon^{\alpha \dot\beta} & 0 \\
0 & \varepsilon^{\dot\alpha \beta}
\end{pmatrix}\,, \quad \varepsilon^{12} = - \varepsilon^{21} = 1\,, \quad \varepsilon^{11} = \varepsilon^{22} = 0\,
\end{equation}
splitting generating elements $a_{1,2}{}_{A}$ and $b_{1,2}{}^{B}$ into the pairs of two-component spinors $a_{1,2}{}_\alpha$, $b_{1,2}{}^\alpha$, $\tilde{a}_{1,2}{}_{\dot\alpha}$, $\tilde{b}_{1,2}{}^{\dot\alpha}$. Then commutators (\ref{e: commutators ab (AB)}) transform to
\begin{equation}\label{e: commutators ab}
    [a_{i\alpha}, b_j^\beta]_* = \delta_{ij}\delta_\alpha^\beta\,, \quad [\tilde{a}_{i \dot\alpha}, \tilde{b}_{j}{}^{\dot\beta}]_* = \delta_{ij}\delta_{\dot\alpha}^{\dot\beta}\,
\end{equation}
with the other commutation relations being zero. The  conjugation rules
(\ref{e: reality condition (AB)}) read as
\begin{equation}
    \bar{a}_{i\alpha} = \tilde{b}_{i \dot\alpha}\,, \quad \bar{b}{}_{i}^{\alpha} = \tilde{a}{}_{i}^{\dot\alpha}\,, \quad \bar{\tilde{a}}{}_{i\dot\alpha} = b_{i\alpha}\,, \quad \bar{\tilde{b}}{}_i^{\dot\alpha} = a_i^\alpha\,. \label{e: complex conjugation}
\end{equation}

Let us now introduce vacua $\pi^i{}_p$ for each set of $a_{iA}, b_i{}^B$, by imposing the following conditions:
\begin{gather}
    a_{i\alpha}*\pi^i{}_1=0=\pi^i{}_1*\tilde{a}_{i\dot\alpha}\,, \quad \tilde{b}_{i}^{\dot\alpha}*\pi^i{}_1=0=\pi^i{}_1*b_{i}^{\alpha}\,, \\ b_{i}^{\alpha}*\pi^i{}_2=0=\pi^i{}_2*\tilde{b}_{i}^{\dot\alpha}\,, \quad \tilde{a}_{i\dot\alpha}*\pi^i{}_2=0=\pi^i{}_2*a_{i\alpha}\,, \\ a_{i\alpha}*\pi^i{}_3=0=\pi^i{}_3*\tilde{b}_{i}^{\dot\alpha}\,, \quad \tilde{a}_{i\dot\alpha}*\pi^i{}_3=0=\pi^i{}_3*b_{i}^{\alpha}\,,\\ b_{i}^{\alpha}*\pi^i{}_4=0=\pi^i{}_4*\tilde{a}_{i\dot\alpha}\,,\quad \tilde{b}_{i}^{\dot\alpha}*\pi^i{}_4=0=\pi^i{}_4*a_{i\alpha}\,.
\end{gather}

Such vacua can be realized as elements of the star-product algebra using (\ref{e: ab*}) and (\ref{e: *ab}):
\begin{gather}
    \pi^i{}_1 = \exp{-2a_{i\alpha} b_i{}^\alpha + 2\tilde{a}_{i \dot\alpha} \tilde{b}_i{}^{\dot\alpha}}\,, \quad \pi^i{}_2 = \exp{2a_{i\alpha} b_i{}^\alpha - 2\tilde{a}_{i \dot\alpha} \tilde{b}_i{}^{\dot\alpha}}\,, \label{e: Fock vacuum 1}\\
    \pi^i{}_3 = \exp{-2a_{i\alpha} b_i{}^\alpha - 2\tilde{a}_{i \dot\alpha} \tilde{b}_i{}^{\dot\alpha}}\,, \quad \pi^i{}_4 = \exp{2a_{i\alpha} b_i{}^\alpha + 2\tilde{a}_{i \dot\alpha} \tilde{b}_i{}^{\dot\alpha}}\,. \label{e: Fock vacuum 2}
\end{gather}

This yields an explicit realization of the Fock module with states created from a particular pair of vacua, for instance, $\pi^1{}_1$ and $\pi^2{}_1$:
\begin{equation}
    \ket{C^{11}} =C^{11}(b_1, \tilde{a}_1, b_2, \tilde{a}_2) *\pi^1{}_1 \pi^2{}_1 = C^{11}(2b_1, 2\tilde{a}_1, 2b_2, 2\tilde{a}_2) \pi^1{}_1 \pi^2{}_1 \,. \label{FockModule11}
\end{equation}
As will be explained bellow, all choices of Fock-space vacuum projectors are equivalent in the context of studying properties of $AdS_4$ modules.

Let us now introduce the $su(2,2)$ generators in a canonical way:
\begin{gather}
    L^i{}_\alpha{}^\beta = a_{i \alpha}b_i{}^\beta - \frac{1}{2} \delta_\alpha{}^\beta a_{i \gamma}b_i{}^\gamma \,, \quad P^i {}_\alpha {}^{\dot\beta} = a_{i \alpha} \tilde{b}_i{}^{\dot\beta} \,,\\
    \bar L^i{}_{\dot\alpha}{}^{\dot\beta} = \tilde{a}_{i \dot\alpha}\tilde{b}_i{}^{\dot\beta} - \frac{1}{2} \delta_{\dot\alpha}{}^{\dot\beta} \tilde{a}_{i \dot\gamma}\tilde{b}_i{}^{\dot\gamma} \,, \quad K^i{}_{\dot\alpha}{}^\beta = \tilde{a}_{i \dot\alpha} b_i{}^\beta\,.\\
    D^i = \frac{1}{2}  (a_{i \alpha} b_i{}^{\alpha} - \tilde{a}_{i \dot{\alpha}} \tilde{b}_i{}^{\dot{\alpha}})\,,
\end{gather}
and the central elements
\begin{equation}
    H^i = \frac{1}{2}  (a_{i \alpha} b_i{}^{\alpha} + \tilde{a}_{i \dot{\alpha}} \tilde{b}_i{}^{\dot{\alpha}})
\end{equation}
that correspond to the helicity operators.

The $AdS_4$ connection can be introduced via an embedding of $AdS_4$ algebra into $su(2,2)\oplus su(2,2)$
\begin{equation}\label{e: connection su(2,2)}
    \omega_0 = \omega_0 {}^\alpha{}_\beta (L^1{}_\alpha{}^\beta + L^2{}_\alpha{}^\beta) + \bar{\omega}_0 {}^{\dot\alpha}{}_{\dot\beta}( \bar{L}^1{}_{\dot\alpha}{}^{\dot\beta} + \bar{L}^2 {}_{\dot\alpha}{}^{\dot\beta}) + e_0 {}^\alpha {}_{\dot\beta} (P^1 {}_\alpha {}^{\dot\beta} + P^2 {}_\alpha {}^{\dot\beta} + K^1 {}^{\dot\beta} {}_\alpha + K^2 {}^{\dot\beta} {}_\alpha)\,.
\end{equation}
It obeys the flatness condition
\begin{equation}
    \d_x \omega_0 + \omega_0 \wedge* \omega_0 = 0\,.
\end{equation}
That the connection (\ref{e: connection su(2,2)}) is flat implies that $\omega_0 {}^\alpha{}_\beta$, $\bar{\omega}_0 {}^{\dot\alpha}{}_{\dot\beta}$ and $e_0 {}^\alpha {}_{\dot\beta}$ describe $AdS_4$ Lorentz connection and vierbein, respectively. Note that the generator $P^i {}_\alpha {}^{\dot\beta} + K^i {}^{\dot\beta} {}_\alpha$ describes
the embedding of the $AdS_4$ translations (transvections) into the conformal algebra $su(2, 2)$.

Note that vacua $\pi^i{}_p$ are bi-Lorentz invariant
\begin{equation}
    L^j{}_\alpha{}^\beta * \pi^i{}_p = 0 = \pi^i{}_p * L^j{}_\alpha{}^\beta\,, \quad \bar{L}^j {}_{\dot\alpha}{}^{\dot\beta} * \pi^i{}_p = 0 = \pi^i{}_p * \bar{L}^j {}_{\dot\alpha}{}^{\dot\beta}\,
\end{equation}
and eigenvectors of $D^i$ and $H^i$. Eigenvalues for the vacuum $\pi^1{}_1\pi^2{}_1$ are
\begin{equation}
    D^i*\pi^1{}_1\pi^2{}_1 = \pi^1{}_1\pi^2{}_1\,, \quad H^i*\pi^1{}_1\pi^2{}_1 = 0\,.
\end{equation}

The $M_{tw\otimes tw}$ module (\ref{e:cov4}) can then be obtained by subjecting the Fock module (\ref{FockModule11}) to equations of the form:

\begin{equation}
    \d_x \ket{C^{11}} + \omega_0*\ket{C^{11}} = 0\,.
\end{equation}
Indeed, after some calculation it yields
\begin{equation}\label{e: twist-twist ab}
    D_L C^{11}(2b_1, 2\tilde{a}_1, 2b_2, 2\tilde{a}_2) + e_0{}^{\alpha \dot\alpha}\sum_{i = 1}^2\left(4\tilde{a}_{i \dot\alpha}b_{i\alpha} - \frac{1}{4}\frac{\partial^2}{\partial b_i^\alpha \partial \tilde{a}_{i}^{\dot\alpha}}\right)C^{11}(2b_1, 2\tilde{a}_1, 2b_2, 2\tilde{a}_2) = 0\,,
\end{equation}
where
\begin{equation}
    D_L = \d_x + \omega_0 {}^\alpha{}_\beta \sum_{i = 1}^{2}\bigg(b_i^\beta \frac{\partial}{\partial b_i^\alpha} - \frac{1}{2}\delta_{\alpha}{}^\beta b_i^\gamma \frac{\partial}{\partial b_i^\gamma}\bigg) - \bar{\omega}_0 {}^{\dot\alpha}{}_{\dot\beta}\sum_{i=1}^{2}\bigg(\tilde{a}_{i\dot\alpha}\frac{\partial}{\partial \tilde{a}_{i\dot\beta}} - \frac{1}{2}\delta_{\dot\alpha}{}^{\dot\beta}\tilde{a}_{i\dot\gamma}\frac{\partial}{\partial \tilde{a}_{i\dot\gamma}}\bigg)\,.
\end{equation}

The module (\ref{e:cov4}) is related to  (\ref{e: twist-twist ab}) via the substitution
$2b_i \rightarrow y_i$,  $2\tilde{a}_i \rightarrow \bar{y}_i$.

To reproduce the other equations in (\ref{e:cov1})-(\ref{e:cov8}), however, one has to apply one of the automorphisms $\rho$ of the $AdS_4$ algebra that can be associated with the action of Klein operators $\hat{k}$, $\hat{\bar k}$. It turns out that these $AdS_4$ automorphisms are also automorphisms of the star product algebra (\ref{e: commutators ab}). To map module $M_{tw\otimes tw}$ to any other $B_2$ module one has to apply the automorphism $\rho$ to the $AdS_4$ algebra and keep the underlying Fock module $\ket{C^{11}}$ unchanged meaning that we stick to the Fock module generated from the $\pi^1{}_1 \pi^2{}_1$ vacuum. Note that the entire analysis can be carried out over any other Fock vacuum (\ref{e: Fock vacuum 1}), (\ref{e: Fock vacuum 2}), leading to the same results since the transition from the description of modules in terms of one Fock module to  another can be induced by a suitable automorphism. The idea of connecting field-theoretically different modules by automorphisms is inspired by the observation that in a standard HS theory formulated in terms of $Y^A$-oscillators one can obtain an adjoint module from a twisted-adjoint one via the action of an automorphism of the $AdS_4$ algebra.
However, the mapping procedure in a $Y^A$-oscillator setup is highly complicated due to the involvement of half Fourier transform that maps polynomials into derivatives of $\delta$ functions and vice versa. Fortunately, in a $\{a_{i}{}_{A}\,,b_{i}{}^{B}\}$-setup the mapping procedure operates with polynomials only. The required automorphisms of $\{a_{i}{}_{A}\,,b_{i}{}^{B}\}$ star-product algebra have the following form (trivial action is omitted in each case):

\begin{align}
   \Biggl\{ \rho_i(a_{i\alpha}) =  b_{i\alpha} \,, \quad \rho_i(b_{i}^{\alpha}) =  a_{i}^{\alpha} \,, &\quad \bar{\rho}_i(\tilde{a}_{i \dot\alpha}) =  \tilde{b}_{i \dot\alpha} \, \quad \bar{\rho}_i(\tilde{b}_{i}^{\dot\alpha}) =  \tilde{a}_{i}^{\dot\alpha} \Biggr\}\nonumber\\
   &\Updownarrow \nonumber\\
   \{\hat{k}_i* y_i^\alpha = -y_i^\alpha*\hat{k}_i\,, &\quad \hat{\bar k}_i*\bar{y}_i^{\dot\alpha} = -\bar{y}_i^{\dot\alpha} *\hat{\bar k}_i\}\,,
\end{align}

\begin{align}
    \Biggl\{ \psi_+ (a_{1 \alpha}) = \frac{1}{2} (b_1 + b_2 + a_1 - a_2)_\alpha &\,, \quad \psi_+ (a_{2 \alpha}) = \frac{1}{2} (b_1 + b_2 + a_2 - a_1)_\alpha \,, \nonumber\\
    \quad \psi_+ (b_{1}^{\alpha}) = \frac{1}{2} (a_1 + a_2 + b_1 - b_2)^\alpha &\,, \quad \psi_+ (b_{2}^{\alpha}) = \frac{1}{2} (a_1 + a_2 +b_2 - b_1)^\alpha \Biggr\}\nonumber\\
    &\Updownarrow \nonumber\\ \{\hat{k}_{12}^+*y_1^\alpha=-y_2^\alpha*\hat{k}_{12}^+&\,, \quad \hat{k}_{12}^+*y_2^\alpha=-y_1^\alpha*\hat{k}_{12}^+\}\,,
\end{align}

\begin{align}
    \Biggl\{\psi_- (a_{1 \alpha}) = \frac{1}{2} (b_1 - b_2 + a_1 + a_2)_\alpha &\,, \quad \psi_- (a_{2 \alpha}) = \frac{1}{2} (b_2 - b_1 + a_1 + a_2)_\alpha \,, \nonumber\\
    \quad \psi_- (b_{1}^{\alpha}) = \frac{1}{2} (a_1 - a_2 + b_1 + b_2)^\alpha &\,, \quad \psi_- (b_{2}^{\alpha}) = \frac{1}{2} (a_2 - a_1 + b_2 +  b_1)^\alpha \Biggr\}\nonumber\\
    &\Updownarrow \nonumber\\ \{\hat{k}_{12}*y_1^\alpha=y_2^\alpha*\hat{k}_{12}&\,, \quad \hat{k}_{12}*y_2^\alpha=y_1^\alpha*\hat{k}_{12}\}\,,
\end{align}
and an additional practically useful automorphism
\begin{equation}
    \chi(a_{1 \alpha}) = \frac{1}{\sqrt{2}} (a_{1\alpha} + a_{2\alpha}) \,, \quad \chi(a_{2 \alpha}) = \frac{1}{\sqrt{2}} (a_{2\alpha} - a_{1\alpha}) \,, \quad \chi(b_{1}^{\alpha}) = \frac{1}{\sqrt{2}} (b_{1}^{\alpha} + b_{2}^{\alpha}) \,, \quad \chi(b_{2}^{\alpha}) = \frac{1}{\sqrt{2}} (b_{2}^{\alpha} - b_{1}^{\alpha})\,,
\end{equation}
that does not correspond to any Klein operator. Instead, the automorphisms $\chi$ and $\bar\chi$ generate a change of variables (\ref{e:change of variables y}) that relates conjugacy classes $\mathcal{R}_1$ and $\mathcal{R}_2$.

The involutive automorphisms $\rho_i$, $\psi_+$, $\psi_-$ and their complex conjugated leave  $D_L$ invariant, \ie
\begin{equation}
    \rho\bigg(\omega_0{}^{\alpha\beta}[L^1_{\alpha\beta}+L^2_{\alpha\beta}]\bigg) = \omega_0{}^{\alpha\beta}[L^1_{\alpha\beta}+L^2_{\alpha\beta}]\,, \quad \forall \rho \in \{\rho_i, \psi_+, \psi_-\}\,,
\end{equation}
while non-trivially transforming the $e_0{}^{\alpha \dot\alpha}\sum_{i=1}^2(P^i_{\alpha\dot\alpha}+K^i_{\alpha\dot\alpha})$ term of the connection. Therefore, for any composition of the Klein-related automorphisms $\rho$
\begin{equation}\label{e: all ab equations}
    \d_x \ket{C^{11}} + \rho(\omega_0)*\ket{C^{11}} = 0
\end{equation}
is a new equation imposed on the Fock module $\ket{C^{11}}$ which means that we obtain some other $B_2$ module. These new equations can be identified with equations (\ref{e:cov1})-(\ref{e:cov8}). For example, the automorphism $\rho_1$ leads to the module $M_{adj\otimes tw}$ described by the equation (\ref{e:cov3})
\begin{align}
    \d_x \ket{C^{11}} &+ \rho_1(\omega_0)*\ket{C^{11}} = 0\\
    &\Updownarrow\nonumber\\
    \bigg(D_L + e_0{}^{\alpha \dot{\alpha}}\bigg[\tilde{a}_{1\dot{\alpha}}\frac{\partial}{\partial b_1^\alpha} - b_{1\alpha}\frac{\partial}{\partial \tilde{a}_{1}^{\dot{\alpha}}} &+ 4\tilde{a}_{2\dot{\alpha}}b_{2\alpha}-\frac{1}{4}\frac{\partial^2}{\partial b_2^\alpha \tilde{a}_{2}^{\dot{\alpha}}}\bigg] \bigg)C^{11}(2b_1, 2\tilde{a}_1, 2b_2, 2\tilde{a}_2) = 0\,.
\end{align}

The realization of linear equations (CHS modules) in terms of the Fock modules can be easily extended to the case of general $B_p$ models. Indeed, to describe $B_p$ modules we should consider star-product algebra $\{a_{i}{}_{A}\,,b_{i}{}^{B}\}$ with $i\in\{1,...,p\}$ and extend all summation above over the index $i$ from the range $\{1,2\}$ to $\{1,...,p\}$. Then automorphisms $\rho_i$ reproduce an action of Klein operators $\hat{k}_i$ and automorphisms $\psi_\pm{}_{ij}$ give an action of Klein operators $\hat{k}_{ij}$ and $\hat{k}^+_{ij}$ (replace $1,2$ with $i,j$ in formulas for $\psi_\pm$).

Considering compositions $\rho_1 \psi_+$ and $\rho_2 \psi_+$, one arrives at equations associated with two entangled modules (\ref{e:cov7}) and (\ref{e:cov8}).  Indeed,
\begin{align}
    \d_x \ket{C^{11}} &+ \rho_1\psi_+(\omega_0)*\ket{C^{11}} = 0\\
    &\Updownarrow \nonumber\\
    \bigg(D_L + \frac{1}{2}e_0{}^{\alpha \dot{\alpha}}\bigg[b_{1 \alpha} \frac{\partial}{\partial \tilde{a}_2^{\dot{\alpha}}}-b_{1 \alpha} \frac{\partial}{\partial \widetilde{a}_1^{\dot{\alpha}}}-b_{2 \alpha} &\frac{\partial}{\partial \tilde{a}_1^{\dot{\alpha}}}-b_{2 \alpha} \frac{\partial}{\partial \tilde{a}_2^{\dot{\alpha}}}+\tilde{a}_{1 \dot{\alpha}} \frac{\partial}{\partial b_1^\alpha}+\widetilde{a}_{1 \dot{\alpha}} \frac{\partial}{\partial b_2^\alpha}+\tilde{a}_{2 \dot{\alpha}} \frac{\partial}{\partial b_2^\alpha}-\widetilde{a}_{2 \dot{\alpha}} \frac{\partial}{\partial b_1^\alpha}\bigg] + \nonumber\\
    +\frac{1}{2}e_0{}^{\alpha \dot{\alpha}}\bigg[4 \tilde{a}_{1 \dot{\alpha}} b_{1 \alpha}&-4 \tilde{a}_{1 \dot{\alpha}} b_{2 \alpha}+4 \tilde{a}_{2 \dot{\alpha}} b_{1 \alpha}+4 \tilde{a}_{2 \dot{\alpha}} b_{2 \alpha} -\nonumber\\
    -\frac{1}{4} \frac{\partial^2}{\partial b_1^\alpha \partial \tilde{a}_1^{\dot\alpha}}+\frac{1}{4} \frac{\partial^2}{\partial b_2^\alpha \partial \tilde{a}_1^{\dot\alpha}} -\frac{1}{4} \frac{\partial^2}{\partial b_1^\alpha \partial \tilde{a}_2^{\dot\alpha}}-&\frac{1}{4} \frac{\partial^2}{\partial b_2^\alpha \partial \tilde{a}_2^{\dot\alpha}}\bigg] \bigg)C^{11}(2b_1, 2\tilde{a}_1, 2b_2, 2\tilde{a}_2)=0\,.\label{e: ab entangled eq}
\end{align}
In each case, except for the entangled modules, the vacuum $\pi^1{}_1 \pi^2{}_1$ and the Fock module $\ket{C^{11}}$ diagonalize dilation $\rho(D)=\rho(D^1)+\rho(D^2)$ and helicity $\rho(H)=\rho(H^1)+\rho(H^2)$ operators. For example,
\begin{gather}
    D*\ket{C^{11}}=\frac{1}{2} \bigg(b^\alpha_i\frac{\partial}{\partial b^\alpha_i}C^{11} + \tilde{a}_{i \dot\alpha}\frac{\partial}{\partial \tilde{a}_{i \dot\alpha}}C^{11} + 4C^{11}\bigg)\pi^1{}_1 \pi^2{}_1\,, \\
    \rho_1(D)*\ket{C^{11}} = \frac{1}{2} \bigg(- b^\alpha_1\frac{\partial}{\partial b^\alpha_1}C^{11} + b^\alpha_2\frac{\partial}{\partial b^\alpha_2}C^{11} + \tilde{a}_{i \dot\alpha}\frac{\partial}{\partial \tilde{a}_{i \dot\alpha}}C^{11} + 2 C^{11}\bigg)\pi^1{}_1 \pi^2{}_1\,, \\
    \psi_+(D)*\ket{C^{11}} = \frac{1}{2} \bigg(- 2b^\alpha_1\frac{\partial}{\partial b^\alpha_2}C^{11} - 2b^\alpha_2\frac{\partial}{\partial b^\alpha_1}C^{11} + \tilde{a}_{i \dot\alpha}\frac{\partial}{\partial \tilde{a}_{i \dot\alpha}}C^{11} + 2 C^{11}\bigg)\pi^1{}_1 \pi^2{}_1\,,\\
    \rho_1\psi_+(D)*\ket{C^{11}} = \frac{1}{2} \bigg(4b_{2\alpha}b_1^{\alpha}C^{11}+\frac{1}{4}\frac{\partial^2}{\partial b_1^\alpha \partial b_{2\alpha}}C^{11} + \tilde{a}_{i \dot\alpha}\frac{\partial}{\partial \tilde{a}_{i \dot\alpha}}C^{11} + 2 C^{11}\bigg)\pi^1{}_1 \pi^2{}_1\,.\label{e: noneigen vacuum ab}
\end{gather}
Note that the vacuum $\pi^1{}_1 \pi^2{}_1$ does not diagonalize operators $\rho_1\psi_+(D)$ and $\rho_1\psi_+(H)$, but $\exp(\pm 4 b_{1\alpha}b_2^{\alpha})\pi^1{}_1 \pi^2{}_1$ diagonalizes them both.

For the entangled module the exponential ansatz (\ref{def:expAnsatz}) becomes clearer in variables $\{a_{iA}, b_j^B\}$. While the ansatz does not diagonalize operators $\rho(D)$ and $\rho(H)$, it reduces the entangled equation generated by the automorphism $\rho_1 \psi_+$ to take the form of the equation on the product of twisted-adjoint modules in new variables, with conjugation rules being violated.
The ansatz has a form
\begin{equation}
    C^{11}(2b_1, 2\tilde{a}_1, 2b_2, 2\tilde{a}_2) = \exp\bigg(4 b_{1\alpha}b_2^{\alpha} - 4\tilde{a}_{1\dot\alpha}\tilde{a}_2^{\dot\alpha}\bigg)\tilde{C}^{11}(2b_1, 2\tilde{a}_1, 2b_2, 2\tilde{a}_2)\, \label{e: ansatz 1 ab}
\end{equation}
and equation (\ref{e: ab entangled eq}) turns into
\begin{multline}
    \bigg(D_L + \frac{1}{2}e_0{}^{\alpha \dot{\alpha}}\bigg[8 \tilde{a}_{1 \dot{\alpha}} (b_{1 \alpha} - b_{2 \alpha})+8 \tilde{a}_{2 \dot{\alpha}}( b_{1 \alpha} + b_{2 \alpha}) - \\
    - \frac{1}{4} \frac{\partial^2}{ \partial \tilde{a}_1^{\dot\alpha} (\partial b_1^\alpha - \partial b_2^\alpha)} + \frac{1}{4} \frac{\partial^2}{ \partial \tilde{a}_2^{\dot\alpha} (\partial b_1^\alpha + \partial b_2^\alpha)}
    \bigg] \bigg)\tilde{C}^{11}(2b_1, 2\tilde{a}_1, 2b_2, 2\tilde{a}_2) = 0\,.
\end{multline}

Here, once again, a change to $b^\pm_\alpha =  b_{1 \alpha} \pm b_{2 \alpha}$, while keeping $\tilde{a}_{i \dot{\alpha}}$ unchanged (thus violating conjugation rules) leaves us with the equation for the product of two twisted-adjoint modules. Compared to $Y^A$ variables the exponential function from the ansatz behaves starkly different in the $\{a_{iA}, b_j^B\}$, such that its star product square
\begin{equation}
    \exp\bigg(b_{1\alpha}b_2^{\alpha} - \tilde{a}_{1\dot\alpha}\tilde{a}_2^{\dot\alpha}\bigg) * \exp\bigg( b_{1\alpha}b_2^{\alpha} - \tilde{a}_{1\dot\alpha}\tilde{a}_2^{\dot\alpha}\bigg) = \exp\bigg(2 b_{1\alpha}b_2^{\alpha} - 2\tilde{a}_{1\dot\alpha}\tilde{a}_2^{\dot\alpha}\bigg)\,
\end{equation}
is a well-defined expression, as in practical terms star product only acts as a point-wise product in this case. While not changing the unitarizability of the module, this ansatz is still important for further analysis of the spectrum of the theory.

Overall, the action of the Klein-related automorphisms $\rho$ on the $AdS_4$ connection $\omega$ reproduces all $B_2$ modules. However, the resulting modules are not unitary as a result of the Lorentz invariance of the vacua $\pi^i{}_p$. The dependence on the space-time coordinates of the elements of the field $\ket{C^{11}}$ is completely determined by the equation (\ref{e: all ab equations}) in terms of its value at any fixed point $x_0$. This means that the module $\ket{C^{11}(x_0)}$ contains the complete information on the on-mass-shell dynamics of the $4d$ field. Therefore, the question of unitarizability reduces to the (non-)existence of transformation between the module $\ket{C^{11}(x_0)}$ and some unitary $su(2, 2)$ module.

To analyze the unitarizability of such modules, we shall use the explicit construction in the twisted-adjoint case of standard $4d$ HS theory presented in \cite{Vasiliev:2001zy} and Klein-related automorphisms defined above. To that end, we introduce a new set of oscillators $e^i_{\nu A}$ and $f^i{}_A{}^\nu$ such that

\begin{gather}
    [e^i_{\nu A}, e^j_{\mu B}]_* = 0\,, \quad [f^i{}_A{}^\nu, f^j{}_B{}^\mu]_* = 0\,, \quad [e^i_{\nu A}, f^j{}_B{}^\mu]_* = \delta^{ij}\delta_{\nu}^{\mu} K_{AB}\,,
\end{gather}
where $K_{AB} = \begin{pmatrix}
1 & 0 \\
0 & -1
\end{pmatrix}$ and $i\,,\nu\,, A\in\{1\,,2\}$. The oscillators obey the Hermiticity conditions
\begin{equation}
    (e^i_{\nu A})^\dagger = f^i{}_A{}^\nu\,.
\end{equation}

Note that
\begin{equation}\label{e:change of variables ef}
    e^\pm_{\nu A} = \frac{1}{\sqrt{2}}(e^1_{\nu A} \pm e^2_{\nu A})\,, \quad f^\pm{}_A{}^\nu = \frac{1}{\sqrt{2}}(f^1{}_A{}^\nu \pm f^2{}_A{}^\nu)
\end{equation}
satisfy
\begin{gather}
    [e^\pm_{\nu A}, e^\pm_{\mu B}]_* = [e^\pm_{\nu A}, e^\mp_{\mu B}]_* = 0\,, \quad [f^\pm{}_A{}^\nu, f^\pm{}_B{}^\mu]_* = [f^\pm{}_A{}^\nu, f^\mp{}_B{}^\mu]_* = 0\,, \\
    [e^\pm_{\nu A}, f^\pm{}_B{}^\mu]_* = [e^\mp_{\nu A}, f^\mp{}_B{}^\mu]_* = \delta_{\nu}^{\mu} K_{AB}\,, \quad [e^\pm_{\nu A}, f^\mp{}_B{}^\mu]_* = 0\,,\\
     (e^\pm_{\nu A})^\dagger = f^\pm{}_A{}^\nu\,.
\end{gather}
The transition from oscillators $\{e^i_{\nu A}, f^i{}_A{}^\nu\}$ to $\{e^\pm_{\nu A}, f^\pm{}_A{}^\nu\}$ can be attributed to the action of the automorphism $\chi$.

These oscillators allow us to construct the Lie algebra $su(2,2) \oplus su(2,2) \subset sp(8) \oplus sp(8)$
\begin{gather}
    \tau^i_{A\nu}{}^\mu = f^i{}_A{}^\mu e^i_{\nu A} \quad \text{($A,i = 1, 2$ with no summation over $A,i$)}\label{def:tau}\,, \\
    t^{+i}{}^\mu_\nu = e^i_{\nu 2}f^i{}_1{}^\mu\,, \quad t^{-i}{}^\mu_\nu = e^i_{\nu 1}f^i{}_2{}^\mu\,, \label{def:t,t}\\
    E^i = f^i{}_1{}^\lambda e^i_{\lambda 1} + f^i{}_2{}^\lambda e^i_{\lambda 2}\,, \label{def:E}\\
    H^i = f^i{}_1{}^\lambda e^i_{\lambda 1} - f^i{}_2{}^\lambda e^i_{\lambda 2}\,,
\end{gather}
where $\tau^i_{A\nu}{}^\mu$ generate compact subalgebra $(u(2)\oplus u(2))^i$, non-compact generators are $t^{+i}{}^\mu_\nu$ and $t^{-i}{}^\mu_\nu$, operator $E^i$ can be interpreted as an energy operator in the $i$-th sector and central elements $H^i$ are helicity operators. Recall that we use the Weyl star-product notation, \ie all bilinears listed above
are elements of the star-product algebra. For further analysis, we extract the diagonal subalgebra $su(2,2)$ with generators
\begin{gather}
    \tau_{A\nu}{}^\mu= \tau^1_{A\nu}{}^\mu + \tau^2_{A\nu}{}^\mu\,, \quad t^{\pm}{}^\mu_\nu = t^{\pm1}{}^\mu_\nu + t^{\pm2}{}^\mu_\nu\,,\\
    E = E^1+E^2\,,\quad H = H^1+H^2\,.\label{e: total energy}
\end{gather}

We also introduce the Fock module $F$ constructed from the vacuum $\Pi$ defined as

\begin{equation}\label{e: ef vacuum}
    e^i_{\nu 1} * \Pi = 0\,, \quad f^i{}_2{}^\mu * \Pi = 0\,, \quad  \Pi * e^i_{\nu 2} = 0\,, \quad \Pi *  f^i{}_1{}^\mu = 0\,.
\end{equation}

For the Fock module $F$ to be suitable for the description of physical states as a representation of $su(2,2)$ it must satisfy two conditions:
\begin{itemize}
  \item $F$ is a highest/lowest-weight module meaning the energy $E$ is bounded from above/from below and spins are finite.
  \item $F$ admits an invariant positive-definite Hermitian form, \ie $F$ is a unitary module.
\end{itemize}

The two sets of oscillators $\{a_{iA}, b_{jB}\}$ and $\{e^i_{\nu A}, f^i{}_A{}^\nu\}$ can be related via a Bogolyubov transform
\begin{gather}
    e^i_{11} = \frac{1}{\sqrt{2}} (a_{i1} + i \tilde{a}_{i \dot2}) \,, \quad e^i_{12} = \frac{1}{\sqrt{2}} (a_{i1} - i \tilde{a}_{i \dot2}) \,, \quad e^i_{21} = \frac{1}{\sqrt{2}} (\tilde{a}_{i\dot1} + i a_{i 2}) \,, \quad e^i_{22} = \frac{1}{\sqrt{2}} (\tilde{a}_{i\dot1} - i a_{i 2})\,, \label{e: bogolyubov 1} \\
    f^i{}_1{}^{1} = \frac{1}{\sqrt{2}} (b_{i2} + i \tilde{b}_{i \dot1}) \,, \quad f^i{}_2{}^{1} = \frac{1}{\sqrt{2}} (-b_{i2} + i \tilde{b}_{i \dot1}) \,, \quad f^i{}_1{}^{2} = \frac{1}{\sqrt{2}} (\tilde{b}_{i \dot2} + i b_{i 1}) \,, \quad f^i{}_2{}^{2} = \frac{1}{\sqrt{2}} (- \tilde{b}_{i \dot2} + i b_{i 1}) \,. \label{e: bogolyubov 2}
\end{gather}
Then the Fock vacuum $\Pi$ is realized in terms of the star product algebra as
\begin{equation}
    \Pi = \exp\bigg\{-2 e^1_{\nu1}f^1{}_1{}^\nu -2 e^1_{\nu2}f^1{}_2{}^\nu -2 e^2_{\nu1}f^2{}_1{}^\nu -2 e^2_{\nu2}f^2{}_2{}^\nu\bigg\}\,.
\end{equation}

Bogolyubov transform relates modules $\ket{C^{11}(x_0)} \simeq M_{tw \otimes tw}$ and $F_{tw \otimes tw}$. We shall be using the automorphisms $\rho_i$, $\psi_+$, $\psi_-, \chi$, their complex conjugated and their counterparts on oscillators $\{e^i_{\nu A}, f^i{}_A{}^\nu\}$. This allows us to analyze all emerging modules starting with the product of two twisted-adjoint ones $M_{tw \otimes tw}$. Due to (\ref{e: all ab equations}), each module of the $B_2$ theory has the same underlying Fock module $F$ but different realizations of the algebra $su(2,2)$ in terms of oscillators $\{e^i_{\nu A}, f^i{}_A{}^\nu\}$. Alternatively, after a composition with an automorphism acting on the full equation, it can be viewed as the same realization of $su(2,2)$ algebra acting on different vacua. In other words, application of an automorphism changes the slicing of the underlying Fock module in terms of spin-$s$ submodules of the background isometry algebra. We shall adopt the latter approach for the following section. These modules can be obtained by Klein-related automorphisms $\rho_i, \psi_+, \psi_-$ both for $\{a_{iA}, b_{jB}\}$ and $\{e^i_{\nu A}, f^i{}_A{}^\nu\}$. Since the total spectra of representations remains the same in both oscillator realizations, we can establish correspondence of representations presented in terms of any set of oscillators. The values of Casimir operators can always be the final check. It may happen that in some representations the total energy $E$ or helicity $H$ does not have the vacuum as its eigenvector as was in case of (\ref{e: noneigen vacuum ab}), meaning that the module under consideration is not a highest/lowest-weight module. The unitarity can be straightforwardly checked by inspecting whether the creation and annihilation operators for any particular vacuum are each other's conjugate (\ie bilinear form is positive-definite) and whether the compact generators $\tau^i_{A\nu}{}^\mu$, energy $E^i$ and helicity $H^i$ are Hermitian with $(t^{+}{}^\mu_\nu)^\dagger = t^{-}{}^\mu_\nu$ (\ie bilinear form is invariant). Since we keep the oscillator realization of $su(2,2)$ the same, these conditions are enough to fix conjugation in oscillators $\{e^i_{\nu A}\,,f^i{}_A{}^\nu\}$ as the same in all modules, requiring
\begin{equation}\label{e: conj rules ef}
    (e^i_{\nu A})^\dagger = f^i{}_A{}^\nu\,.
\end{equation}
However, these conjugation rules can lead to the module's vacuum $\rho(\Pi)$ being not self-conjugated which automatically means that the module is non-unitary. As will be shown later, in all cases, except for the entangled modules, the vacuum $\rho(\Pi)$ is self-conjugated with respect to the conjugation rules (\ref{e: conj rules ef}).

We start with the illustration of the above procedure by the standard $4d$ HS theory, then uplifting it to the $B_2$-modules.

\subsection{Standard HS modules}
The exposition of Section \ref{Fock Space Realization} applies to the standard $A_1$ HS theory
of \cite{Vasiliev:1992av} upon dropping the Coxeter index $i$ and automorphisms $\psi_\pm$.

In the standard HS theory a twisted-adjoint module $F_{tw}$, which describes physical sector of $C(Y;K|x)$, is induced from the vacuum (\ref{e: ef vacuum}). The conjugation correctly relates creation and annihilation operators
\begin{equation}
    (e_{\nu A})^\dagger = f{}_A{}^\nu\,,
\end{equation}
making the vacuum $\Pi = \ket{0}_{tw}$ self-conjugated. The $su(2,2)$ properly acts correctly on the module. Namely, energy is positive-definite and Hermitian, with the vacuum being its eigenvector
\begin{equation}
    E = f{}_1{}^\lambda e_{\lambda 1} + f{}_2{}^\lambda e_{\lambda 2}, \quad E*\ket{0}_{tw}=2\ket{0}_{tw}\,.
\end{equation}

Helicity operator $H$ is Hermitian and together with $t^{\pm i}{}^\mu_\nu$ also act appropriately
\begin{equation}
    H = f{}_1{}^\lambda e_{\lambda 1} - f{}_2{}^\lambda e_{\lambda 2}, \quad H*\ket{0}_{tw}=0, \quad t^{-i}{}^\mu_\nu *\ket{0}_{tw} = 0\,, \quad t^{+i}{}^\mu_\nu *\ket{0}_{tw} \neq 0\,.
\end{equation}
Vectors $(e_{\nu2})^n\ket{0}_{tw}$ and $(f_1{}^\mu)^m\ket{0}_{tw}$ are singular ones that generate an infinite-dimensional irreducible nonintersecting submodules of helicities $(-n)$ and $m$.
Therefore, $F_{tw}$ is a lowest-weight unitary module that decomposes into the direct sum of irreducible modules of all spins (helicities).

The adjoint module $F_{adj}$ is a representation of $su(2,2)$ over the transformed self-conjugated vacuum
\begin{equation}
  \rho(\Pi) = \ket{0}_{adj} = \exp\bigg\{-2 e_{\nu1}f{}_1{}^\nu +2 e_{\nu2}f{}_2{}^\nu\bigg\}\,,
\end{equation}
where $\rho$ is an automorphism corresponding to the Klein operator $k$. It acts as
\begin{equation}
    \rho(e_{\nu 2}) = - f{}_2{}^\nu, \quad \rho(f{}_2{}^\nu) = e_{\nu 2}
\end{equation}
and as identity on the remaining oscillators. For this vacuum the annihilation operators are $\{e_{\nu 1}\,, e_{\nu 2}\}$ and as such the module is not unitary, as the pair of creation operators are $\{f{}_1{}^\nu\,, -f{}_2{}^\nu\}$, the norm of the state $||(-f{}_2{}^\nu) * \ket{0}_{adj}||^2 = -1$ (creation and annihilation operators are not conjugated). This is anticipated since the adjoint module decomposes into an infinite sum of finite-dimensional modules of a non-compact algebra. However, the highest weight structure is respected
\begin{equation}
E*\ket{0}_{adj}= 0,\quad  H*\ket{0}_{adj} = 2\ket{0}_{adj}, \quad t^{\pm}{}^\mu_\nu *\ket{0}_{adj}= 0\,.
\end{equation}
Vectors $(f{}_2{}^\nu)^n\ket{0}_{adj}$ are singular ones that lead to finite-dimensional submodules.

It is worth noting that the unitary left Fock module built from the vacuum $\Pi$ identifies with a doubled singleton Fock space known as doubleton representation of $su(2, 2)$ \cite{Gunaydin:1984fk, Gunaydin:1984vz}, that contains all irreducible $4d$ massless unitary representations of the conformal algebra. As shown in \cite{Iazeolla:2008ix, Basile:2018dzi}, the standard adjoint HS module corresponds to the tensor product of singleton and anti-singleton, which decomposes under the background isometry algebra into the sum of all different adjoint spin-$s$ modules. Therefore, the action of the automorphism $\rho$ can be viewed as a flipping of singleton into the anti-singleton in the tensor product.

\subsection{$B_2$ HS modules}\label{Modules}
\subsubsection{Module $R(k)\bar R(\bar k){}^T = \begin{pmatrix}
-1 & 0 \\
0 & -1
\end{pmatrix}$ } \label{ModuleTT}
Let us apply the procedure to the $B_2$ modules, starting with the module $F_{tw\otimes tw} \simeq M_{tw \otimes tw}$. This module is equipped with positive-definite Hermitian conjugation
\begin{equation}
    (e^i_{\nu A})^\dagger = f^i{}_A{}^\nu
\end{equation}
and total energy and helicity have the self-conjugated vacuum
\begin{equation}
     \ket{0}_{tw \otimes tw} = \exp\bigg\{-2(e^1_{\nu 1} f^1{}_1{}^\nu + e^1_{\nu 2}f^1{}_2{}^\nu + e^2_{\nu 1}f^2{}_1{}^\nu + e^2_{\nu 2}f^2{}_2{}^\nu)\bigg\}
\end{equation}
that satisfies
\begin{equation}
    E = \sum_{i=1}^{2}\bigg(f^i{}_1{}^\lambda e^i_{\lambda 1} + f^i{}_2{}^\lambda e^i_{\lambda 2} \bigg)\,, \quad H = \sum_{i=1}^{2}\bigg(f^i{}_1{}^\lambda e^i_{\lambda 1} - f^i{}_2{}^\lambda e^i_{\lambda 2} \bigg)\,,
\end{equation}

\begin{equation}
    E*\ket{0}_{tw \otimes tw} = 4\ket{0}_{tw \otimes tw}, \quad H*\ket{0}_{tw \otimes tw} = 0\,,
\end{equation}

\begin{equation}
    t^{-}{}^\mu_\nu*\ket{0}_{tw \otimes tw} = 0\,, \quad t^{+}{}^\mu_\nu*\ket{0}_{tw \otimes tw} \neq 0\,.
\end{equation}

Unitary lowest weight module $F_{tw \otimes tw}$ corresponds to the case $R(k)\bar R(\bar k){}^T = \begin{pmatrix}
-1 & 0 \\
0 & -1
\end{pmatrix}$. Other $B_2$ modules result from
$F_{tw \otimes tw}$ via application of combinations of Klein-related automorphisms of the $\{e^i_{\nu A}, f^i{}_A{}^\nu\}$ algebra that correspond to the remaining seven cases of possible matrix products $R(k)\bar R(\bar k){}^T$.

\subsubsection{Module $R(k)\bar R(\bar k){}^T = \begin{pmatrix} \label{ModuleTA1}
-1 & 0 \\
0 & 1
\end{pmatrix}$ or $
\begin{pmatrix}
1 & 0 \\
0 & -1
\end{pmatrix}$}

These modules $F_{tw \otimes adj}$ and $F_{adj \otimes tw}$ result from the automorphisms corresponding to $k_1$ and $k_2$, respectively, which can be realized on the $\{e^i_{\nu A}, f^i{}_A{}^\nu\}$. For the $k_2$ example of the non-trivial transformation:
\begin{equation}
    \rho_2(e^2_{\nu 2}) = - f^2{}_2{}^\nu \,, \quad \rho_2(f^2{}_2{}^\nu) = e^2_{\nu 2}\,.
\end{equation}

The vacuum then takes the form:

\begin{equation}
    \ket{0}_{tw \otimes adj} = \rho_2(\ket{0}_{tw \otimes tw}) = \exp\bigg\{-2(e^1_{\nu 1} f^1{}_1{}^\nu + e^1_{\nu 2}f^1{}_2{}^\nu + e^2_{\nu 1}f^2{}_1{}^\nu - e^2_{\nu 2}f^2{}_2{}^\nu)\bigg\}\,.
\end{equation}

We see that $F_{tw \otimes adj}$ is indeed a product of adjoint and twisted-adjoint modules of the standard HS theory, its annihilation operators being $\{e^1_{\nu 1}\,, e^2_{\nu 1}\,, f^1{}_2{}^\nu \,, e^2_{\nu 2}\}$. It is non-unitary but contains a unitary lowest-weight submodule resulting from the quotienting by the adjoint part. At the field level, it can be achieved by enforcing $C(Y_1,Y_2,I;\hat{K}|x) = C(Y_2,I;\hat{K}|x)$ via the boundary condition (\ref{e: boundary condition}). Analogously, application of the $\rho_1$-automorphism leads to the non-unitary module $F_{adj \otimes tw}$ that contains a unitary lowest-weight submodule which can be extracted at the field level through the condition $C(Y_1,Y_2,I;\hat{K}|x) = C(Y_1,I;\hat{K}|x)$ imposed by the boundary asymptotic behavior (\ref{e: boundary condition}).

\subsubsection{Module $R(k)\bar R(\bar k){}^T = \begin{pmatrix} \label{ModuleAA}
1 & 0 \\
0 & 1
\end{pmatrix}$}

This case results from the composition of automorphisms $\rho_1 \rho_2$, which yields the vacuum
\begin{equation}
    \ket{0}_{adj \otimes adj} = \rho_1\rho_2(\ket{0}_{tw \otimes tw}) = \exp\bigg\{-2(e^1_{\nu 1} f^1{}_1{}^\nu - e^1_{\nu 2}f^1{}_2{}^\nu + e^2_{\nu 1}f^2{}_1{}^\nu - e^2_{\nu 2}f^2{}_2{}^\nu)\bigg\}\,
\end{equation}
with annihilation operators $\{e^1_{\nu 1}\,, e^2_{\nu 1}\,, e^1_{\nu 2} \,, e^2_{\nu 2}\}$ and obviously leads to the product of non-unitary adjoint modules $F_{adj\otimes adj}$. Module $F_{adj\otimes adj}$ contains a unitary trivial submodule that in terms of fields has the form of $C(Y_1,Y_2,I;\hat{K}|x) = C(0,0,I;\hat{K}|0)$.

\subsubsection{Module $R(k)\bar R(\bar k){}^T = \begin{pmatrix} \label{ModuleTA2}
0 & 1 \\
1 & 0
\end{pmatrix}$ or $
\begin{pmatrix}
0 & -1 \\
-1 & 0
\end{pmatrix}$}

Application of the automorphisms $\psi_+$ and $\psi_-$ reproduces these modules. We can define $\psi_+$ on $\{e^i_{\nu A}, f^i{}_A{}^\nu\}$ as (the action on other oscillators is trivial)
\begin{align}
    &\psi_+(e^1_{\nu 2})= \frac{1}{2}(e^1_{\nu 2} - e^2_{\nu 2} - f^1{}_2{}^\nu - f^2{}_2{}^\nu), \quad \psi_+(e^2_{\nu 2})= \frac{1}{2}(-e^1_{\nu 2} + e^2_{\nu 2} - f^1{}_2{}^\nu - f^2{}_2{}^\nu)\,,\\
    &\psi_+(f^1{}_2{}^\nu)= \frac{1}{2}(e^1_{\nu 2} + e^2_{\nu 2} + f^1{}_2{}^\nu - f^2{}_2{}^\nu), \quad \psi_+(f^2{}_2{}^\nu)= \frac{1}{2}(e^1_{\nu 2} + e^2_{\nu 2} - f^1{}_2{}^\nu + f^2{}_2{}^\nu)\,.
    \nonumber
\end{align}

Under this automorphism the vacuum transforms to

\begin{equation}
    \psi_+(\ket{0}_{tw \otimes tw})= \exp\bigg\{-2(e^1_{\nu 1} f^{1 \nu}_1 - e^1_{\nu 2} f^{2 \nu}_2 + e^2_{\nu 1} f^{2 \nu}_1 - e^2_{\nu 2} f^{1 \nu}_2)\bigg\}\,.
\end{equation}

The set of annihilation operators for this vacuum contains linear combinations of the $\{e^i_{\nu A}, f^i{}_A{}^\nu\}$ oscillators and is more conveniently described in terms of oscillators $\{e^\pm_{\nu A}, f^\pm{}_A{}^\nu\}$: $\{e^+_{\nu 1}\,, e^-_{\nu 1}\,, e^+_{\nu 2}\,, f^-{}_2{}^\nu\}$. In these terms one can see that the module is entirely analogous to \ref{ModuleTA1}, also being product of adjoint and twisted-adjoint modules of the standard HS theory. Therefore, similarly quotienting away the adjoint part, generated by the oscillators $\{f^+{}_1{}^\nu, -f^+{}_2{}^\nu\}$, yields a unitary lowest-weight submodule equivalent to the standard twisted-adjoint module. The field realization of this submodule is $C(Y_1,Y_2,I;\hat{K}|x) = C(Y_1-Y_2,I;\hat{K}|x)$. Likewise, the automorphism $\psi_-$ leads to the module that is a product of adjoint and twisted-adjoint modules with a lowest-weight unitary submodule which field realization is $C(Y_1,Y_2,I;\hat{K}|x) = C(Y_1+Y_2,I;\hat{K}|x)$.

\subsubsection{Module $R(k)\bar R(\bar k){}^T = \begin{pmatrix} \label{ModuleM}
0 & -1 \\
1 & 0
\end{pmatrix}$ or $
\begin{pmatrix}
0 & 1 \\
-1 & 0
\end{pmatrix}$}

These entangled modules result from the automorphisms $\rho_1 \psi_+$ and $\rho_2 \psi_+$. The automorphism $\rho_1 \psi_+$ yields the vacuum
\begin{equation}
    \ket{0}_{ent}=\exp\bigg\{-2(e^1_{\nu 1} f^1{}_1{}^\nu - e^1_{\nu 2} e^2_{\nu 2} + e^2_{\nu 1} f^2{}_1{}^\nu  + f^1{}_2{}^\nu f^2{}_2{}^\nu )\bigg\}\,.
\end{equation}
Note that the vacuum $\ket{0}_{ent}$ is not self-conjugated with respect to the conjugation rules (\ref{e: conj rules ef})
\begin{equation}
    \ket{0}_{ent}^\dagger=\exp\bigg\{-2(e^1_{\nu 1} f^1{}_1{}^\nu - f^1{}_2{}^\nu f^2{}_2{}^\nu + e^2_{\nu 1} f^2{}_1{}^\nu  + e^1_{\nu 2} e^2_{\nu 2} )\bigg\} \neq \ket{0}_{ent}\,.
\end{equation}
Therefore, to introduce a bilinear form we have to impose a different conjugation rules on oscillators $\{e^i_{\nu A}, f^i{}_A{}^\nu\}$:
\begin{equation}
     (e^i_{\nu 1})^\dagger = f^i{}_1{}^\nu\,, \quad  (e^1_{\nu 2})^\dagger = f^1{}_2{}^\nu\,, \quad (e^2_{\nu 2})^\dagger = -f^2{}_2{}^\nu\,.
\end{equation}
These rules result in a wrong conjugation of non-compact generators: $\{(t^{+1}{}^\mu_\nu)^\dagger = t^{-1}{}^\mu_\nu\,, (t^{+2}{}^\mu_\nu)^\dagger = -t^{-2}{}^\mu_\nu\}$ (\ie the bilinear form is not invariant).

The set of annihilation operators for this vacuum mixes $\{e^i_{\nu A}, f^i{}_A{}^\nu\}$ : $v^-_{\nu a}=\{e^1_{\nu 1}\,, e^2_{\nu 1}\,, \frac{1}{\sqrt{2}}(f^2{}_2{}^\nu - e^1_{\nu 2})\,, \frac{1}{\sqrt{2}}(f^1{}_2{}^\nu - e^2_{\nu 2})\}$. The corresponding set of creation operators is $v^+_{\nu a} = \{f^1{}_1{}^\nu\,, f^2{}_1{}^\nu\,,\frac{1}{\sqrt{2}}(e^2_{\nu 2} + f^1{}_2{}^\nu)\,,\frac{1}{\sqrt{2}}(e^1_{\nu 2} + f^2{}_2{}^\nu)\}$ so that $[v^-_{\nu a}\,, v^+_{\mu b}] = \delta_{\nu \mu} \delta_{ab}$. As can be seen, the norm $||v^+_{\nu 3}*\ket{0}_{ent}||^2=-1$ (\ie the bilinear form is not positive-definite). Therefore, the module is not unitary. Moreover, the lowest/highest weight structure is also lost. While we have
\begin{gather}
    t^{-}{}^\mu_\nu = e^1_{\nu 1}f^1{}_2{}^\mu + e^2_{\nu 1}f^2{}_2{}^\mu \equiv \frac{1}{\sqrt{2}}\bigg(v_{\nu 1}^-(v_{\mu 4}^- + v_{\mu 3}^+) + v_{\nu 2}^-(v_{\mu 4}^+ + v_{\mu 3}^-)\bigg) \Rightarrow\\
    t^{-}{}^\mu_\nu*\ket{0}_{ent} = 0\,, \quad \quad t^{-}{}^\mu_\nu*G(v^+_3\,,v^+_4)*\ket{0}_{ent} = 0\,, \text{ for any function $G(v^+_3\,,v^+_4)$\,,}
\end{gather}
total energy and helicity no longer act diagonally as in (\ref{e: noneigen vacuum ab})
\begin{gather}
    E = \sum_{i=1}^{2}\bigg(f^i{}_1{}^\lambda e^i_{\lambda 1} + f^i{}_2{}^\lambda e^i_{\lambda 2} \bigg) \equiv (v_{\lambda 1}^+ v_{\lambda 1}^- + v_{\lambda 2}^+ v_{\lambda 2}^- + v_{\lambda 3}^+ v_{\lambda 4}^+ - v_{\lambda 3}^- v_{\lambda 4}^-)\,, \\
    H = \sum_{i=1}^{2}\bigg(f^i{}_1{}^\lambda e^i_{\lambda 1} - f^i{}_2{}^\lambda e^i_{\lambda 2} \bigg) \equiv (v_{\lambda 1}^+ v_{\lambda 1}^- + v_{\lambda 2}^+ v_{\lambda 2}^- - v_{\lambda 3}^+ v_{\lambda 4}^+ + v_{\lambda 3}^- v_{\lambda 4}^-)\,,
\end{gather}
\begin{equation}
    E*{\ket{0}_{ent}}= 2(1 + (e^1_{\nu 2} + f^2{}_2{}^\nu)(e^2_{\nu 2} + f^1{}_2{}^\nu ))\ket{0}_{ent} = 2(1 + 2 v^+_{\lambda 3} v^+_{\lambda 4})\ket{0}_{ent}\,,
\end{equation}
\begin{equation}
    H*{\ket{0}_{ent}}= 2(1 - (e^1_{\nu 2} + f^2{}_2{}^\nu)(e^2_{\nu 2} + f^1{}_2{}^\nu ))\ket{0}_{ent}=2(1 - 2 v^+_{\lambda 3} v^+_{\lambda 4})\ket{0}_{ent}\,.
\end{equation}

This point can be further illustrated by taking the flat limit in the free equations. To that end one can restore the $AdS_4$ radius in the equation (\ref{e:cov7}) and take the limit $\lambda \rightarrow 0$ after rescaling $y^\alpha \rightarrow \lambda^{1/2} y^\alpha; \quad \partial_\alpha = \lambda^{-1/2} \partial_\alpha$. For the module under consideration this yields the equation of the form
\begin{equation}
    \bigg(\d_x + \frac{i}{2} e^{\alpha\dot\alpha}\bigg(\partial_{\alpha 1} \bar{\partial}_{\dot{\alpha} 1} +\partial_{\alpha 1} \bar{\partial}_{\dot{\alpha} 2} +\partial_{\alpha 2} \bar{\partial}_{\dot{\alpha} 2} - \partial_{\alpha 2} \bar{\partial}_{\dot{\alpha} 1}\bigg) \bigg)C(Y_1,Y_2,I;\hat{k},\hat{\bar k}|x) = 0\,. \label{e:flatlimit}
\end{equation}

This equation admits plane wave solutions
\begin{equation}
     C(Y_1,Y_2,I;\hat{k},\hat{\bar k}|x) = \exp{i\bigg(A^{IJ} \xi_{I \alpha} \bar{\xi}_{J \dot{\alpha}} x^{\alpha\dot\alpha} + \delta^{IJ} \xi_{I \alpha} y^\alpha_J + \delta^{IJ} \bar{\xi}_{I \dot{\alpha}} \bar{y}_{J}^{\dot{\alpha}}\bigg)}\,, \label{e:planewave}
\end{equation}
where $\xi$, $\bar{\xi}$ are the Fourier partners for $y$ and $\bar{y}$ and
\begin{equation}
A = \frac{1}{2}\begin{pmatrix}
1 & 1 \\
-1 & 1
\end{pmatrix}\,.
\end{equation}

As this matrix is not diagonalizable in real numbers, the positive and negative frequencies cannot be separated, thus the module is not highest/lowest weight. Therefore, the module in question is not suitable for the description of physical states.

\subsection{Truncation to unitary submodules}

For the $B_2$ model to be a generalization of the standard $4d$ HS theory, one has to ensure that there is a way to consistently eliminate all non-unitary modules in the zero-form sector of the full nonlinear system. As in the standard theory, nonlinear Coxeter system admits an automorphism $\hat{K}_v \rightarrow -\hat{K}_v$ related to the total-parity of Klein operators. Finding an invariant subsystem of this automorphism is rather straightforward.
Indeed, the only equation in the nonlinear system (\ref{e:nonlinear system 1})-(\ref{e:nonlinear system 5}) that has an explicit dependence on the Klein operators is (\ref{e:nonlinear system 5})
\begin{equation}
     S*S = i \bigg(dZ^{An}dZ_{An} + \sum_i \sum_{v \in \mathcal{R}_i} \bigg[F_{i*}(B) \frac{v^n v^m}{(v,v)}dz^\alpha_n dz_{\alpha m} * \varkappa_v \hat{k}_v  + \bar F_{i*}(B) \frac{v^n v^m}{(v,v)}d\bar z^{\dot\alpha}_n d\bar z_{\dot\alpha m} * \bar\varkappa_v \hat{\bar k}_v\bigg]\bigg)\,.
\end{equation}
The invariance under the total parity transformation of the dressed Klein operators condition demands the equation to be even in Klein operators, hence $F_{i*}(B)$ and $\bar F_{i*}(B)$ to be odd. One can confirm by inspection that zero-form modules \ref{ModuleTA1} and \ref{ModuleTA2} and none other fit this restriction. Indeed, since the product of matrices $R\bar{R}^T$ determines the type of module, using the Klein-matrix correspondence (\ref{e: B2 R first})-(\ref{e: B2 R last}) we can make sure that the product of odd number of Klein operators lead to modules \ref{ModuleTA1} and \ref{ModuleTA2} and even number of $\hat{K}_v$ yields the remaining type of modules. Since zero-form fields $B(Y,Z,I;\hat{K}|x)$ valued in modules \ref{ModuleTA1} and \ref{ModuleTA2} guarantee the invariance of the r.h.s of nonlinear equation (\ref{e:nonlinear system 5}), a truncation of the system to fields $B(Y,Z,I;-\hat{K}|x) = -B(Y,Z,I;\hat{K}|x)$ and $W(Y,Z,I;-\hat{K}|x) = W(Y,Z,I;\hat{K}|x)$ is consistent. Therefore, the $B_2$ model is a consistent extension of the standard HS theory that faithfully represents all the massless single-particle states of the standard theory as unitary submodules of \ref{ModuleTA1} and \ref{ModuleTA2} encode generalized Weyl tensors. In particular, it contains the spin-2 gravity sector of the standard theory.

\subsection{Summary}

Summarizing the results, we can distinguish between $B_2$-modules of four categories (Modules \ref{ModuleTA1} and \ref{ModuleTA2} being equivalent via a change of variables). Module \ref{ModuleTT} is unitary. Modules \ref{ModuleAA}, \ref{ModuleTA1} and \ref{ModuleTA2} admit truncation to unitary submodules supported by functions of the reduced number of spinor variables, such as, for instance, $C(Y_1,Y_2,I;\hat{K}|x) = C(Y_2,I;\hat{K}|x)$ in \ref{ModuleTA1}. Let us stress that this formal truncation results from  imposing the boundary conditions on the fields at the linear order, as explained in Section \ref{BoundaryConditions}. The truncation to the unitary submodules in the full nonlinear system is less obvious due to the non-trivial intermixing of different $B_2$ modules, each with their own restriction to unitary submodule conditions, at the vertices. This situation is reminiscent of the analogous entanglement problem of topological and dynamical fields in the $3d$ HS theory \cite{Korybut:2022kdx} so that the solution beyond the linear order could be provided order by order by a suitable shifted or differential homotopy.

Modules \ref{ModuleM}, not being unitary lower-weight modules, are of a new type not present in the standard HS theory. While field equations can be transformed to resemble equation on the product of two twisted-adjoint modules, modules \ref{ModuleM} are not isomorphic to the $M_{tw\otimes tw}$ and form a distinct family of modules specific to the Coxeter extension. Entangled modules arise due to nontrivial mixing of $Y_A^n$ oscillators induced by the action of Coxeter group and their role is to be explored. Moving from the $B_2$ to other Coxeter groups of higher rank the number of entangled modules  rapidly increases. For example, in $B_2$ model entanglement occurs if we combine a transposition with a reflection with respect to the basis vector $e_i$ of the root space. In case of a general $B_p$ model, $n$-cycles and its combinations with reflection with respect to $e_i$ also lead to the entangled modules whose properties are yet to be studied.

Generalization of the approach developed in this section to higher order Coxeter groups is a fascinating topic for the future. A disentanglement criterion for a higher order CHS models allows us to easily separate CHS modules into two groups: products of standard HS modules with well-known properties and entangled modules, which require special considerations.

\section{First On-Shell Theorem}\label{FOST}

In this section we adapt the shifted homotopy technique of \cite{Didenko:2018fgx} in a way applicable to CHS models and extract the First On-Shell Theorem from the general CHS model including the $B_2$ theory.

\subsection{Modified shifted homotopy}\label{Shifted homotopies}

\subsubsection{Contracting homotopy operator}
To reconstruct interaction vertices which look schematically
\begin{gather}
    \d_x \omega = -\omega*\omega + \Upsilon(\omega\,,\omega\,,C) + \Upsilon(\omega\,,\omega\,,C\,,C) + ...\,,\\
    \d_x C = -[\omega\,,C]_* + \Upsilon(\omega\,,C\,,C) + ...\,
\end{gather}
from nonlinear CHS equations one has to repetitively solve equations of the form
\begin{equation}\label{e: generalhomoeq}
    \d_Z f(Y,Z,I;\hat{K};dZ) = g(Y,Z,I;\hat{K};dZ)\,,
\end{equation}
with some $g(Y,Z,I;\hat{K};dZ)$ built from nonlinear combinations of the lower-order fields, that obeys the consistency condition $\d_Z g(Y,Z,I;\hat{K};dZ) = 0$. These equations can be solved by modifying a well-known homotopy trick. Firstly, similarly to \cite{Didenko:2018fgx}, we choose a nilpotent homotopy operator
\begin{equation}
    \partial = (Z^A_n + I_n Q^A_n) \frac{\partial}{\partial d Z^A_n}\,,
\end{equation}
where $Q^A_n$ is some $ Z^A_n$ independent operator,
\begin{equation}
    \frac{\partial Q^B_m}{\partial Z^A_n} = 0\,.
\end{equation}
Note that idempotents $I_n$ appear in the definition of the homotopy operator and in any object derived from $\partial$. Therefore, we denote the dressed shift parameters as
\begin{equation}
    \hat{Q}_n^A := I_n Q_n^A\,.
\end{equation}

Then we introduce operator
\begin{equation}
    N=\{\d_{Z}, \partial\}
\end{equation}
and its almost inverse
\begin{equation}
    N^{*} g(Y, Z, I; dZ):=\int_{0}^{1}  \frac{d t}{t} g(Y, t Z_n-(1-t) \hat{Q}_n, I; t dZ), \quad g(Y, -\hat{Q}_n, I; 0)=0\,.
\end{equation}
	
The contracting homotopy operator
\begin{equation}
    \Delta_{Q}:=\partial N^{*}, \quad \Delta_{Q} g(Y, Z ; dZ)=\left(Z^{A}_n+\hat{Q}^{A}_n\right)
\frac{\partial}{\partial dZ^{A}_n} \int_{0}^{1} \frac{d t }{t} g(Y, t Z_i-(1-t) \hat{Q}_i, I; t dZ) \label{e: shifted homopoty operator}
\end{equation}
satisfies the resolution of identity
\begin{equation}
    \{\d_Z, \Delta_Q \} = 1 - h_Q\,, \label{e: resolution of identity}
\end{equation}
with $h_Q$ being a cohomology projector
\begin{equation}
    h_Q f(Z, I; dZ) = f(-\hat{Q}_n, I; 0)\,.
\end{equation}
	
Hence, resolution of identity yields a particular solution to (\ref{e: generalhomoeq})
\begin{equation}\label{fqg}
    f = \Delta_Q g
\end{equation}
as long as $h_Q g = 0$. General solution of (\ref{e: generalhomoeq}) is
\begin{equation}
    f(Y,Z,I;dZ) = \Delta_Q g(Y,Z,I;dZ) + h(Y,I) + \d_Z \epsilon(Y,Z,I;dZ)\,,
\end{equation}
where $h(Y,I)$ is a cohomology representative and $\epsilon(Y,Z,I;dZ)$ is a gauge transformation parameter ($\d_Z$-exact term). Transition from one $Q$ to another affects the $h$ and $\epsilon$-dependent parts of the solution. As a result, the choice of $Q$ in (\ref{fqg}) affects the choice of field variables and is essential for the analysis of locality.

\subsubsection{Properties of $\Delta_Q$}

Here properties of the operators $\Delta_Q$ and $h_Q$ are presented. Most of them directly generalize those of \cite{Didenko:2018fgx} with a notable change in the star-exchange formulas. Firstly, operators $\Delta_Q$ and $\Delta_P$ anti-commute
\begin{equation}
    \Delta_Q \Delta_P = - \Delta_P \Delta_Q\,
\end{equation}
that follows from a direct application of (\ref{e: shifted homopoty operator}).
Analogously, anti-symmetry in the indices $P$ and $Q$ is present in
\begin{equation}
    h_P \Delta_Q = - h_Q \Delta_P\,.
\end{equation}
Other important relations are
\begin{equation}
    h_P h_Q = h_Q\,, \quad \Delta_P h_Q = 0\,
\end{equation}
and
\begin{equation}
    \Delta_B - \Delta_A = [\d_Z,\Delta_A \Delta_B] + h_A \Delta_B \label{e: delta difference}
\end{equation}
that follows from the resolution of identity (\ref{e: resolution of identity}).

Confining ourselves to the holomorphic variables $(Z^n_A, Y^n_A, \hat K) \rightarrow (z^n_\mu, y^n_\mu, \hat k)$,
let us write down how $\Delta_b \Delta_a$ and $h_c \Delta_b \Delta_a$ act
\begin{equation}
    \Delta_b \Delta_a f(y,z, I) dz^{n \mu} dz_{n \mu} = 2 \int_{[0,1]^3} d^3 \tau \delta(1-\tau_1 -
\tau_2 - \tau_3) (z+\hat{b})_{m \nu} (z+\hat{a})^{m \nu} f(y, \tau_1 z - \tau_3 \hat{b} -\tau_2 \hat{a},I)\,, \label{dd}
\end{equation}
\begin{equation}
    h_c \Delta_b \Delta_a f(y,z, I) dz^{n \mu} dz_{n \mu} = 2 \int_{[0,1]^3} d^3 \tau \delta(1-
\tau_1 - \tau_2 - \tau_3)(\hat{b}-\hat{c})_{m\nu} (\hat{a}-\hat{c})^{m\nu} f(y, -\tau_1 \hat{c} - \tau_3 \hat{b} - \tau_2 \hat{a}, I)\, \label{hdd}\,,
\end{equation}
where $\{\hat{a}_n,\hat{b}_n ,\hat{c}_n\} = \{I_n a_n, I_n b_n, I_n c_n\}$.
	
Note that from (\ref{hdd}) it follows that for any parameter $\kappa$
\begin{equation}
    h_{(\kappa + 1)q_2 - \kappa q_1} \Delta_{q_2} \Delta_{q_1} = 0 \,. \label{hdd=0}
\end{equation}

Application of (\ref{hdd}) to the $\hat{\gamma}_v$ defined as
\begin{equation}
    \hat{\gamma}_v = \exp\bigg(i \frac{v^p v^q}{(v,v)} z_{\alpha p}y^\alpha_q \bigg)\frac{v^n v^m}{(v,v)}dz^\alpha_n dz_{\alpha m} \hat{k}_v
\end{equation}
yields
\begin{equation}
    h_c \Delta_b \Delta_a \hat{\gamma}_v = 2 \int_{[0,1]^3} d^3 \tau \delta(1-\tau_1 - \tau_2 - \tau_3)
 (b-c)_{n\nu} (a-c)^\nu_m \frac{v^n v^m}{(v,v)} exp\bigg\{-i\frac{v^p v^q}{(v,v)}(\tau_1 c + \tau_2 a + \tau_3 b)_{p\alpha} y^\alpha_q\bigg\} \hat{k}_v \,.\label{e: hdd_int}
\end{equation}
Note that we absorbed idempotents that go with shift parameters $a, b, c$ into vectors $v^n$. Element $\hat{\gamma}_v$ commutes with all $Y^A, Z^A$ variables, but is non-trivialy transformed by the action of Klein operators
\begin{equation}
    \hat{k}_u * \hat{\gamma}_v = \hat{\gamma}_{R_u(v)}*\hat{k}_u\,.
\end{equation}

Another important property of the operators $\Delta_Q$ and $h_P$, implying the $z$-independence of the vertices resulting from the nonlinear equations, is
\begin{equation}\label{e: hdd_0}
    \left(\Delta_{d}-\Delta_{c}\right)\left(\Delta_{a}-\Delta_{b}\right) \hat{\gamma}_v=\left(h_{d}-h_{c}\right) \Delta_{a} \Delta_{b} \hat{\gamma}_v\,.
\end{equation}
Indeed, one can check that $\d_z \hat{\gamma}_v = 0\,, h_a \hat{\gamma}_v = 0\,, h_a \Delta_b \hat{\gamma}_v = 0$ and $\Delta_a \Delta_b \Delta_c \hat{\gamma}_v = 0$. Combining all these facts with (\ref{e: delta difference}) one arrives at  (\ref{e: hdd_0}).

A practically important consequence of (\ref{e: hdd_0}) at $d = a$ is
\begin{equation}\label{e: hdd}
    (\Delta_c \Delta_b - \Delta_c \Delta_a + \Delta_b \Delta_a) \hat{\gamma}_v = h_c \Delta_b \Delta_a \hat{\gamma}_v\,.
\end{equation}

Star-exchange relations	with $z$-independent elements in a CHS theory take the form
\begin{equation}
    \Delta_{q + \alpha y} ( C(y,I) * \phi(z,y,I;\hat{k}_v;dZ)) = C(y,I) * \Delta_{q + (1-\alpha)I p + \alpha y} \phi(z,y,I;\hat{k}_v;dZ)\,, \label{e: Cexchange}
\end{equation}
\begin{equation}
    \Delta_{q + \alpha y} ( \phi(z,y,I;dZ) * \hat{k}_v * C(y,I) ) = \Delta_{q + (1+\alpha)I R_v(p) + \alpha y} (\phi(z,y,I;dZ) *\hat{k}_v) * C(y,I) \label{e: exchangeC}\,,
\end{equation}
where
\begin{equation}
p^n_\mu C(Y,I;\hat{K}) = C(Y,I;\hat{K})p^n_\mu := -i \frac{\partial}{\partial y^\mu_n} C(Y,I;\hat{K})\,.
\end{equation}

Comparison with the star-exchange in the standard framework of \cite{Didenko:2018fgx} shows that shift parameter $p^n_\mu$ acquires an idempotent dressing in both formulas and $p^n_\mu$ is reflected with respect to the Klein $\hat{k}_v$ in (\ref{e: exchangeC}). Standard HS theory corresponds to the $\mathds{Z}_2$ case with a unique reflection matrix $R_k = -1$.

The central elements in the Coxeter models can be obtained by summation of $\hat{\gamma}_v$ over root vectors of any conjugacy class $\mathcal{R}_i$ with equal weights to preserve the $\mathcal{C}$ invariance
\begin{equation}\label{e: gammai}
    \hat{\gamma}_i = \sum_{v \in \mathcal{R}_i} \hat{\gamma}_v\,.
\end{equation}

Modified star-exchange properties (\ref{e: Cexchange}) and (\ref{e: exchangeC}) yield
\begin{equation}\label{e: gammaexch}
\Delta_{q + \alpha y} \hat{\gamma}_v * C(y,I) = C(y,I) * \Delta_{q + \alpha y + (1-\alpha)I p - (1+\alpha)I R_v(p)} \hat{\gamma}_{v}\,.
\end{equation}
Therefore,
\begin{equation}
   \Delta_{q + \alpha y}  \hat{\gamma}_i * C(y,I) = C(y,I) * \sum_{v \in \mathcal{R}_i}  \Delta_{q + \alpha y + (1-\alpha)I p - (1+\alpha)I R_v(p)} \hat{\gamma}_{v}\,.
\end{equation}

Note that the field $C$ does not depend on Klein operators in star-exchange formulas (\ref{e: Cexchange})-(\ref{e: exchangeC}) and (\ref{e: gammaexch}) (in \cite{Didenko:2018fgx} the analogous formulas have a Klein-dependent field $C$). This is important because in a general CHS model Klein operators $\hat{K}_v$ and $\hat{K}_u$ do not commute (\ref{e: qIK2}). Therefore, it is necessary to control the placement and order of Klein operators in each expression. By default we pull all Klein operators from the fields to the far right position in each expression and arrange them in the order in which the fields containing them are located.

\subsection{First order of a general CHS}

A vacuum solution to the full nonlinear general CHS system (\ref{e:nonlinear system 1})-(\ref{e:nonlinear system 5}) is taken in the form
\begin{align}
    B_0(Y,Z,I;\hat{K}|x) &= 0\,,\\
    S_0(Y,Z,I;\hat{K}|x) &= dZ^{\alpha n}Z_{\alpha n}\,,\\
    W_0(Y,Z,I;\hat{K}|x) &= \omega_0(Y,I;\hat{K}|x)\,,
\end{align}
where $\omega_0$ is some flat connection,
\begin{equation}
    \d_x \omega_0(Y,I;\hat{K}|x) + \omega_0(Y,I;\hat{K}|x)*\omega_0(Y,I;\hat{K}|x)=0\,.
\end{equation}

It is important to notice that
\begin{equation}
    [S_0, f(Y,Z,I;\hat{K})]_* = -2i dZ^A_n \frac{\partial}{\partial Z^A_n} f(Y,Z,I;\hat{K}) = -2i \d_Z f(Y,Z,I;\hat{K})\,. \label{e: dz}
\end{equation}

Then, in the first order, equation (\ref{e:nonlinear system 4}) yields
\begin{equation}
[S_0, B_1]_* + [S_1, B_0]_* = 0\,.
\end{equation}
From the vacuum solution and (\ref{e: dz}) it follows that $B_1$ is $Z$-independent, $B_1 = C(Y,I;\hat{K}|x)$.
Therefore eq.(\ref{e:nonlinear system 2}) leads to
\begin{equation}
    \d_x C + [\omega, C]_* = 0\,, \label{dC+CC}
\end{equation}
which encodes the covariant constancy equations studied in Section \ref{Derivative and modules}. To simplify resulting equations, here and in the sequel we will combine the background field $\omega_0(Y,I;\hat{K}|x)$ and the $Z$-independent part of first-order fluctuations $\omega_1(Y,I;\hat{K}|x)$ into a single field $\omega(Y,I;\hat{K}|x) = \omega_0(Y,I;\hat{K}|x) + \omega_1(Y,I;\hat{K}|x)$ . We shall consider the resulting equations up to the first order, meaning that out of the total $\omega$ in (\ref{dC+CC}) only the zero-order $\omega_0$ is present since the field $C$ is a first order field.

Expression for $S_1$ via the field $C$ can be extracted from eq.(\ref{e:nonlinear system 5})
\begin{equation}\label{e: S1}
    -2i \d_Z S_1 = i\sum_l \bigg(\eta_l C*\hat{\gamma}_l +  \bar{\eta}_l C * \hat{\bar{\gamma}}_l\bigg)\,,
\end{equation}
where $\hat{\gamma}_l$ and $\hat{\bar{\gamma}}_l$ are central elements (\ref{e: gammai}) corresponding to the conjugacy class $\mathcal{R}_l$.
Then, for $S_1 = S_1^\eta +S_1^{\bar\eta}$, we obtain in the $\eta$-dependent (holomorphic)
sector using standard homotopy
\begin{equation}\label{e: S1_solution}
    S_1^\eta =-\sum_k \frac{\eta_k}{2} \Delta_{0} \big(C * \hat{\gamma}_k \big) =-\sum_k \frac{\eta_k}{2}\sum_{v \in \mathcal{R}_k} \Delta_{0} \big(C * \hat{\gamma}_v \big) = -\sum_k \frac{\eta_k}{2}\sum_{v \in \mathcal{R}_k} C * \Delta_{I p}\hat{\gamma}_v\,.
\end{equation}
	
The next step is to solve eq.(\ref{e:nonlinear system 3})  which yields in the first order
\begin{equation}\label{e: W1}
    \d_z W_1^\eta = \frac{1}{2 i}\big(\d_x S^\eta_1 + \omega * S_1^\eta + S_1^\eta * \omega \big)\,.
\end{equation}

Adopting notation from \cite{Didenko:2018fgx}
\begin{equation}
    t_\mu^n \omega (Y,I;\hat{K}|x) = -i \frac{\partial}{\partial y^\mu_n} \omega(Y,I;\hat{K}|x)\,,
\end{equation}
conventional homotopy leads to
\begin{equation}\label{e: W1eta}
    W_1^\eta = \frac{1}{4 i}\sum_k \eta_k \sum_{v \in \mathcal{R}_k} \bigg(\omega*C*\Delta_{I(t+p)}\Delta_{Ip} \hat{\gamma}_v - C*\omega*\Delta_{I(t+p)}\Delta_{I(t+p-R_v(t))}\hat{\gamma}_v\bigg)\,.
\end{equation}

Now consider equation (\ref{e:nonlinear system 1}). In the first order it yields
\begin{equation}
    \d_x \omega + \omega*\omega + \d_x W_1^\eta + \omega*W_1^\eta + W_1^\eta *\omega + c.c. = 0\,.
\end{equation}
Using (\ref{e: W1eta}) and applying formulas (\ref{e: hdd})-(\ref{e: exchangeC}), (\ref{e: gammaexch}) one can obtain
\begin{equation}
    \d_x \omega + \omega*\omega = \Upsilon^\eta(\omega,\omega,C) + \Upsilon^\eta(\omega,C,\omega) + \Upsilon^\eta(C,\omega,\omega) + c.c.\,,
\end{equation}
where
\begin{equation}\label{e: full_vert_beginning}
    \Upsilon^\eta(\omega,\omega,C) = \frac{1}{4i}\sum_k \eta_k \sum_{v \in \mathcal{R}_k} \omega*\omega*C*h_{I(t_1+t_2+p)}\Delta_{I p}\Delta_{I(p+t_2)}\hat{\gamma}_v\,,
\end{equation}
\begin{multline}
    \Upsilon^\eta(\omega,C,\omega) = -\frac{1}{4i}\sum_k \eta_k \sum_{v \in \mathcal{R}_k} \omega*C*\omega*\bigg(h_{I(t_1+t_2+p)}\Delta_{I(p+t_1+t_2-R_v(t_2))}\Delta_{I(p+t_2)}\hat{\gamma}_v+ \\
    + h_{I(p+t_1+t_2-R_v(t_2))}\Delta_{I(p+t_2-R_v(t_2))}\Delta_{I(p+t_2)}\hat{\gamma}_v\bigg)\,,
\end{multline}
\begin{equation}\label{e: full_vert_ending}
    \Upsilon^\eta(C,\omega,\omega) = \frac{1}{4i}\sum_k \eta_k \sum_{v \in \mathcal{R}_k} C*\omega*\omega*h_{I(t_1+t_2+p)}\Delta_{I(p+t_1+t_2-R_v(t_2))}\Delta_{I(p+t_1+t_2-R_v(t_1+t_2))}\hat{\gamma}_v\,.
\end{equation}

Structurally, vertices (\ref{e: full_vert_beginning})-(\ref{e: full_vert_ending}) resemble those of \cite{Didenko:2018fgx} in the standard HS theory. Therefore, standard First On-Shell Theorem should be present in the vertices decomposition over $AdS_4$ background. However, in the Coxeter HS model we have a variety of Klein operators and related elements $\hat{\gamma}_v$ which results in additional non-standard terms in the expansion over $AdS_4$ background. An important distinction that may be important for finding a connection with the String theory is the presence of multiple constants $\eta_k$. In the next section we consider vertices (\ref{e: full_vert_beginning})-(\ref{e: full_vert_ending}) over $AdS_4$ space both in a general CHS theory and in the $B_2$ case.

An extension of Coxeter-modified shifted homotopies to the differential homotopy, introduced in \cite{Vasiliev:2023yzx}, would be a productive avenue of work, useful at higher orders in the perturbative procedure. However, at the linear level the shifted homotopy is sufficient.

\subsection{First On-Shell Theorem}
\subsubsection{General case}
In this section we calculate $\Upsilon^\eta(\Omega_{AdS},\Omega_{AdS},C)$, $\Upsilon^\eta(\Omega_{AdS},C,\Omega_{AdS})$ and $\Upsilon^\eta(C,\Omega_{AdS},\Omega_{AdS})$ with
$\Omega_{AdS}(Y|x)$ (\ref{e:ansatz}).
Using (\ref{e: hdd_int}), properties of idempotents $I_n$ and reflection matrices $R_v{}^n{}_m$ we obtain
\begin{multline}
    h_{I(t_1+t_2+p)}\Delta_{I p}\Delta_{I(p+t_2)}\hat{\gamma}_v = 2 \int_{[0,1]^3} d^3 \tau \delta(1-
\sum_i \tau_i) t_2{}_{\alpha m}t_1{}^\alpha{}_n \frac{v^n v^m}{(v,v)} \\exp\bigg(-i \frac{v^a v^b}{(v,v)}(y^\alpha_a p_{\alpha b} 
+y^\alpha_a [\tau_1(t_1+t_2)+\tau_2 t_2]_{\alpha b}) \bigg)*\hat{k}_v\,,
\end{multline}
\begin{multline}
    h_{I(t_1+t_2+p)}\Delta_{I(p+t_1+t_2-R_v(t_2))}\Delta_{I(p+t_1+t_2-R_v(t_1+t_2))}\hat{\gamma}_v = 2 \int_{[0,1]^3} d^3 \tau \delta(1-
\sum_i \tau_i) t_2{}_{\alpha m}t_1{}^\alpha{}_n \frac{v^n v^m}{(v,v)} \\
exp\bigg(-i \frac{v^a v^b}{(v,v)}(y^\alpha_a [p+t_1+t_2]_{\alpha b} +y^\alpha_a [\tau_1 t_2+\tau_2 (t_1+t_2)]_{\alpha b}) \bigg)*\hat{k}_v\,,
\end{multline}
\begin{multline}
    h_{I(t_1+t_2+p)}\Delta_{I(p+t_1+t_2-R_v(t_2))}\Delta_{I(p+t_2)}\hat{\gamma}_v = -2 \int_{[0,1]^3} d^3 \tau \delta(1-
\sum_i \tau_i) t_2{}_{\alpha m}t_1{}^\alpha{}_n \frac{v^n v^m}{(v,v)} \\
exp\bigg(-i \frac{v^a v^b}{(v,v)}(y^\alpha_a [p + t_2]_{\alpha b} +y^\alpha_a [\tau_1 t_1 + \tau_2 (t_1+t_2)]_{\alpha b}) \bigg)*\hat{k}_v\,,
\end{multline}
\begin{multline}
    h_{I(p+t_1+t_2-R_v(t_2))}\Delta_{I(p+t_2-R_v(t_2))}\Delta_{I(p+t_2)}\hat{\gamma}_v = -2 \int_{[0,1]^3} d^3 \tau \delta(1-
\sum_i \tau_i) t_2{}_{\alpha m}t_1{}^\alpha{}_n \frac{v^n v^m}{(v,v)} \\
exp\bigg(-i \frac{v^a v^b}{(v,v)}(y^\alpha_a [p + t_2]_{\alpha b} +y^\alpha_a [\tau_1 (t_1 + t_2) + \tau_2 t_2]_{\alpha b}) \bigg)*\hat{k}_v\,.
\end{multline}

Consequently,
\begin{multline}
    \Omega_{AdS}(Y|x)*\Omega_{AdS}(Y|x)*C(Y,I;\hat{K}_C|x) *  h_{I(t_1+t_2+p)}\Delta_{I p}\Delta_{I(p+t_2)}\hat{\gamma}_v \bigg|_{ee} = \\
    = -\frac{1}{4}\frac{v^n v^m}{(v,v)}\bar{H}_{\dot\alpha \dot\beta}\bigg[\bar{y}^{\dot\alpha}_n\bar{y}^{\dot\beta}_m + i \bar{y}^{\dot\alpha}_n \bar{\partial}^{\dot\beta}_m + i \bar{y}^{\dot\beta}_m \bar{\partial}^{\dot\alpha}_n - \bar{\partial}^{\dot\alpha}_n \bar{\partial}^{\dot\beta}_m \bigg]C(\mathds{P}_v(y),\bar{y},I;\hat{K}_C|x)*\hat{k}_v\,,
\end{multline}
\begin{multline}
    C(Y,I;\hat{K}_C|x)*\Omega_{AdS}(Y|x)*\Omega_{AdS}(Y|x) *  h_{I(t_1+t_2+p)}\Delta_{I(p+t_1+t_2-R_v(t_2))}\Delta_{I(p+t_1+t_2-R_v(t_1+t_2))}\hat{\gamma}_v \bigg|_{ee} = \\
    = -\frac{1}{4}\frac{v_k v_l}{(v,v)} \delta^{nm} \delta^{pq} R(\hat{K}_C)^k_n R(\hat{K}_C)^l_p \bar{R}(\hat{K}_C)^w_m \bar{R}(\hat{K}_C)^z_q \bar{H}_{\dot\alpha \dot\beta}
    \bigg[\bar{y}^{\dot\alpha}_w\bar{y}^{\dot\beta}_z - i \bar{y}^{\dot\alpha}_w \bar{\partial}^{\dot\beta}_z -\\
    -i \bar{y}^{\dot\beta}_z \bar{\partial}^{\dot\alpha}_w - \bar{\partial}^{\dot\alpha}_w \bar{\partial}^{\dot\beta}_z \bigg]C(\mathds{P}_v(y),\bar{y},I;\hat{K}_C|x)*\hat{k}_v\,,
\end{multline}
\begin{multline}
    \Omega_{AdS}(Y|x)*C(Y,I;\hat{K}_C|x)*\Omega_{AdS}(Y|x) *  h_{I(t_1+t_2+p)}\Delta_{I(p+t_1+t_2-R_v(t_2))}\Delta_{I(p+t_2)}\hat{\gamma}_v \bigg|_{ee} = \\
    = \frac{1}{4}\frac{v^n v_k}{(v,v)}\delta^{pq} R(\hat{K}_C)^k_p \bar{R}(\hat{K}_C)^l_q \bar{H}_{\dot\alpha \dot\beta}\bigg[\bar{y}^{\dot\alpha}_n\bar{y}^{\dot\beta}_l - i \bar{y}^{\dot\alpha}_n \bar{\partial}^{\dot\beta}_l + i \bar{y}^{\dot\beta}_l \bar{\partial}^{\dot\alpha}_n + \bar{\partial}^{\dot\alpha}_n \bar{\partial}^{\dot\beta}_l \bigg]C(\mathds{P}_v(y),\bar{y},I;\hat{K}_C|x)*\hat{k}_v\,,
\end{multline}
\begin{multline}
    \Omega_{AdS}(Y|x)*C(Y,I;\hat{K}_C|x)*\Omega_{AdS}(Y|x) *  h_{I(p+t_1+t_2-R_v(t_2))}\Delta_{I(p+t_2-R_v(t_2))}\Delta_{I(p+t_2)}\hat{\gamma}_v \bigg|_{ee} = \\
    = \frac{1}{4}\frac{v^n v_k}{(v,v)}\delta^{pq} R(\hat{K}_C)^k_p \bar{R}(\hat{K}_C)^l_q \bar{H}_{\dot\alpha \dot\beta}\bigg[\bar{y}^{\dot\alpha}_n\bar{y}^{\dot\beta}_l - i \bar{y}^{\dot\alpha}_n \bar{\partial}^{\dot\beta}_l + i \bar{y}^{\dot\beta}_l \bar{\partial}^{\dot\alpha}_n + \bar{\partial}^{\dot\alpha}_n \bar{\partial}^{\dot\beta}_l \bigg]C(\mathds{P}_v(y),\bar{y},I;\hat{K}_C|x)*\hat{k}_v\,,
\end{multline}
where we only account for the terms that contain the product of two vierbeins $e$
\begin{equation}\label{e: HH_decomp}
    e^{\nu\dot\nu} e^{\lambda\dot\lambda} = \frac{1}{2}H^{\nu\lambda}\bar\varepsilon^{\dot\nu\dot\lambda}+\frac{1}{2}\bar H^{\dot\nu\dot\lambda}\varepsilon^{\nu\lambda}\,,
\end{equation}
where the basis two-forms are
\begin{align}
    H^{\nu\lambda}=H^{(\nu\lambda)}:=e^{\nu}_{\ \dot\gamma} e^{\lambda\dot\gamma}\,,\qquad
    \bar H^{\dot\nu\dot\lambda}=H^{(\dot\nu\dot\lambda)}:=e^{\ \dot\nu}_{\gamma} e^{\gamma\dot\lambda}\,.
\end{align}
Matrices $R(\hat{K}_C)$ and $\bar{R}(\hat{K}_C)$ are reflections corresponding to the Klein operator $\hat{K}_C$. Since Klein operators can come from fields $C$, $\omega$ (although the latter will not play a role in our
  analysis), as well as via $\hat{\gamma}_v$ for clarity we have introduced a subscript designating the source of the Klein operator, such as $\hat{K}_C$ sourced by the fields $C$.
\begin{equation}\label{e: projector}
    (\mathds{P}_v){}^n{}_m = \delta^n_m - \frac{v^n v_m}{(v,v)}
\end{equation}
is a projector onto a plane orthogonal to the root vector $v$ that reduces the number of spinor variables in the field $C$.

Note that, according to the convention on the arrangement of Klein operators, the expression $C(\mathds{P}_v(y),\bar{y},I;\hat{K}_C|x)*\hat{k}_v \equiv C(\mathds{P}_v(y),\bar{y},I|x)*\hat{k}_v*\hat{K}_C$, \ie we have already pulled Klein operators $\hat{K}_C$ to the right most position.

In the standard HS theory the underlying structure of deformed oscillator algebra guarantees the Lorenz covariant form of the resulting equations \cite{Vasiliev:1999ba}. Since the framed Cherednik algebra (\ref{e:cherednik comm}) is a generalization of the deformed oscillator algebra respecting $sl_2(\mathds{R})$ one can carry out the same reasoning, implying that terms $\omega\omega$ and $\omega e$ cancel out in the vertices.

The resulting holomorphic vertices are
\begin{equation}\label{e: COST OOC}
    \Upsilon^\eta(\Omega_{AdS},\Omega_{AdS},C) = \frac{i}{16}\sum_k \eta_k \sum_{v \in \mathcal{R}_k} \frac{v^n v^m}{(v,v)}\bar{H}_{\dot\alpha \dot\beta}\bigg[\bar{y}^{\dot\alpha}_n\bar{y}^{\dot\beta}_m + i \bar{y}^{\dot\alpha}_n \bar{\partial}^{\dot\beta}_m + i \bar{y}^{\dot\beta}_m \bar{\partial}^{\dot\alpha}_n - \bar{\partial}^{\dot\alpha}_n \bar{\partial}^{\dot\beta}_m \bigg]C(\mathds{P}_v(y),\bar{y},I;\hat{K}_C|x)*\hat{k}_v\,,
\end{equation}
\begin{multline}\label{e: COST COO}
    \Upsilon^\eta(C,\Omega_{AdS},\Omega_{AdS}) = \frac{i}{16}\sum_k \eta_k \sum_{v \in \mathcal{R}_k} \frac{v_k v_l}{(v,v)} \delta^{nm} \delta^{pq} R(\hat{K}_C)^k_n R(\hat{K}_C)^l_p \bar{R}(\hat{K}_C)^w_m \bar{R}(\hat{K}_C)^z_q \\
    \bar{H}_{\dot\alpha \dot\beta}\bigg[\bar{y}^{\dot\alpha}_w\bar{y}^{\dot\beta}_z - i \bar{y}^{\dot\alpha}_w \bar{\partial}^{\dot\beta}_z -i \bar{y}^{\dot\beta}_z \bar{\partial}^{\dot\alpha}_w - \bar{\partial}^{\dot\alpha}_w \bar{\partial}^{\dot\beta}_z \bigg]C(\mathds{P}_v(y),\bar{y},I;\hat{K}_C|x)*\hat{k}_v\,,
\end{multline}
\begin{multline}\label{e: COST OCO}
    \Upsilon^\eta(\Omega_{AdS},C,\Omega_{AdS}) = \frac{i}{8}\sum_k \eta_k \sum_{v \in \mathcal{R}_k}\frac{v^n v_k}{(v,v)}\delta^{pq} R(\hat{K}_C)^k_p \bar{R}(\hat{K}_C)^l_q \bar{H}_{\dot\alpha \dot\beta}\bigg[\bar{y}^{\dot\alpha}_n\bar{y}^{\dot\beta}_l - i \bar{y}^{\dot\alpha}_n \bar{\partial}^{\dot\beta}_l + \\
    + i \bar{y}^{\dot\beta}_l \bar{\partial}^{\dot\alpha}_n + \bar{\partial}^{\dot\alpha}_n \bar{\partial}^{\dot\beta}_l \bigg]C(\mathds{P}_v(y),\bar{y},I;\hat{K}_C|x)*\hat{k}_v\,.
\end{multline}

At $\mathcal{C} = \mathds{Z}_2$ total holomorphic vertex reduces to
\begin{multline}
    \Upsilon_{tot}^{\eta}(\Omega_{AdS},\Omega_{AdS},C) = \Upsilon^\eta(\Omega_{AdS},\Omega_{AdS},C) + \Upsilon^\eta(C,\Omega_{AdS},\Omega_{AdS}) + \Upsilon^\eta(\Omega_{AdS},C,\Omega_{AdS}) = \\
    =\frac{i \eta}{4}\bar{H}_{\dot\alpha \dot\beta} \bar{y}^{\dot\alpha}\bar{y}^{\dot\beta}\bigg(C(0,\bar{y};K_C|x) + C(0,\bar{y};-K_C|x) \bigg)*k - \frac{i \eta}{4}\bar{H}_{\dot\alpha \dot\beta} \bar{\partial}^{\dot\alpha}\bar{\partial}^{\dot\beta}\bigg(C(0,\bar{y};K_C|x) - C(0,\bar{y};-K_C|x) \bigg)*k\,,
\end{multline}
which is the standard $4d$ First On-Shell Theorem \cite{Vasiliev:1988sa}. Moreover, for a general Coxeter group $\mathcal{C}$ the standard form of the First On-Shell Theorem is preserved along any root vector. Indeed, consider the field $C(y,\bar{y},I;\hat{k}_u|x)$, where $u$ is a root vector belonging to the conjugacy class $\mathcal{R}_0$. Then
\begin{equation}
    \Upsilon_{tot}^{\eta_0}(\Omega_{AdS},\Omega_{AdS},C) = \frac{\eta_0}{2i} \frac{u^n u^m}{(u,u)}\bar{H}_{\dot\alpha \dot\beta}\bar{\partial}^{\dot\alpha}_n\bar{\partial}^{\dot\beta}_m C(\mathds{P}_u(y),\bar{y},I;\hat{k}_u|x)*\hat{k}_u + \frac{i \eta_0}{16}\bar{H}_{\dot\alpha \dot\beta}\sum_{v \in \mathcal{R}_0, v\neq \pm u}(...)\,,
\end{equation}
where $\Upsilon_{tot}^{\eta_0}(\Omega_{AdS},\Omega_{AdS},C)$ is a $\eta_0$ part of a total holomorphic vertex.
Projection $\mathds{P}_u$ transforms the vertex to a partial ultra-local form in the terminology of \cite{Didenko:2018fgx} since
it projects out variables $Y$ along the root vector $u$.

\subsubsection{$B_2$}

Now we use reflection matrices of $B_2$ (\ref{e: B2 R first})-(\ref{e: B2 R last}) and conjugacy classes (\ref{e: B2 conjugacy classes}) to derive the explicit form of the First On-Shell Theorem. We introduce the notation
\begin{equation}
    \Upsilon_{tot}^{\eta_1}(\Omega_{AdS},\Omega_{AdS},C) \Leftrightarrow \mathcal{R}_1\,, \quad \Upsilon_{tot}^{\eta_2}(\Omega_{AdS},\Omega_{AdS},C) \Leftrightarrow\mathcal{R}_2\,.
\end{equation}

Then
\begin{multline}\label{e: B2 vert 1}
    \Upsilon_{tot}^{\eta_1}(\Omega_{AdS},\Omega_{AdS},C) = \frac{i \eta_1}{8}\bar{H}_{\dot\alpha \dot\beta}\bigg[\bar{y}^{\dot\alpha}_1\bar{y}^{\dot\beta}_1 + 2 i \bar{y}^{\dot\alpha}_1 \bar{\partial}^{\dot\beta}_1  - \bar{\partial}^{\dot\alpha}_1 \bar{\partial}^{\dot\beta}_1 \bigg]C(0,y_2,\bar{y}_1,\bar{y}_2,I;\hat{K}_C|x)*\hat{k}_1 + \\
    +\frac{i \eta_1}{8}\bar{H}_{\dot\alpha \dot\beta}\delta^{nm} \delta^{pq} R(\hat{K}_C)^1_n R(\hat{K}_C)^1_p \bar{R}(\hat{K}_C)^w_m \bar{R}(\hat{K}_C)^z_q
    \bigg[\bar{y}^{\dot\alpha}_w\bar{y}^{\dot\beta}_z - 2 i \bar{y}^{\dot\alpha}_w \bar{\partial}^{\dot\beta}_z  - \bar{\partial}^{\dot\alpha}_w \bar{\partial}^{\dot\beta}_z \bigg]C(0,y_2,\bar{y}_1,\bar{y}_2,I;\hat{K}_C|x)*\hat{k}_1 + \\
    + \frac{i \eta_1}{4}\bar{H}_{\dot\alpha \dot\beta}\delta^{pq} R(\hat{K}_C)^1_p \bar{R}(\hat{K}_C)^l_q \bar{H}_{\dot\alpha \dot\beta}\bigg[\bar{y}^{\dot\alpha}_1\bar{y}^{\dot\beta}_l - i \bar{y}^{\dot\alpha}_1 \bar{\partial}^{\dot\beta}_l + i \bar{y}^{\dot\beta}_l \bar{\partial}^{\dot\alpha}_1 + \bar{\partial}^{\dot\alpha}_1 \bar{\partial}^{\dot\beta}_l \bigg]C(0,y_2,\bar{y}_1,\bar{y}_2,I;\hat{K}_C|x)*\hat{k}_1 +\\
    +\frac{i \eta_1}{8}\bar{H}_{\dot\alpha \dot\beta}\bigg[\bar{y}^{\dot\alpha}_2\bar{y}^{\dot\beta}_2 + 2 i \bar{y}^{\dot\alpha}_2 \bar{\partial}^{\dot\beta}_2  - \bar{\partial}^{\dot\alpha}_2 \bar{\partial}^{\dot\beta}_2 \bigg]C(y_1,0,\bar{y}_1,\bar{y}_2,I;\hat{K}_C|x)*\hat{k}_2 + \\
    +\frac{i \eta_1}{8}\bar{H}_{\dot\alpha \dot\beta}\delta^{nm} \delta^{pq} R(\hat{K}_C)^2_n R(\hat{K}_C)^2_p \bar{R}(\hat{K}_C)^w_m \bar{R}(\hat{K}_C)^z_q
    \bigg[\bar{y}^{\dot\alpha}_w\bar{y}^{\dot\beta}_z - 2 i \bar{y}^{\dot\alpha}_w \bar{\partial}^{\dot\beta}_z  - \bar{\partial}^{\dot\alpha}_w \bar{\partial}^{\dot\beta}_z \bigg]C(y_1,0,\bar{y}_1,\bar{y}_2,I;\hat{K}_C|x)*\hat{k}_2 + \\
    + \frac{i \eta_1}{4}\bar{H}_{\dot\alpha \dot\beta}\delta^{pq} R(\hat{K}_C)^2_p \bar{R}(\hat{K}_C)^l_q \bar{H}_{\dot\alpha \dot\beta}\bigg[\bar{y}^{\dot\alpha}_2\bar{y}^{\dot\beta}_l - i \bar{y}^{\dot\alpha}_2 \bar{\partial}^{\dot\beta}_l + i \bar{y}^{\dot\beta}_l \bar{\partial}^{\dot\alpha}_2 + \bar{\partial}^{\dot\alpha}_2 \bar{\partial}^{\dot\beta}_l \bigg]C(y_1,0,\bar{y}_1,\bar{y}_2,I;\hat{K}_C|x)*\hat{k}_2\,,
\end{multline}
\begin{multline}\label{e: B2 vert 2}
    \Upsilon_{tot}^{\eta_2}(\Omega_{AdS},\Omega_{AdS},C) = \frac{i \eta_2}{16}\bar{H}_{\dot\alpha \dot\beta}\bigg[\bar{y}^{\dot\alpha}_1\bar{y}^{\dot\beta}_1 + 2 i \bar{y}^{\dot\alpha}_1 \bar{\partial}^{\dot\beta}_1  - \bar{\partial}^{\dot\alpha}_1 \bar{\partial}^{\dot\beta}_1 + \bar{y}^{\dot\alpha}_2\bar{y}^{\dot\beta}_2 + 2 i \bar{y}^{\dot\alpha}_2 \bar{\partial}^{\dot\beta}_2  - \bar{\partial}^{\dot\alpha}_2 \bar{\partial}^{\dot\beta}_2  - 2\bar{y}^{\dot\alpha}_1\bar{y}^{\dot\beta}_2 -2 i \bar{y}^{\dot\alpha}_1 \bar{\partial}^{\dot\beta}_2 - \\
    -2 i \bar{y}^{\dot\alpha}_2 \bar{\partial}^{\dot\beta}_1 + 2 \bar{\partial}^{\dot\alpha}_1 \bar{\partial}^{\dot\beta}_2\bigg]C\bigg(\frac{1}{2}(y_1+y_2),\frac{1}{2}(y_1+y_2),\bar{y}_1,\bar{y}_2,I;\hat{K}_C|x\bigg)*\hat{k}_{12} +\\
    + \frac{i \eta_2}{16}\bar{H}_{\dot\alpha \dot\beta}\delta^{nm}\delta^{pq}\bigg(R(\hat{K}_C)^1_n R(\hat{K}_C)^1_p + R(\hat{K}_C)^2_n R(\hat{K}_C)^2_p - R(\hat{K}_C)^1_n R(\hat{K}_C)^2_p - R(\hat{K}_C)^2_n R(\hat{K}_C)^1_p\bigg)\\
    \bar{R}(\hat{K}_C)^w_m \bar{R}(\hat{K}_C)^z_q \bigg[\bar{y}^{\dot\alpha}_w\bar{y}^{\dot\beta}_z - i \bar{y}^{\dot\alpha}_w \bar{\partial}^{\dot\beta}_z -i \bar{y}^{\dot\beta}_z \bar{\partial}^{\dot\alpha}_w - \bar{\partial}^{\dot\alpha}_w \bar{\partial}^{\dot\beta}_z \bigg]C\bigg(\frac{1}{2}(y_1+y_2),\frac{1}{2}(y_1+y_2),\bar{y}_1,\bar{y}_2,I;\hat{K}_C|x\bigg)*\hat{k}_{12} + \\
    +\frac{i \eta_2}{8}\bar{H}_{\dot\alpha \dot\beta} \delta^{pq}\bigg(R(\hat{K}_C)^1_p - R(\hat{K}_C)^2_p\bigg)\bar{R}(\hat{K}_C)^l_q \bigg[\bar{y}^{\dot\alpha}_1\bar{y}^{\dot\beta}_l - i \bar{y}^{\dot\alpha}_1 \bar{\partial}^{\dot\beta}_l + i \bar{y}^{\dot\beta}_l \bar{\partial}^{\dot\alpha}_1 + \bar{\partial}^{\dot\alpha}_1 \bar{\partial}^{\dot\beta}_l \bigg]\\
    C\bigg(\frac{1}{2}(y_1+y_2),\frac{1}{2}(y_1+y_2),\bar{y}_1,\bar{y}_2,I;\hat{K}_C|x\bigg)*\hat{k}_{12}    + \frac{i \eta_2}{8}\bar{H}_{\dot\alpha \dot\beta} \delta^{pq}\bigg(R(\hat{K}_C)^2_p - R(\hat{K}_C)^1_p\bigg)\bar{R}(\hat{K}_C)^l_q \bigg[\bar{y}^{\dot\alpha}_2\bar{y}^{\dot\beta}_l - i \bar{y}^{\dot\alpha}_2 \bar{\partial}^{\dot\beta}_l \\
    + i \bar{y}^{\dot\beta}_l \bar{\partial}^{\dot\alpha}_2 + \bar{\partial}^{\dot\alpha}_2 \bar{\partial}^{\dot\beta}_l \bigg]C\bigg(\frac{1}{2}(y_1+y_2),\frac{1}{2}(y_1+y_2),\bar{y}_1,\bar{y}_2,I;\hat{K}_C|x\bigg)*\hat{k}_{12} +
\end{multline}
\begin{multline*}
    + \frac{i \eta_2}{16}\bar{H}_{\dot\alpha \dot\beta}\bigg[\bar{y}^{\dot\alpha}_1\bar{y}^{\dot\beta}_1 + 2 i \bar{y}^{\dot\alpha}_1 \bar{\partial}^{\dot\beta}_1  - \bar{\partial}^{\dot\alpha}_1 \bar{\partial}^{\dot\beta}_1 + \bar{y}^{\dot\alpha}_2\bar{y}^{\dot\beta}_2 + 2 i \bar{y}^{\dot\alpha}_2 \bar{\partial}^{\dot\beta}_2  - \bar{\partial}^{\dot\alpha}_2 \bar{\partial}^{\dot\beta}_2  + 2\bar{y}^{\dot\alpha}_1\bar{y}^{\dot\beta}_2 +2 i \bar{y}^{\dot\alpha}_1 \bar{\partial}^{\dot\beta}_2 + \\
    +2 i \bar{y}^{\dot\alpha}_2 \bar{\partial}^{\dot\beta}_1 - 2 \bar{\partial}^{\dot\alpha}_1 \bar{\partial}^{\dot\beta}_2\bigg]C\bigg(\frac{1}{2}(y_1-y_2),-\frac{1}{2}(y_1-y_2),\bar{y}_1,\bar{y}_2,I;\hat{K}_C|x\bigg)*\hat{k}^+_{12} +\\
    + \frac{i \eta_2}{16}\bar{H}_{\dot\alpha \dot\beta}\delta^{nm}\delta^{pq}\bigg(R(\hat{K}_C)^1_n R(\hat{K}_C)^1_p + R(\hat{K}_C)^2_n R(\hat{K}_C)^2_p + R(\hat{K}_C)^1_n R(\hat{K}_C)^2_p + R(\hat{K}_C)^2_n R(\hat{K}_C)^1_p\bigg)\\
    \bar{R}(\hat{K}_C)^w_m \bar{R}(\hat{K}_C)^z_q \bigg[\bar{y}^{\dot\alpha}_w\bar{y}^{\dot\beta}_z - i \bar{y}^{\dot\alpha}_w \bar{\partial}^{\dot\beta}_z -i \bar{y}^{\dot\beta}_z \bar{\partial}^{\dot\alpha}_w - \bar{\partial}^{\dot\alpha}_w \bar{\partial}^{\dot\beta}_z \bigg]C\bigg(\frac{1}{2}(y_1-y_2),-\frac{1}{2}(y_1-y_2),\bar{y}_1,\bar{y}_2,I;\hat{K}_C|x\bigg)*\hat{k}^+_{12} + \\
    +\frac{i \eta_2}{8}\bar{H}_{\dot\alpha \dot\beta} \delta^{pq}\bigg(R(\hat{K}_C)^1_p + R(\hat{K}_C)^2_p\bigg)\bar{R}(\hat{K}_C)^l_q \bigg[\bar{y}^{\dot\alpha}_1\bar{y}^{\dot\beta}_l - i \bar{y}^{\dot\alpha}_1 \bar{\partial}^{\dot\beta}_l + i \bar{y}^{\dot\beta}_l \bar{\partial}^{\dot\alpha}_1 + \bar{\partial}^{\dot\alpha}_1 \bar{\partial}^{\dot\beta}_l \bigg]\\
    C\bigg(\frac{1}{2}(y_1-y_2),-\frac{1}{2}(y_1-y_2),\bar{y}_1,\bar{y}_2,I;\hat{K}_C|x\bigg)*\hat{k}^+_{12}    + \frac{i \eta_2}{8}\bar{H}_{\dot\alpha \dot\beta} \delta^{pq}\bigg(R(\hat{K}_C)^2_p + R(\hat{K}_C)^1_p\bigg)\bar{R}(\hat{K}_C)^l_q \bigg[\bar{y}^{\dot\alpha}_2\bar{y}^{\dot\beta}_l - i \bar{y}^{\dot\alpha}_2 \bar{\partial}^{\dot\beta}_l \\
    + i \bar{y}^{\dot\beta}_l \bar{\partial}^{\dot\alpha}_2 + \bar{\partial}^{\dot\alpha}_2 \bar{\partial}^{\dot\beta}_l \bigg]C\bigg(\frac{1}{2}(y_1-y_2),-\frac{1}{2}(y_1-y_2),\bar{y}_1,\bar{y}_2,I;\hat{K}_C|x\bigg)*\hat{k}^+_{12}\,.
\end{multline*}
Even without specifying the explicit structure of the operator $\hat{K}_C$, it is clear from the vertices (\ref{e: B2 vert 1}) and (\ref{e: B2 vert 2}) that the differential operators can be partially transformed into  operators with respect to variables collinear to the root vectors. To further simplify the form of vertices one has to sort through all possible combination of Klein operators $\hat{k}_v\,, \hat{\bar{k}}_v$. However, as it stated in Section \ref{Derivative and modules}, there is no need to consider all $64$ possible combinations since only the products of the reflection matrices $R(\hat{k}_C)\bar R(\hat{\bar{k}}_C)^T$ (\ref{e: RbarR1})-(\ref{e: RbarR2}) matter.

\subsubsection{$R(\hat{k}_C)\bar R(\hat{\bar{k}}_C)^T = \begin{pmatrix}
1 & 0 \\
0 & 1
\end{pmatrix}$}

\begin{equation}
    \Upsilon_{tot}^{\eta_1}(\Omega_{AdS},\Omega_{AdS},C) = \frac{i \eta_1}{2}\bar{H}_{\dot\alpha \dot\beta}\bar{y}^{\dot\alpha}_1\bar{y}^{\dot\beta}_1 C(0,y_2,\bar{y}_1,\bar{y}_2,I;\hat{K}_C|x)*\hat{k}_1 + \frac{i \eta_1}{2}\bar{H}_{\dot\alpha \dot\beta}\bar{y}^{\dot\alpha}_2\bar{y}^{\dot\beta}_2 C(y_1,0,\bar{y}_1,\bar{y}_2,I;\hat{K}_C|x)*\hat{k}_2\,,
\end{equation}
\begin{multline}
    \Upsilon_{tot}^{\eta_2}(\Omega_{AdS},\Omega_{AdS},C) = \frac{i \eta_2}{4}\bar{H}_{\dot\alpha \dot\beta}(\bar{y}_1 - \bar{y}_2)^{\dot\alpha} (\bar{y}_1 - \bar{y}_2)^{\dot\beta} C\bigg(\frac{1}{2}(y_1+y_2),\frac{1}{2}(y_1+y_2),\bar{y}_1,\bar{y}_2,I;\hat{K}_C|x\bigg)*\hat{k}_{12} + \\
    + \frac{i \eta_2}{4}\bar{H}_{\dot\alpha \dot\beta} (\bar{y}_1 + \bar{y}_2)^{\dot\alpha} (\bar{y}_1 + \bar{y}_2)^{\dot\beta} C\bigg(\frac{1}{2}(y_1-y_2),-\frac{1}{2}(y_1-y_2),\bar{y}_1,\bar{y}_2,I;\hat{K}_C|x\bigg)*\hat{k}^+_{12}\,.
\end{multline}

\subsubsection{$R(\hat{k}_C)\bar R(\hat{\bar{k}}_C)^T = \begin{pmatrix}
-1 & 0 \\
0 & 1
\end{pmatrix}$}\label{Vertex TA1}

\begin{equation}
    \Upsilon_{tot}^{\eta_1}(\Omega_{AdS},\Omega_{AdS},C) = -\frac{i \eta_1}{2}\bar{H}_{\dot\alpha \dot\beta}\bar{\partial}^{\dot\alpha}_1\bar{\partial}^{\dot\beta}_1 C(0,y_2,\bar{y}_1,\bar{y}_2,I;\hat{K}_C|x)*\hat{k}_1 + \frac{i \eta_1}{2}\bar{H}_{\dot\alpha \dot\beta}\bar{y}^{\dot\alpha}_2\bar{y}^{\dot\beta}_2 C(y_1,0,\bar{y}_1,\bar{y}_2,I;\hat{K}_C|x)*\hat{k}_2\,,
\end{equation}
\begin{multline}
    \Upsilon_{tot}^{\eta_2}(\Omega_{AdS},\Omega_{AdS},C) = \frac{i \eta_2}{4}\bar{H}_{\dot\alpha \dot\beta}(\bar{y}_2 - i\bar{\partial}_1)^{\dot\alpha} (\bar{y}_2 - i\bar{\partial}_1)^{\dot\beta} C\bigg(\frac{1}{2}(y_1+y_2),\frac{1}{2}(y_1+y_2),\bar{y}_1,\bar{y}_2,I;\hat{K}_C|x\bigg)*\hat{k}_{12} + \\
    + \frac{i \eta_2}{4}\bar{H}_{\dot\alpha \dot\beta} (\bar{y}_2 + i\bar{\partial}_1)^{\dot\alpha} (\bar{y}_2 + i\bar{\partial}_1)^{\dot\beta} C\bigg(\frac{1}{2}(y_1-y_2),-\frac{1}{2}(y_1-y_2),\bar{y}_1,\bar{y}_2,I;\hat{K}_C|x\bigg)*\hat{k}^+_{12}\,.
\end{multline}

\subsubsection{$R(\hat{k}_C)\bar R(\hat{\bar{k}}_C)^T = \begin{pmatrix}
1 & 0 \\
0 & -1
\end{pmatrix}$}\label{Vertex TA2}

\begin{equation}
    \Upsilon_{tot}^{\eta_1}(\Omega_{AdS},\Omega_{AdS},C) = \frac{i \eta_1}{2}\bar{H}_{\dot\alpha \dot\beta}\bar{y}^{\dot\alpha}_1\bar{y}^{\dot\beta}_1 C(0,y_2,\bar{y}_1,\bar{y}_2,I;\hat{K}_C|x)*\hat{k}_1 - \frac{i \eta_1}{2}\bar{H}_{\dot\alpha \dot\beta}\bar{\partial}^{\dot\alpha}_2\bar{\partial}^{\dot\beta}_2 C(y_1,0,\bar{y}_1,\bar{y}_2,I;\hat{K}_C|x)*\hat{k}_2\,,
\end{equation}
\begin{multline}
    \Upsilon_{tot}^{\eta_2}(\Omega_{AdS},\Omega_{AdS},C) = \frac{i \eta_2}{4}\bar{H}_{\dot\alpha \dot\beta}(\bar{y}_1 - i\bar{\partial}_2)^{\dot\alpha} (\bar{y}_1 - i\bar{\partial}_2)^{\dot\beta} C\bigg(\frac{1}{2}(y_1+y_2),\frac{1}{2}(y_1+y_2),\bar{y}_1,\bar{y}_2,I;\hat{K}_C|x\bigg)*\hat{k}_{12} + \\
    + \frac{i \eta_2}{4}\bar{H}_{\dot\alpha \dot\beta} (\bar{y}_1 + i\bar{\partial}_2)^{\dot\alpha} (\bar{y}_1 + i\bar{\partial}_2)^{\dot\beta} C\bigg(\frac{1}{2}(y_1-y_2),-\frac{1}{2}(y_1-y_2),\bar{y}_1,\bar{y}_2,I;\hat{K}_C|x\bigg)*\hat{k}^+_{12}\,.
\end{multline}

\subsubsection{$R(\hat{k}_C)\bar R(\hat{\bar{k}}_C)^T = \begin{pmatrix}
-1 & 0 \\
0 & -1
\end{pmatrix}$}

\begin{equation}
    \Upsilon_{tot}^{\eta_1}(\Omega_{AdS},\Omega_{AdS},C) = -\frac{i \eta_1}{2}\bar{H}_{\dot\alpha \dot\beta}\bar{\partial}^{\dot\alpha}_1\bar{\partial}^{\dot\beta}_1 C(0,y_2,\bar{y}_1,\bar{y}_2,I;\hat{K}_C|x)*\hat{k}_1 - \frac{i \eta_1}{2}\bar{H}_{\dot\alpha \dot\beta}\bar{\partial}^{\dot\alpha}_2\bar{\partial}^{\dot\beta}_2 C(y_1,0,\bar{y}_1,\bar{y}_2,I;\hat{K}_C|x)*\hat{k}_2\,,
\end{equation}
\begin{multline}
    \Upsilon_{tot}^{\eta_2}(\Omega_{AdS},\Omega_{AdS},C) = - \frac{i \eta_2}{4}\bar{H}_{\dot\alpha \dot\beta}(\bar{\partial}_1- \bar{\partial}_2)^{\dot\alpha} (\bar{\partial}_1 - \bar{\partial}_2)^{\dot\beta} C\bigg(\frac{1}{2}(y_1+y_2),\frac{1}{2}(y_1+y_2),\bar{y}_1,\bar{y}_2,I;\hat{K}_C|x\bigg)*\hat{k}_{12} - \\
    - \frac{i \eta_2}{4}\bar{H}_{\dot\alpha \dot\beta} (\bar{\partial}_1 + \bar{\partial}_2)^{\dot\alpha} (\bar{\partial}_1 + \bar{\partial}_2)^{\dot\beta} C\bigg(\frac{1}{2}(y_1-y_2),-\frac{1}{2}(y_1-y_2),\bar{y}_1,\bar{y}_2,I;\hat{K}_C|x\bigg)*\hat{k}^+_{12}\,.
\end{multline}

\subsubsection{$R(\hat{k}_C)\bar R(\hat{\bar{k}}_C)^T = \begin{pmatrix}
0 & 1 \\
1 & 0
\end{pmatrix}$}\label{Vertex TA3}

\begin{multline}
    \Upsilon_{tot}^{\eta_1}(\Omega_{AdS},\Omega_{AdS},C) = \frac{i \eta_1}{8}\bar{H}_{\dot\alpha \dot\beta}(\bar{y}_1+\bar{y}_2 + i\bar{\partial}_1 - i\bar{\partial}_2)^{\dot\alpha}(\bar{y}_1+\bar{y}_2 + i\bar{\partial}_1 - i\bar{\partial}_2)^{\dot\beta}  C(0,y_2,\bar{y}_1,\bar{y}_2,I;\hat{K}_C|x)*\hat{k}_1 +\\
    +\frac{i \eta_1}{8}\bar{H}_{\dot\alpha \dot\beta}(\bar{y}_1+\bar{y}_2 - i\bar{\partial}_1 + i\bar{\partial}_2)^{\dot\alpha}(\bar{y}_1+\bar{y}_2 - i\bar{\partial}_1 + i\bar{\partial}_2)^{\dot\beta} C(y_1,0,\bar{y}_1,\bar{y}_2,I;\hat{K}_C|x)*\hat{k}_2\,,
\end{multline}
\begin{multline}
    \Upsilon_{tot}^{\eta_2}(\Omega_{AdS},\Omega_{AdS},C) = - \frac{i \eta_2}{4}\bar{H}_{\dot\alpha \dot\beta} (\bar{\partial}_1- \bar{\partial}_2)^{\dot\alpha} (\bar{\partial}_1 - \bar{\partial}_2)^{\dot\beta} C\bigg(\frac{1}{2}(y_1+y_2),\frac{1}{2}(y_1+y_2),\bar{y}_1,\bar{y}_2,I;\hat{K}_C|x\bigg)*\hat{k}_{12} + \\
    + \frac{i \eta_2}{4}\bar{H}_{\dot\alpha \dot\beta} (\bar{y}_1 + \bar{y}_2)^{\dot\alpha} (\bar{y}_1 + \bar{y}_2)^{\dot\beta} C\bigg(\frac{1}{2}(y_1-y_2),-\frac{1}{2}(y_1-y_2),\bar{y}_1,\bar{y}_2,I;\hat{K}_C|x\bigg)*\hat{k}^+_{12}\,.
\end{multline}

\subsubsection{$R(\hat{k}_C)\bar R(\hat{\bar{k}}_C)^T = \begin{pmatrix}
0 & -1 \\
-1 & 0
\end{pmatrix}$}\label{Vertex TA4}

\begin{multline}
    \Upsilon_{tot}^{\eta_1}(\Omega_{AdS},\Omega_{AdS},C) = \frac{i \eta_1}{8}\bar{H}_{\dot\alpha \dot\beta}(\bar{y}_1-\bar{y}_2 +i\bar{\partial}_1 + i\bar{\partial}_2)^{\dot\alpha}(\bar{y}_1-\bar{y}_2 +i\bar{\partial}_1 + i\bar{\partial}_2)^{\dot\beta}  C(0,y_2,\bar{y}_1,\bar{y}_2,I;\hat{K}_C|x)*\hat{k}_1 + \\
    +\frac{i \eta_1}{8}\bar{H}_{\dot\alpha \dot\beta}(\bar{y}_1-\bar{y}_2 - i \bar{\partial}_1 - i \bar{\partial}_2)^{\dot\alpha}(\bar{y}_1-\bar{y}_2 - i \bar{\partial}_1 - i \bar{\partial}_2)^{\dot\beta} C(y_1,0,\bar{y}_1,\bar{y}_2,I;\hat{K}_C|x)*\hat{k}_2\,,
\end{multline}
\begin{multline}
    \Upsilon_{tot}^{\eta_2}(\Omega_{AdS},\Omega_{AdS},C) =  \frac{i \eta_2}{4}\bar{H}_{\dot\alpha \dot\beta} (\bar{y}_1- \bar{y}_2)^{\dot\alpha} (\bar{y}_1 - \bar{y}_2)^{\dot\beta} C\bigg(\frac{1}{2}(y_1+y_2),\frac{1}{2}(y_1+y_2),\bar{y}_1,\bar{y}_2,I;\hat{K}_C|x\bigg)*\hat{k}_{12} - \\
    - \frac{i \eta_2}{4}\bar{H}_{\dot\alpha \dot\beta} (\bar{\partial}_1 + \bar{\partial}_2)^{\dot\alpha} (\bar{\partial}_1 + \bar{\partial}_2)^{\dot\beta} C\bigg(\frac{1}{2}(y_1-y_2),-\frac{1}{2}(y_1-y_2),\bar{y}_1,\bar{y}_2,I;\hat{K}_C|x\bigg)*\hat{k}^+_{12}\,.
\end{multline}

\subsubsection{$R(\hat{k}_C)\bar R(\hat{\bar{k}}_C)^T = \begin{pmatrix}
0 & -1 \\
1 & 0
\end{pmatrix}$}

\begin{multline}
    \Upsilon_{tot}^{\eta_1}(\Omega_{AdS},\Omega_{AdS},C) = \frac{i \eta_1}{8}\bar{H}_{\dot\alpha \dot\beta}(\bar{y}_1-\bar{y}_2 + i\bar{\partial}_1 +i\bar{\partial}_2)^{\dot\alpha}(\bar{y}_1-\bar{y}_2 + i\bar{\partial}_1 +i\bar{\partial}_2)^{\dot\beta}C(0,y_2,\bar{y}_1,\bar{y}_2,I;\hat{K}_C|x)*\hat{k}_1 + \\
    + \frac{i \eta_1}{8}\bar{H}_{\dot\alpha \dot\beta}(\bar{y}_1+\bar{y}_2 - i\bar{\partial}_1 + i\bar{\partial}_2)^{\dot\alpha}(\bar{y}_1+\bar{y}_2 - i\bar{\partial}_1 + i\bar{\partial}_2)^{\dot\beta}  C(y_1,0,\bar{y}_1,\bar{y}_2,I;\hat{K}_C|x)*\hat{k}_2\,,
\end{multline}
\begin{multline}
    \Upsilon_{tot}^{\eta_2}(\Omega_{AdS},\Omega_{AdS},C) = \frac{i \eta_2}{4}\bar{H}_{\dot\alpha \dot\beta}(\bar{y}_2 - i\bar{\partial}_1)^{\dot\alpha} (\bar{y}_2 - i\bar{\partial}_1)^{\dot\beta} C\bigg(\frac{1}{2}(y_1+y_2),\frac{1}{2}(y_1+y_2),\bar{y}_1,\bar{y}_2,I;\hat{K}_C|x\bigg)*\hat{k}_{12} + \\
    + \frac{i \eta_2}{4}\bar{H}_{\dot\alpha \dot\beta} (\bar{y}_1 + i\bar{\partial}_2)^{\dot\alpha} (\bar{y}_1 + i\bar{\partial}_2)^{\dot\beta} C\bigg(\frac{1}{2}(y_1-y_2),-\frac{1}{2}(y_1-y_2),\bar{y}_1,\bar{y}_2,I;\hat{K}_C|x\bigg)*\hat{k}^+_{12}\,.
\end{multline}
\subsubsection{$R(\hat{k}_C)\bar R(\hat{\bar{k}}_C)^T = \begin{pmatrix}
0 & 1 \\
-1 & 0
\end{pmatrix}$}

\begin{multline}
    \Upsilon_{tot}^{\eta_1}(\Omega_{AdS},\Omega_{AdS},C) = \frac{i \eta_1}{8}\bar{H}_{\dot\alpha \dot\beta}(\bar{y}_1+\bar{y}_2 + i\bar{\partial}_1 - i\bar{\partial}_2)^{\dot\alpha}(\bar{y}_1+\bar{y}_2 + i\bar{\partial}_1 - i\bar{\partial}_2)^{\dot\beta} C(0,y_2,\bar{y}_1,\bar{y}_2,I;\hat{K}_C|x)*\hat{k}_1 + \\
    + \frac{i \eta_1}{8}\bar{H}_{\dot\alpha \dot\beta}(\bar{y}_1-\bar{y}_2 - i\bar{\partial}_1 - i\bar{\partial}_2)^{\dot\alpha}(\bar{y}_1-\bar{y}_2 - i\bar{\partial}_1 - i\bar{\partial}_2)^{\dot\beta} C(y_1,0,\bar{y}_1,\bar{y}_2,I;\hat{K}_C|x)*\hat{k}_2\,,
\end{multline}
\begin{multline}
    \Upsilon_{tot}^{\eta_2}(\Omega_{AdS},\Omega_{AdS},C) = \frac{i \eta_2}{4}\bar{H}_{\dot\alpha \dot\beta}(\bar{y}_1 - i\bar{\partial}_2)^{\dot\alpha} (\bar{y}_1 - i\bar{\partial}_2)^{\dot\beta} C\bigg(\frac{1}{2}(y_1+y_2),\frac{1}{2}(y_1+y_2),\bar{y}_1,\bar{y}_2,I;\hat{K}_C|x\bigg)*\hat{k}_{12} + \\
    + \frac{i \eta_2}{4}\bar{H}_{\dot\alpha \dot\beta} (\bar{y}_2 + i\bar{\partial}_1)^{\dot\alpha} (\bar{y}_2 + i\bar{\partial}_1)^{\dot\beta} C\bigg(\frac{1}{2}(y_1-y_2),-\frac{1}{2}(y_1-y_2),\bar{y}_1,\bar{y}_2,I;\hat{K}_C|x\bigg)*\hat{k}^+_{12}\,.
\end{multline}

One can observe that in all cases except for $R(\hat{k}_C)\bar R(\hat{\bar{k}}_C)^T = \pm 1$ the total holomorphic vertices contain standard terms such as $\bar{H}^{\dot\alpha\dot\beta}\bar\partial_{i\dot\alpha} \bar\partial_{i\dot\beta}$ and $\bar{H}_{\dot\alpha\dot\beta}\bar{y}^{\dot\alpha}_i \bar{y}^{\dot\beta}_i$ supplemented by a new type of terms such as $\bar{H}_{\dot\alpha\dot\beta}(\bar{y}_2 + i\bar{\partial}_1)^{\dot\alpha} (\bar{y}_2 + i\bar{\partial}_1)^{\dot\beta}$ that mix $\bar{y}$ with $\bar{\partial}$. These new terms glue entangled modules \ref{ModuleM} (present in all CHS models other than $\mathds{Z}_2$) to the remaining $B_2$ modules $\{M_{tw\otimes tw}\,,M_{adj \otimes adj}\,, M_{tw \otimes adj}\,,M_{adj \otimes tw}\}$.

Note that pairs of vertices (\ref{Vertex TA1}; \ref{Vertex TA2}) and (\ref{Vertex TA3}; \ref{Vertex TA4}) are connected by the change of variables automorphism of the star product algebra (\ref{e:change of variables y}) that swaps conjugacy classes $\mathcal{R}_1$ and $\mathcal{R}_2$. This automorphism relates vertices ($\Upsilon_{tot}^{\eta_1}\,,\Upsilon_{tot}^{\eta_2}$) of (\ref{Vertex TA1}; \ref{Vertex TA2}) and ($\Upsilon_{tot}^{\eta_2}\,,\Upsilon_{tot}^{\eta_1}$) of (\ref{Vertex TA3}; \ref{Vertex TA4}).

Restriction to the invariant subspace of total dressed Klein operator involutive automorphism $\hat{K}_v \rightarrow -\hat{K}_v$, which eliminates the non-unitary and non-highest-weight modules from the zero-form sector and preserves modules $\{M_{tw\otimes adj}\,,M_{adj\otimes tw}\}$ that have unitary submodules, leaves us with vertices \ref{Vertex TA1}-\ref{Vertex TA2} and \ref{Vertex TA3}-\ref{Vertex TA4}. One can see that vertices \ref{Vertex TA1}-\ref{Vertex TA2} and \ref{Vertex TA3}-\ref{Vertex TA4} contain standard terms that glue zero-form modules $\{M_{tw\otimes adj}\,,M_{adj\otimes tw}\}$ to one-form modules $\{M_{adj\otimes adj}\,,M_{tw\otimes tw}\}$, and new terms that glue zero-form modules $\{M_{tw\otimes adj}\,,M_{adj\otimes tw}\}$ to one-form entangled modules. Considering the one-form $M_{adj \otimes adj}$ sector we observe that gluing is carried out by $\bar{H} \bar{\partial} \bar{\partial}$ terms and, therefore, the First On-Shell Theorem has an expected form. This sector should contain a number of copies of the standard Fronsdal HS equations and fields. Other one-form sectors have not been previously observed and their physical interpretation is not yet fully clear.

\section{Dynamical content}\label{dynamics}

In this Section we go over all linear equations that remain after the $\hat{K}_v \rightarrow -\hat{K}_v$ truncation coupled with the boundary condition (\ref{e: boundary condition}) and discuss their dynamical content. While this truncation may not be the only possible one, it nonetheless provides a natural starting point as an obvious generalization of that of the standard HS system.

As explained in Section \ref{Derivative and modules} and \ref{FOST}, Klein-related truncation leaves us with the one-form modules \ref{ModuleTT}, \ref{ModuleAA}, \ref{ModuleM} glued to the zero-form modules \ref{ModuleTA1} and \ref{ModuleTA2}. However, zero-form modules $\{M_{tw\otimes adj}\,,M_{adj\otimes tw}\}$ should be further subjected to the boundary condition (\ref{e: boundary condition}), otherwise they are not complex equivalent to unitary modules. The boundary condition effectively eliminates dependence of the zero-form fields $C$ on the spinor oscillators responsible for the description of the adjoint factor. Thus, linearized zero-form equation reduces to the standard twisted-adjoint one
\begin{equation}
    \bigg(D_L - i e^{\alpha \dot \alpha}(y_{\alpha i} \bar y_{\dot\alpha i} - \partial_{\alpha i}\bar\partial_{\dot \alpha i})\bigg)C(Y_i,I;\hat{K}_C|x) = 0\,,
\end{equation}
where $i$ is either $\{1,2\}$ or $\{+,-\}$. Hence, the fields $C(Y_i,I;\hat{K}_C|x)$ encode the Weyl tensors and their descendants. Since idempotents $I_n$ induce filtration and decompose the CHS system into the corresponding sectors, we observe that there are $2$ zero-form fields in each $I_n$ sector and $32$ zero-form fields in the $I_1 I_2$ sector.

Now we turn to the one-form equations. In general it should be noted that since the only remaining $C$ fields describe Weyl tensors and their descendants, which at the linear level completely define the dynamics, the $\omega$ fields glued to them have to consist of a combination of Fronsdal fields and, may be, some topological fields, that carry no local degrees of freedom. Indeed, in $d=4$ massless mixed symmetry fields  do not exist
(carry no degrees of freedom). Therefore Fronsdal fields are the only propagating massless fields free of ghosts. However, an $AdS_4$ algebra admits non-unitary partially massless fields \cite{Deser:1983tm}-\cite{Khabarov:2019dvi} not present in a standard HS theory due to the insufficient number of oscillator copies. It is anticipated that the $(d=4\,, B_2)$ CHS model not truncated to its unitary subsector should contain partially massless fields since the doubling of oscillator variables allows one to encode $sp(4)$ two-row Young diagrams.

 More in detail, let us first consider the one-form field $\omega$ that takes values in the tensor product of two adjoint modules $M_{adj \otimes adj}$, which arises, for example, when  $\omega$  contains no Klein dependencies. Collecting the terms from the previous section, the one-form equation after the $\hat{K}_v \rightarrow -\hat{K}_v$ truncation is
\begin{multline} \label{e: adj-adj omega}
    \left[D_L + e^{\alpha \dot\alpha} \sum_{i=1}^{2}(\bar{y}_{\dot\alpha i} \partial_{\alpha i} + y_{\alpha i} \bar{\partial}_{\dot\alpha i})\right] \omega(y_1, y_2, \bar{y}_1, \bar{y}_2,I | x) = \\ = -\frac{i \eta_1}{2}\bar{H}_{\dot\alpha \dot\beta}\bar{\partial}^{\dot\alpha}_1\bar{\partial}^{\dot\beta}_1 C(0,y_2,\bar{y}_1,\bar{y}_2,I;\hat{k}_1 | x) * \hat{k}_1 - \frac{i \eta_1}{2}\bar{H}_{\dot\alpha \dot\beta}\bar{\partial}^{\dot\alpha}_2\bar{\partial}^{\dot\beta}_2 C(y_1,0,\bar{y}_1,\bar{y}_2,I; \hat{k}_2 |x)*\hat{k}_2 - \\ -\frac{i \eta_2}{2}\bar{H}_{\dot\alpha \dot\beta} \bar{\partial}_-^{\dot\alpha} \bar{\partial}_-^{\dot\beta} C\bigg(y_+,0,\bar{y}_+,\bar{y}_-,I;\hat{k}_{12}|x\bigg)*\hat{k}_{12}
    - \frac{i \eta_2}{2}\bar{H}_{\dot\alpha \dot\beta} \bar{\partial}_+^{\dot\alpha} \bar{\partial}_+^{\dot\beta} C\bigg(0,y_-,\bar{y}_+,\bar{y}_-,I;\hat{k}^+_{12}|x\bigg)*\hat{k}^+_{12} + \text{c.c}\,.
\end{multline}

We see that the structure of this equation is reminiscent of the standard coupling between the $\omega$ field in the adjoint sector and the $C$ field in the twisted sector. Here, however, the $C$ fields belong to the tensor product of the adjoint and twisted-adjoint modules, but with the imposed boundary condition (\ref{e: boundary condition}) leaving only the twisted-adjoint factor the analogy becomes clear:
\begin{multline}\label{e: equation adj-adj truncated}
    \left[D_L + e^{\alpha \dot\alpha} \sum_{i=1}^{2}(\bar{y}_{\dot\alpha i} \partial_{\alpha i} + y_{\alpha i} \bar{\partial}_{\dot\alpha i})\right] \omega(y_1, y_2, \bar{y}_1, \bar{y}_2,I | x) = \\ = -\frac{i \eta_1}{2}\bar{H}_{\dot\alpha \dot\beta}\bar{\partial}^{\dot\alpha}_1\bar{\partial}^{\dot\beta}_1 C(0,\bar{y}_1,I;\hat{k}_1 | x) * \hat{k}_1 - \frac{i \eta_1}{2}\bar{H}_{\dot\alpha \dot\beta}\bar{\partial}^{\dot\alpha}_2\bar{\partial}^{\dot\beta}_2 C(0,\bar{y}_2,I; \hat{k}_2 |x)*\hat{k}_2 - \\ -\frac{i \eta_2}{2}\bar{H}_{\dot\alpha \dot\beta} \bar{\partial}_-^{\dot\alpha} \bar{\partial}_-^{\dot\beta} C\bigg(0,\bar{y}_-,I;\hat{k}_{12}|x\bigg)*\hat{k}_{12}
    - \frac{i \eta_2}{2}\bar{H}_{\dot\alpha \dot\beta} \bar{\partial}_+^{\dot\alpha} \bar{\partial}_+^{\dot\beta} C\bigg(0,\bar{y}_+,I;\hat{k}^+_{12}|x\bigg)*\hat{k}^+_{12} + \text{c.c.}\,.
\end{multline}
The one-form module $M_{adj\otimes adj}$ glues to the set of Weyl modules and according to the standard HS theory encodes several copies of dynamical Fronsdal fields and equations. Indeed, at the linear order we can identify the following component one-forms in $\omega$
\begin{equation} \label{e:linear omega decomposition}
    \omega(y_1, y_2, \bar{y}_1, \bar{y}_2,I | x) = \omega_1(y_1, \bar{y}_1,I | x) + \omega_2(y_2, \bar{y}_2,I | x) + \omega_+(y_+, \bar{y}_+,I | x) + \omega_-(y_-, \bar{y}_-,I | x) +\dots\,,
\end{equation}
where the remaining $(\dots)$ terms are glued to zero-forms excluded by the truncation procedure.
For example, after such a decomposition,
\begin{equation}
    \left[D_L + e^{\alpha \dot\alpha} (\bar{y}_{\dot\alpha 1} \partial_{\alpha 1} + y_{\alpha 1} \bar{\partial}_{\dot\alpha 1})\right] \omega_1(y_1,  \bar{y}_1, I | x) = -\frac{i \eta_1}{2}\bar{H}_{\dot\alpha \dot\beta}\bar{\partial}^{\dot\alpha}_1\bar{\partial}^{\dot\beta}_1 C(0,\bar{y}_1,I;\hat{k}_1 | x) * \hat{k}_1 + \text{c.c.}\,
\end{equation}
reproduces the linear equation of the standard HS theory.

As clarified in Section \ref{Coxeter higher-spin equations}, all oscillator variables $Y^A_n$ implicitly carry a corresponding idempotent $I_n$ and constant terms not multiplied by an idempotent are also not present. Therefore, idempotents induce a filtration that decomposes the full CHS system into sectors that are independent at the linear level but interact in a triangle-like manner in the higher orders of the perturbation theory. Indeed, for simplicity consider a case of $B_2$ group (the same decomposition occurs in a general $B_p$ model). All fields decompose into the three sectors: $F(Y_1;\hat{K}_1|x)*I_1$, $F(Y_2;\hat{K}_2|x)*I_2$ and $F(Y_1,Y_2;\hat{K}|x)*I_1 I_2$, where $F$ is either $\omega$ or $C$. Due to the presence of idempotents in a star product (\ref{e:star product}), fields from sectors $I_2$ and $I_1 I_2$ do not contribute to the sector $I_1$, and $I_1$ and $I_1 I_2$ give no contribution to the sector $I_2$. However, the product of fields from sectors $I_1$ and $I_2$ belongs to the sector $I_1 I_2$. Therefore, interaction vertices decompose into the components along $I_n$ and $I_1 I_2$. The components along $I_n$ are built out of the fields from the corresponding $I_n$ sectors and coincide with the vertices of the standard HS theory, where variables $Y^A$ are replaced by $Y^A_n$. The vertices proportional to $I_1 I_2$ are built from the fields of all sectors and, consequently, differ from the standard ones.

Let us look at the module $M_{adj\otimes adj}$ encoded by (\ref{e: equation adj-adj truncated}) from the filtration perspective. In the $I_n$ sector, we arrive at a singular adjoint module $\omega(Y_n;\hat{K}_n|x)*I_n$ from the standard theory coupled with the twisted Weyl module $C(Y_n;\hat{K}_n|x)*I_n$. The $B_2$ CHS theory features two complete copies of the standard HS theory associated with their own set of spinor variables $Y^A_n$ that exist in the sectors $I_n$. Although the $I_1 I_2$ sector contains the same equations at the linear level, it differs significantly in the full non-linear system. While we have determined the dynamical primary fields and equations embedded into the equation for the one-form $M_{adj \otimes adj}$ glued to the zero-forms $\{M_{tw\otimes adj}\,,M_{adj\otimes tw}\}$ restricted by (\ref{e: boundary condition}), there can be non-dynamical primary fields and equations, \ie one-form fields outside of (\ref{e:linear omega decomposition}) decomposition, the gluing zero-form terms for which get eliminated in the truncation procedure. Such non-dynamical fields can be important since non-zero VEVs of topological fields could serve as a mass parameters. Thus, a $\sigma_-$ cohomological analysis for the case of $M_{adj \otimes adj}$ is needed.

Consider the case of one-form $\omega$ valued in the product of two twisted-adjoint modules $M_{tw\otimes tw}$. For example, the field $\omega(Y_1,Y_2,I;\hat{k}_1 \hat{k}_2| x)$ takes value in $M_{tw\otimes tw}$. Then the equation after Klein truncation is
\begin{multline} \label{e: tw-tw omega}
    \left[D_L - i e^{\alpha \dot\alpha}\sum_{i = 1}^2 (y_{\alpha i}\bar y_{\dot \alpha i} - \partial_{\alpha i} \bar \partial_{\dot\alpha i})\right] \omega(y_1, y_2, \bar{y}_1, \bar{y}_2,I; \hat{k}_1 \hat{k}_2| x) = \\ =\frac{i \eta_1}{2}\bar{H}_{\dot\alpha \dot\beta}\bar{y}^{\dot\alpha}_2\bar{y}^{\dot\beta}_2 C(y_1,0,\bar{y}_1,\bar{y}_2,I;\hat{k}_1|x)*\hat{k}_2 + \frac{i \eta_1}{2}\bar{H}_{\dot\alpha \dot\beta}\bar{y}^{\dot\alpha}_1\bar{y}^{\dot\beta}_1 C(0,y_2,\bar{y}_1,\bar{y}_2,I;\hat{k}_2|x)*\hat{k}_1 + \\ +\frac{i \eta_2}{2}\bar{H}_{\dot\alpha \dot\beta} \bar{y}_+^{\dot\alpha} \bar{y}_+^{\dot\beta} C\bigg(0,y_-,\bar{y}_+,\bar{y}_-,I;\hat{k}_{12}|x\bigg)*\hat{k}^+_{12} + \frac{i \eta_2}{2}\bar{H}_{\dot\alpha \dot\beta} \bar{y}_-^{\dot\alpha} \bar{y}_-^{\dot\beta} C\bigg(y_+,0,\bar{y}_+,\bar{y}_-,I;\hat{k}^+_{12}|x\bigg)*\hat{k}_{12}\,.
\end{multline}
Imposing boundary condition we arrive at
\begin{multline}
    \left[D_L - i e^{\alpha \dot\alpha}\sum_{i = 1}^2 (y_{\alpha i}\bar y_{\dot \alpha i} - \partial_{\alpha i} \bar \partial_{\dot\alpha i})\right] \omega(y_1, y_2, \bar{y}_1, \bar{y}_2,I; \hat{k}_1 \hat{k}_2| x) = \\ =\frac{i \eta_1}{2}\bar{H}_{\dot\alpha \dot\beta}\bar{y}^{\dot\alpha}_2\bar{y}^{\dot\beta}_2 C(y_1,\bar{y}_1,I;\hat{k}_1|x)*\hat{k}_2 + \frac{i \eta_1}{2}\bar{H}_{\dot\alpha \dot\beta}\bar{y}^{\dot\alpha}_1\bar{y}^{\dot\beta}_1 C(y_2,\bar{y}_2,I;\hat{k}_2|x)*\hat{k}_1 + \\ +\frac{i \eta_2}{2}\bar{H}_{\dot\alpha \dot\beta} \bar{y}_+^{\dot\alpha} \bar{y}_+^{\dot\beta} C\bigg(y_-,\bar{y}_-,I;\hat{k}_{12}|x\bigg)*\hat{k}^+_{12} + \frac{i \eta_2}{2}\bar{H}_{\dot\alpha \dot\beta} \bar{y}_-^{\dot\alpha} \bar{y}_-^{\dot\beta} C\bigg(y_+,\bar{y}_+,I;\hat{k}^+_{12}|x\bigg)*\hat{k}_{12}\,.
\end{multline}

This equation shows that Weyl modules are glued to the one-form modules $M_{tw \otimes tw}$, implying the latter are some (most likely non-local) combinations of Fronsdal fields, though their explicit appearance is not yet clear and will be considered elsewhere. The \rhs of equations (\ref{e: adj-adj omega}), (\ref{e: tw-tw omega}) involve not only primary zero-forms, but also their descendants, the fact that has to be taken into account in the $\sigma_-$ cohomological analysis of the independent equations on the one-forms. The general case of the product of two twisted-adjoint modules has been done in \cite{Gelfond:2013lba} where the symmetry properties of the primary fields and equations were considered.

Furthermore, as we have seen in Section \ref{Derivative and modules}, the covariant constancy equation for entangled modules can be transformed into the equation for $M_{tw \otimes tw}$ by the exponential ansatz. Therefore, it can be conjectured that the primary fields and equations in that case can be described in terms of the same Young diagrams as in \cite{Gelfond:2013lba} for $M_{tw\otimes tw}$ albeit after an appropriate resummation and change of variables. Since the exponent is not a graded object, the result is likely to have no compact finite form in terms of $Y_n^A$.

As  mentioned in the beginning of this section, not all possible types of fields can be realized in $d=4$. While this restriction is obviously lifted in higher dimensions, which serves as a motivation for studying CHS theories in $AdS_d$, it is known that 
higher-rank fields in lower dimensions can effectively exhibit behavior of rank one fields in higher dimensions, as was demonstrated in \cite{Gelfond:2003vh, Sorokin:2017irs}.\footnote{In this context, rank means the tensor degree of the fields of the original theory. In terms of the Coxeter extension, this can be understood as the tensor degree of the moduli of the standard HS theory. In the case of the multiparticle extension, this is the tensor degree of the moduli of the theory being extended.} The application of this mechanism to CHS theories or their further multiparticle extensions is an interesting topic for the future research.

\section{Conclusion}\label{Conclusion}

In this paper we have analyzed a Coxeter extension of the standard $4d$ HS theory at the linear order.

It was shown that an $AdS_4$ solution is embedded into the general model with the symmetry $(\mathcal{C}\times \mathcal{C})/\mathcal{J}$, $\mathcal{J} = Span\{I_n - \bar{I}_n\}$, \ie a CHS theory where the holomorphic and anti-holomorphic idempotents are identified, and it is unique. For this embedding a covariant derivative has been constructed for an arbitrary Coxeter group $\mathcal{C}$. We have observed that a new type of modules, that are not isomorphic to the tensor product of standard adjoint and twisted-adjoint modules and referred to as entangled, appears. A necessary and sufficient condition for the module to be entangled has been found.

In case of the $B_2$ Coxeter group a full set of covariant constancy equations and related modules have been determined. All modules are grouped into four categories where three out of four correspond to the tensor product of standard HS modules while the remaining group corresponds to entangled modules. All $B_2$ linear equations have been reformulated in terms of the field-theoretical Fock modules and unitarizability of $B_2$ modules has been analyzed through the identification with $su(2, 2)$ modules induced via a Bogolyubov transform. It has been deduced that entangled $B_2$ modules are not complex equivalent to lowest-weight unitary modules and, therefore, should be eliminated from the zero-form sector of the theory, while they still play an important role in the total system as they remain in the one-form sector. The entangled modules arise due to a mixing of oscillators of different types induced by the action of the Coxeter group that leads to expressions $P_{\pm}^{kl}= \frac{1}{2}\delta^{nm}\bigg(\mathbb{1}^k_n \bar{\mathbb{1}}^l_m \pm R(k){}^k_n \bar{R}(\bar{k}){}^l_m \bigg)$ no longer being orthogonal projectors onto twisted and adjoint terms of the covariant derivative. An increase in the rank of the group in the $B_p$ series leads to the appearance of other types of entangled modules such as linked transpositions $\hat{k}_{ij}\hat{k}_{jl}$ and others combinations of transpositions and basis axis reflections. Their classification and physical meaning beyond $B_2$ is yet to be studied.

In the $B_2$ case, one can truncate to lowest-weight modules which have unitary submodules from the full nonlinear system in a consistent manner by the total involutive automorphism $\hat{K}_v \rightarrow -\hat{K}_v$ leaving  modules that correspond to the product of standard $4d$ HS adjoint and twisted-adjoint modules intact. In those modules the residual formal restriction on the arguments of the zero-form $C(Y_1,Y_2,I;\hat{K}|x)$ field, resulting from the conditions on their asymptotic behaviour at the $AdS_4$ boundary (\ref{e: boundary condition}), further constrain the set of fields, narrowing in down to the twisted module of the standard HS theory, describing the physical single-particle states at the linear order. This is indeed the desired result, as it preserves the interpretation of single-particle states, in particular maintaining a consistent description of gravity within the theory. The restriction to the unitary submodules in the full nonlinear system is yet to be studied but it is anticipated to be consistent since sources, that decrease at infinity cannot induce fields that increase at infinity. A similar intertwining of dynamical and topological sectors can be seen in the HS theory in three dimensions, where it can be successfully resolved by picking very specific families of shifts in the homotopy procedure in all steps of the perturbative expansion \cite{Korybut:2022kdx}.

A generalization of the First On-Shell Theorem has been presented for the case of a general Coxeter group. For this purpose, the shifted homotopy technique was extended to  CHS theory while the extension of the differential homotopy of \cite{Vasiliev:2023yzx} is an interesting problem for the future. In the $B_2$ case all possible linear vertices have been presented. Among these one finds the expected generalizations of the vertices of the standard $4d$ HS system, gluing one-forms $\omega$ from the adjoint sector $M_{adj \otimes adj}$ to dynamical $C$ fields from $\{M_{tw\otimes adj}\,, M_{adj \otimes tw}\}$, which after imposing boundary condition have non-trivial dependencies only in the twisted sector. The resulting equations reproduce the standard First On-Shell theorem and describe multiple copies of Weyl tensors, Fronsdal fields and field equations. New vertices involving one-forms from $M_{tw\otimes tw}$ and the entangled modules are also obtained gluing Weyl modules to some combinations of Fronsdal fields. The exact form of these combinations and the spectrum of primary fields provide a starting point for the further research.

\section*{Acknowledgement}

The authors thank Konstantin Alkalaev, Anatoliy Korybut, Nikita Misuna and especially Olga Gelfond and the referee for helpful comments. MV is grateful for hospitality to Ofer Aharony,
Theoretical High Energy Physics Group of Weizmann Institute of Science where some
part of this work has been done.
The work was supported by the Foundation for the Advancement of Theoretical
Physics and Mathematics “BASIS”.



\begin{thebibliography}{99}
\parindent=0pt
\parskip=0pt

		
\bibitem{Vasiliev:1990en}
M.~A.~Vasiliev,
Phys. Lett. B \textbf{243} (1990), 378-382

\bibitem{Vasiliev:1992av}
M.~A.~Vasiliev,
Phys. Lett. B \textbf{285} (1992), 225-234

\bibitem{Fradkin:1986qy}
E.~S.~Fradkin and M.~A.~Vasiliev,
Nucl. Phys. B \textbf{291} (1987), 141-171

\bibitem{Vasiliev:1988sa}
M.~A.~Vasiliev,
Annals Phys. \textbf{190} (1989), 59-106

\bibitem{Gross:1987ar}
D.~J.~Gross and P.~F.~Mende,
Nucl. Phys. B \textbf{303} (1988), 407-454

\bibitem{Gross:1988ue}
D.~J.~Gross,
Phys. Rev. Lett. \textbf{60} (1988), 1229

\bibitem{Bianchi:2003wx}
M.~Bianchi, J.~F.~Morales and H.~Samtleben,
JHEP \textbf{07} (2003), 062
[arXiv:hep-th/0305052 [hep-th]].

\bibitem{Beisert:2004di}
N.~Beisert, M.~Bianchi, J.~F.~Morales and H.~Samtleben,
JHEP \textbf{07} (2004), 058
[arXiv:hep-th/0405057 [hep-th]].

\bibitem{Bianchi:2004ww}
M.~Bianchi,
Comptes Rendus Physique \textbf{5} (2004), 1091-1099
[arXiv:hep-th/0409292 [hep-th]].

\bibitem{Bianchi:2004npm}
M.~Bianchi and V.~Didenko,
[arXiv:hep-th/0502220 [hep-th]].

\bibitem{Lindstrom:2003mg}
U.~Lindstrom and M.~Zabzine,
Phys. Lett. B \textbf{584} (2004), 178-185
[arXiv:hep-th/0305098 [hep-th]].

\bibitem{Bonelli:2003kh}
G.~Bonelli,
Nucl. Phys. B \textbf{669} (2003), 159-172
[arXiv:hep-th/0305155 [hep-th]].

\bibitem{Sagnotti:2003qa}
A.~Sagnotti and M.~Tsulaia,
Nucl. Phys. B \textbf{682} (2004), 83-116
[arXiv:hep-th/0311257 [hep-th]].

\bibitem{Vasiliev:2018zer}
M.~A.~Vasiliev,
JHEP \textbf{08} (2018), 051
[arXiv:1804.06520 [hep-th]].

\bibitem{Bourbaki}
N.~Bourbaki,
Elements of Mathematics, Lie Groups and Lie Algebras, Chapters 4-6,
Springer-Verlag Berlin Heidelber, New York, 2002.

\bibitem{Cherednik:1992sy}
I.~Cherednik,
RIMS-885.

\bibitem{Brink:1993sz}
L.~Brink, T.~H.~Hansson, S.~Konstein and M.~A.~Vasiliev,
Nucl. Phys. B \textbf{401} (1993), 591-612
[arXiv:hep-th/9302023 [hep-th]].

\bibitem{Vasiliev:2012tv}
M.~A.~Vasiliev,
Class. Quant. Grav. \textbf{30} (2013), 104006
[arXiv:1212.6071 [hep-th]].

\bibitem{Engquist:2005yt}
J.~Engquist and P.~Sundell,
Nucl. Phys. B \textbf{752} (2006), 206-279
[arXiv:hep-th/0508124 [hep-th]].

\bibitem{Engquist:2007pr}
J.~Engquist, P.~Sundell and L.~Tamassia,
JHEP \textbf{02} (2007), 097
[arXiv:hep-th/0701051 [hep-th]].

\bibitem{Gaberdiel:2015mra}
M.~R.~Gaberdiel and R.~Gopakumar,
J. Phys. A \textbf{48} (2015) no.18, 185402
[arXiv:1501.07236 [hep-th]].

\bibitem{Gaberdiel:2015wpo}
M.~R.~Gaberdiel and R.~Gopakumar,
JHEP \textbf{09} (2016), 085
[arXiv:1512.07237 [hep-th]].

\bibitem{Didenko:2018fgx}
V.~E.~Didenko, O.~A.~Gelfond, A.~V.~Korybut and M.~A.~Vasiliev,
J. Phys. A \textbf{51} (2018) no.46, 465202
[arXiv:1807.00001 [hep-th]].


\bibitem{Vasiliev:2001zy}
M.~A.~Vasiliev,
Phys. Rev. D \textbf{66} (2002), 066006
[arXiv:hep-th/0106149 [hep-th]].

\bibitem{Bolotin:1999fa}
K.~I.~Bolotin and M.~A.~Vasiliev,
Phys. Lett. B \textbf{479} (2000), 421-428
[arXiv:hep-th/0001031 [hep-th]].

\bibitem{Didenko:2003aa}
V.~E.~Didenko and M.~A.~Vasiliev,
J. Math. Phys. \textbf{45} (2004), 197-215
[arXiv:hep-th/0301054 [hep-th]].

\bibitem{Gunaydin:1984fk}
M.~Gunaydin and N.~Marcus,
Class. Quant. Grav. \textbf{2} (1985), L11

\bibitem{Gunaydin:1984vz}
M.~Gunaydin and N.~Marcus,
Class. Quant. Grav. \textbf{2} (1985), L19

\bibitem{Iazeolla:2008ix}
C.~Iazeolla and P.~Sundell,
JHEP \textbf{10} (2008), 022
[arXiv:0806.1942 [hep-th]].

\bibitem{Basile:2018dzi}
T.~Basile, X.~Bekaert and E.~Joung,
JHEP \textbf{07} (2018), 009
[arXiv:1802.03232 [hep-th]].

\bibitem{Korybut:2022kdx}
A.~V.~Korybut, A.~A.~Sevostyanova, M.~A.~Vasiliev and V.~A.~Vereitin,
Phys. Lett. B \textbf{838} (2023), 137718
[arXiv:2211.15778 [hep-th]].

\bibitem{Gelfond:2013lba}
O.~A.~Gelfond and M.~A.~Vasiliev,
JHEP \textbf{10} (2016), 067

\bibitem{Gelfond:2003vh}
O.~A.~Gelfond and M.~A.~Vasiliev,
Theor. Math. Phys. \textbf{145} (2005), 1400-1424
[arXiv:hep-th/0304020 [hep-th]].

\bibitem{Sorokin:2017irs}
D.~Sorokin and M.~Tsulaia,
Universe \textbf{4} (2018) no.1, 7
[arXiv:1710.08244 [hep-th]].

\bibitem{Vasiliev:2023yzx}
M.~A.~Vasiliev,
JHEP \textbf{11} (2023), 048
[arXiv:2307.09331 [hep-th]].

\bibitem{Vasiliev:1999ba}
M.~A.~Vasiliev,
[arXiv:hep-th/9910096 [hep-th]].

\bibitem{Deser:1983tm}
S.~Deser and R.~I.~Nepomechie,
Phys. Lett. B \textbf{132} (1983), 321-324

\bibitem{Deser:1983mm}
S.~Deser and R.~I.~Nepomechie,
Annals Phys. \textbf{154} (1984), 396

\bibitem{Brink:2000ag}
L.~Brink, R.~R.~Metsaev and M.~A.~Vasiliev,
Nucl. Phys. B \textbf{586} (2000), 183-205
[arXiv:hep-th/0005136 [hep-th]].

\bibitem{Deser:2001us}
S.~Deser and A.~Waldron,
Nucl. Phys. B \textbf{607} (2001), 577-604
[arXiv:hep-th/0103198 [hep-th]].

\bibitem{Zinoviev:2001dt}
Y.~M.~Zinoviev,
[arXiv:hep-th/0108192 [hep-th]].

\bibitem{Dolan:2001ih}
L.~Dolan, C.~R.~Nappi and E.~Witten,
JHEP \textbf{10} (2001), 016
[arXiv:hep-th/0109096 [hep-th]].

\bibitem{Zinoviev:2002ye}
Y.~M.~Zinoviev,
[arXiv:hep-th/0211233 [hep-th]].

\bibitem{Skvortsov:2006at}
E.~D.~Skvortsov and M.~A.~Vasiliev,
Nucl. Phys. B \textbf{756} (2006), 117-147
[arXiv:hep-th/0601095 [hep-th]].

\bibitem{Buchbinder:2006ge}
I.~L.~Buchbinder, V.~A.~Krykhtin and P.~M.~Lavrov,
Nucl. Phys. B \textbf{762} (2007), 344-376
[arXiv:hep-th/0608005 [hep-th]].

\bibitem{Zinoviev:2008ze}
Y.~M.~Zinoviev,
Nucl. Phys. B \textbf{808} (2009), 185-204
[arXiv:0808.1778 [hep-th]].

\bibitem{Buchbinder:2019olk}
I.~L.~Buchbinder, M.~V.~Khabarov, T.~V.~Snegirev and Y.~M.~Zinoviev,
JHEP \textbf{08} (2019), 116
[arXiv:1904.01959 [hep-th]].

\bibitem{Khabarov:2019dvi}
M.~V.~Khabarov and Y.~M.~Zinoviev,
Nucl. Phys. B \textbf{948} (2019), 114773
[arXiv:1906.03438 [hep-th]].

\end{thebibliography}
\end{document}